\documentclass[11pt,oneside,english]{amsart}
\usepackage[T1]{fontenc}
\usepackage[latin9]{inputenc}
\usepackage{geometry}
\usepackage{textcomp}
\usepackage{amsbsy}
\usepackage{amstext}
\usepackage{amsthm}
\usepackage{amssymb}

\makeatletter

\newcommand{\noun}[1]{\textsc{#1}}

\numberwithin{equation}{section}
\numberwithin{figure}{section}
\theoremstyle{plain}
\newtheorem{thm}{\protect\theoremname}[section]
  \theoremstyle{plain}
  \newtheorem{cor}[thm]{\protect\corollaryname}
  \theoremstyle{remark}
  \newtheorem{rem}[thm]{\protect\remarkname}
  \theoremstyle{definition}
  \newtheorem{defn}[thm]{\protect\definitionname}
  \theoremstyle{plain}
  \newtheorem{lem}[thm]{\protect\lemmaname}
  \theoremstyle{plain}
  \newtheorem{prop}[thm]{\protect\propositionname}
  \theoremstyle{definition}
  \newtheorem{example}[thm]{\protect\examplename}
  \theoremstyle{plain}
  \newtheorem{conjecture}[thm]{\protect\conjecturename}

\def\makebbb#1{
    \expandafter\gdef\csname#1\endcsname{
        \ensuremath{\Bbb{#1}}}
}\makebbb{R}\makebbb{N}\makebbb{Z}\makebbb{C}\makebbb{H}\makebbb{E}\makebbb{H}\makebbb{P}\makebbb{B}\makebbb{Q}\makebbb{E}

\usepackage{babel}

\usepackage{babel}

\makeatother

\usepackage{babel}

\makeatother

\usepackage{babel}
  \providecommand{\conjecturename}{Conjecture}
  \providecommand{\corollaryname}{Corollary}
  \providecommand{\definitionname}{Definition}
  \providecommand{\examplename}{Example}
  \providecommand{\lemmaname}{Lemma}
  \providecommand{\propositionname}{Proposition}
  \providecommand{\remarkname}{Remark}
\providecommand{\theoremname}{Theorem}

\begin{document}

\title{Sharp deviation inequalities for the 2D Coulomb gas and Quantum hall
states, I}

\author{Robert J. Berman}
\begin{abstract}
We establish sharp deviation inequalities for the linear statistics
of the 2D Coulomb gas at inverse temperature $\beta\leq1.$ These
imply sub-Gaussian inequalities, where the variance is given by the
Dirichlet $H^{1}-$norm. The proofs use complex geometry and potential
theory on Riemann surfaces and apply more generally to $\beta-$ensembles,
which also include integer Quantum Hall states on Riemann surfaces.
In a sequel of the paper we give applications to large and moderate
deviation principles, local laws at mesoscopic scales, quantitative
Bergman kernel asymptotics. In a series of companion papers applications
to concentration of measure, Monte-Carlo methods for numerical integration
and random matrices are given and relations to Kähler geometry are
explored.
\end{abstract}

\address{Chalmers University of Technology and University of Gothenburg}

\email{robertb@chalmers.se}
\maketitle

\section{Introduction}

We start by recalling the general setup in statistical mechanics describing
the equilibrium state, at inverse temperature $\beta,$ of $N$ particles
interacting by a pair interaction and subject to an exterior potential.
Let $X$ be a Riemannian manifold and denote by $dV$ the corresponding
volume form. Given a symmetric lower semi-continuous (lsc) function
$W(x,y)$ on $X^{2}$ (the \emph{pair interaction potential}) and
a lower semi-continuous $V(x)$ (the \emph{exterior potential}) the
corresponding $N-$particle Hamiltonian is the function on $X^{N}$
defined by 
\begin{equation}
H^{(N)}(x_{1},...,x_{N}):=\frac{1}{2}\sum_{i\neq j\leq N}W(x_{i},x_{j})+\sum_{i\leq1}V(x_{i})\label{eq:Hamilt for W and V intro}
\end{equation}
 and the corresponding Gibbs measure at inverse temperature $\beta\in]0,\infty[$
is the following symmetric probability measure 
\[
d\mathbb{P}_{N,\beta}:=e^{-\beta H^{(N)}}dV^{\otimes N}/Z_{N,\beta}
\]
on $X^{N},$ where the normalizing constant 
\[
Z_{N,\beta}:=\int_{X^{N}}e^{-\beta H^{(N)}}dV^{\otimes N}
\]
is called the\emph{ ($N-$particle) partition function }and is assumed
to be finite.\emph{ }

The ensemble $(X^{N},d\mathbb{P}_{N,\beta})$ defines a random point
process with $N$ particles on $X.$ The corresponding \emph{empirical
measure} is the random measure 
\begin{equation}
\delta_{N}:\,\,X^{N}\rightarrow\mathcal{P}(X),\,\,\,(x_{1},\ldots,x_{N})\mapsto\delta_{N}(x_{1},\ldots,x_{N}):=\frac{1}{N}\sum_{i=1}^{N}\delta_{x_{i}}\label{eq:emp measure intro}
\end{equation}
 taking values in the space $\mathcal{P}(X)$ of all probability measures
on $X$. Given a bounded continuous function $u\in C_{b}(X)$ the
associated (normalized) \emph{linear statistic }is the real-valued
random variable 
\begin{equation}
U_{N}(x_{1},...,x_{N}):=\left\langle \delta_{N},u\right\rangle =\frac{1}{N}\sum_{i=1}^{N}u(x_{i}).\label{eq:def of linear stat}
\end{equation}
Assuming that the random measure $\delta_{N}$ converges in probability
to a deterministic measure $\mu_{eq}$, in the many particle limit
$N\rightarrow\infty$ (the ``equilibrium measure'') a classical
problem is to study the deviations around $\mu_{eq}.$ More precisely,
assuming that 
\[
U_{N}\rightarrow\bar{u}:=\int_{X}u\mu_{eq},\,\,\,N\rightarrow\infty
\]
in probability, the problem is to quantify - on a logarithmic scale
- the probability $\mathbb{P}\left(\left|U_{N}-\bar{u}\right|>\delta\right)$
that $U_{N}$ deviates by more than $\delta\in\R_{+}$ from its deterministic
limit $\bar{u}.$ More precisely, a \emph{large deviation upper bound
}at a speed $R_{N}\rightarrow\infty$ and with a rate functional $I_{u}(\delta)$
holds if 
\begin{equation}
\mathbb{P}_{N,\beta}\left(\left|U_{N}-\bar{u}\right|>\delta\right)\leq e^{-R_{N}(I_{u}(\delta)+o(1))}\label{eq:LDU intro}
\end{equation}
as $N\rightarrow\infty.$ We will be particularly interested in the
case when $I_{u}(\delta)$ is quadratic in $\delta,$ i.e. $I_{u}(\delta)=\delta^{2}/2\sigma_{u}^{2}$
for a positive constant $\sigma_{u}^{2}$ called the \emph{variance
proxy.} Then it follows from Chebishev's inequality that \ref{eq:LDU intro}
is implied by the following stronger \emph{sub-Gaussian bound }for
the moment generating function of $U_{N}-\bar{u}$
\begin{equation}
\mathbb{E}_{N,\beta}\left(e^{t(U_{N}-u)}\right)\leq e^{R_{N}(\frac{1}{2}\sigma_{u}^{2}t^{2}+o(1))}\label{eq:sub-gaussian bound general notation}
\end{equation}

In what follows we will often simplify the notation by leaving out
the subscript $(N,\beta)$ in the notation for the probability measure
$\mathbb{P}_{N,\beta}$ and the corresponding expectation $\mathbb{E}_{N,\beta}(\cdot):=\int(\cdot)d\mathbb{P}_{N,\beta}.$

\subsection{\label{subsec:The-2D-Coulomb}The 2D Coulomb gas and Quantum Hall
states on Riemann surfaces}

The main aim of the present work is to establish sharp quantitative
non-asymptotic versions of the large deviation bound \ref{eq:LDU intro}
and the sub-Gaussian bound \ref{eq:sub-gaussian bound general notation}
for the \emph{2D Coulomb gas }of unit charge particles (aka the One
Component Plasma or Jellium in the physics literature) and Quantum
Hall states. An important feature of the sub-Gaussian deviation inequalities
is their universality; they are not sensitive to the precise form
of the exterior potential, up to a lower-order error term. But it
should be stressed that the inequalities in question are also new
in the case of the complex Ginibre ensemble, where very precise error
estimates can be obtained (as will be detailed in a separate publication
\cite{berm16}). Applications of the inequalities shown in the present
paper are given in the sequel \cite{berm15}, including large deviation
principles for singular data, moderate deviation principles, local
laws at mesoscopic scales and quantitative Bergman kernel asymptotics.
Various elaborations are also given in the companion papers \cite{berm15b,berm16,berm17}
(see Section \ref{sec:Main-results-in} for precise statements). 

\subsubsection{The 2D Coulomb gas }

We recall that the 2D Coulomb gas is the $N-$particle random point
process on Euclidean $\R^{2},$ identified with $\C,$ is defined
by a repulsive logarithmic pair interaction
\[
W(z,w):=-\log|z-w|^{2}
\]
In other words, fixing an exterior potential $V,$ the corresponding
$N-$particle Hamiltonian is given by 
\begin{equation}
H^{(N)}(z_{1},...,z_{N}):=-\frac{1}{2}\sum_{i\neq j\leq n}\log|z_{i}-z_{j}|^{2}+\sum_{i\leq N}V(z_{i})\label{eq:Hamiltonian in plane intro}
\end{equation}
In order to get a well-defined many particle limit (of mean field
type) it is usually assumed that the potential $V$ depends on $N$
in the following way: 
\[
V=N\phi,
\]
 where the function $\phi$ has sufficient growth at infinity to ensure
that the corresponding partition functions 
\[
Z_{N,\beta}[N\phi]:=\int_{\C^{N}}\prod_{i<j\leq N}|z_{i}-z_{j}|^{2\beta}\prod_{i\leq N}e^{-N\beta\phi(z_{i})}d\lambda^{\otimes N}
\]
 are non-zero (cf. \cite{l-s,l-s-2,b-b-n-y2}). More precisely, it
is usually assumed that $\phi$ has \emph{strictly super logarithmic
growth, }in the following sense: there exist strictly positive constants
$\epsilon$ and $C$ such that 
\begin{equation}
\phi(z)\geq(1+\epsilon)\log((1+|z|^{2})-C,\label{eq:str super log growth in intro}
\end{equation}
Then $Z_{N,\beta}$ is finite for any $N$ and $\beta>0.$ The corresponding
Gibbs measure may be expressed as

\begin{equation}
d\mathbb{P}_{N_{k},\beta}=\frac{1}{Z_{N,\beta}}|D^{(N)}|^{2\beta}\left(e^{-N\beta\phi}d\lambda\right)^{\otimes N},\,\,\,D^{(N)}:=\det(\Psi_{i}(x_{j}))_{i,j\leq N},\label{eq:Gibbs measure as det intro}
\end{equation}
 where $\Psi_{1},...,\Psi_{N}$ is the standard monomial base in the
space $H_{N}$ of all holomorphic polynomials of degree at most $N-1$
on $\R^{2},$ identified with $\C.$ The particular case $\beta=1$
is singled out by the fact that the corresponding random point process
is \emph{determinantal }\cite{h-k-p}. If moreover $\phi(z)=|z|^{2}$
then the corresponding Coulomb gas ensemble is the \emph{complex Ginibre
ensemble} appearing as the law on the spaces of eigenvalues of random
complex matrices of rank $N$ with centered normally distributed entries
of variance $1/N$ \cite{fo}. The case $\beta=1$ with a general
confining potential $\phi$ also admits a random matrix interpretation
in terms of normal rank $N$ matrices \cite{c-z}. \footnote{A rescaling by two of $H^{(N)}$ is often used in the literature (for
example \cite{l-s-2}) and then $\beta$ coincides with the Dyson
index in the random matrix litterature so that $\beta=2$ in the determinantal
setting.} 

In the present paper it will also be important to allow $\epsilon=0$
in the growth condition \ref{eq:str super log growth in intro}, i.e.
to consider potentials $\phi$ with\emph{ super logarithmic growth
\begin{equation}
\phi(z)\geq\psi_{0}(z)-C,\,\,\,\psi_{0}(z):=\log(1+|z|^{2})\label{eq:super log intro}
\end{equation}
}Then we will, for a given inverse temperature $\beta,$ take
\[
V=(N+p)\phi,\,\,\,p:=2/\beta-1
\]
which ensures that the corresponding partition function $Z_{N,\beta}[V]$
is finite for any $N$ and $\beta>0.$ 

The present results will mainly be shown to hold under the assumption
that $\beta\leq1.$ In fact, the proofs reduce to the determinantal
case $\beta=1.$ Moreover, by a scaling argument the proofs reduce
to the setting where $V=(N+1)\phi$ and $\phi$ has logarithmic growth,
which turns out to be the most natural one from the complex geometric
point of view that we shall adopt. In fact, the proofs of the main
results in the present paper are inspired by a circle of ideas originating
in the study of Kähler metrics with constant scalar curvature on complex
manifolds (in particular, \cite{do,bbgz,bern}; see also the survey
\cite{De}). Accordingly, a privileged role will be played by the
potential $\psi_{0}$ in formula \ref{eq:super log intro}, which,
from the complex-geometric point of view, appears as the Kähler potential
of the constant curvature metric on the Riemann sphere. The corresponding
determinantal point process is usually called the \emph{spherical
ensemble} in the probability literature. \footnote{Interestingly, this ensemble can be realized as the eigenvalues of
$AB^{-1}$ where $A$ and $B$ are random Ginibre matrices \cite{kr}. }

\subsubsection{Quantum Hall states on a compact Riemann surface $X$}

The results for the Coulomb gas in the plane will appear as special
cases of the corresponding results for random point processes defined
by Quantum Hall states on a compact Riemann surface $X,$ which is
polarized, i.e. $X$ is endowed with a positive holomorphic line bundle
$L.$ Then the role of the functions $\Psi_{i}$ appearing in the
determinantal expression is played by a basis in the space of all
global holomorphic sections of an appropriate tensor power of $L$
(see Section \ref{subsec:Determinantal-point-process} for the precise
setup). We recall that from a physical point of view $\Psi_{i}$ represent
single particle states on $X$ and the corresponding Gibbs measure
represents the corresponding collective Quantum Hall state. When $\beta=1$
this is the ground state of a gas of free fermions (electrons) confined
to $X$ and subject to a magnetic field whose strength is proportional
to $|\Delta\phi|$. This is the multi-particle state which accounts
for the integer Quantum Hall effect (i.e. the quantization of the
Hall conductance observed in a two-dimensional electron gas a strong
magnetic field; see section Section \ref{subsec:Comp QH}).

In the case of a general polarized Riemann surface $(X,L)$ the role
of the spherical ensemble will be played by the canonical determinantal
point processes introduced in \cite{berm3}, which are induced from
the constant scalar curvature metric on $X.$ These ensembles turn
out to minimize the error terms which appear in the estimates in the
main results (when $\beta=1).$

\subsection{\label{subsec:Statement-of-the}Main results for the Coulomb gas
in the plane }

In this section we will, for the sake of concreteness, state the main
results in the case of the Coulomb gas in the plane (see Section \ref{subsec:Main-results-for beta ens on X}
for the corresponding results for general $\beta-$ensembles on compact
a Riemann surface $X).$ Given $\beta>0$ we set 
\[
p:=2/\beta-1
\]

\subsubsection{Assumptions on the exterior potential $V$ and the test function
$u$}

\emph{\noun{Main Assumption}}\emph{ on $V:$ we always assume (unless
otherwise specified})\emph{ that the exterior potentials $V$ is of
the form $V=N\phi$ for a function $\phi$ with strictly super logarithmic
growth (formula  \ref{eq:str super log growth in intro}) or $V=(N+p)\phi$
for a function $\phi$ with super logarithmic growth (formula \ref{eq:super log intro}).
Moreover, $\phi$ is assumed lsc} and $\{\phi<\infty\}$ is assumed
to not be polar. 

In the case when $\phi$ has strictly super-logarithmic growth it
then follows from classical potential theory \cite{s-t} that the
equilibrium measure $\mu_{\phi}$ associated to $\phi$ is well-defined
and has compact support: $\mu_{\phi}$ is defined as the unique minimizer
of the weighted logarithmic energy $E_{\phi},$ viewed as a functional
on the space $\mathcal{P}(\C)$ of probability measures on $\C,$ 

\[
E_{\phi}(\mu):=-\frac{1}{2}\int\log|z-w|^{2}\mu\otimes\mu+\int\phi\mu
\]

The corresponding\emph{ free energy} $\mathcal{F}(\phi)$ is defined
by 
\begin{equation}
\mathcal{F}(\phi):=\inf_{\mathcal{P}(\C)}E_{\phi}(\mu)\label{eq:def of free energ intro}
\end{equation}

In the more general case when $\phi$ merely has super logarithmic
growth $\mu_{\phi}$ is still well-defined (as shown in \cite{h-k};
a more general complex-geometric setting is considered in Section
\ref{subsec:Potential-theoretic-setup}, specialized to $\C$ in Section
\ref{subsec:Energies-in-}). However, in general, the support of $\mu_{\phi},$
that we shall denote by $S,$ is not compact. 

From place to place we will use the following additional regularity
assumptions.
\begin{itemize}
\item \textbf{(A0)} \emph{$\phi$ is continuous on the complement of a closed
polar subset. }
\item \textbf{(A1)}\emph{ A0 holds and there exists a constant $\lambda$
such that $\Delta\phi\leq\lambda e^{-\psi_{0}}$ on a neighborhood
of $S$ }
\end{itemize}
The assumption A1 implies that
\begin{equation}
\mu_{\phi}=\frac{1}{4\pi}1_{S}\Delta\phi d\lambda\label{eq:mu phi in terms of S intro}
\end{equation}
(but is should be stressed that we do not make any regularity assumptions
on the set $S,$ which can be extremely irregular, even if $\phi$
is smooth; see the discussion in \cite{berm15}).

\emph{\noun{Main assumption on $u:$ }}\emph{We always assume (unless
otherwise specified) that the test function $u$ is in the space $H^{1}(\C)$
of all $u\in L_{loc}^{2}(\C)$ such that $\nabla u\in L^{2}(\C,d\lambda)$
holds in the sense of distributions. The corresponding (Dirichlet)
$H^{1}-$norm will be normalized as follows: }
\[
\left\Vert u\right\Vert _{H^{1}}^{2}:=\frac{1}{4\pi}\int_{\C}|\nabla u|^{2}d\lambda
\]
Note that since $u\in L_{loc}^{2}(\C)$ and the Gibbs measure $d\mathbb{P}_{N,\beta}$
(formula \ref{eq:Gibbs measure as det intro}) is absolutely continuous
wrt $d\lambda^{\otimes N}$ the random variable $U_{N}$ (formula
\ref{eq:def of linear stat}) is well-defined. Moreover, as explained
in Section \ref{subsec:Energies-in-} $\mathcal{F}(\phi+u)$ is well-defined
if $\phi$ and $u$ satisfy the main assumptions above. In fact, by
the approximation result in Theorem \ref{thm:conv of free energy in C},
there is no loss of generality in assuming that $\phi$ is smooth
and $u\in C_{c}^{\infty}$ in the results stated below. 

\subsubsection{Motivation}

In order to motivate the main result first assume that $\phi$ is
a continuous function on $\C$ with strictly super logarithmic growth
and consider the corresponding $N-$particle Coulomb gas ensemble
with exterior potential $V=N\phi.$ It is then well-known that the
corresponding partition function $Z_{N}[N\phi]$ has the following
asymptotics as $N\rightarrow\infty:$
\[
\beta^{-1}N^{-2}\log Z_{N}[N\phi]=-\mathcal{F}(\phi)+o(1)
\]
(which holds, more generally, if $V=(N+o(N)\phi.$ As a consequence,
the moment generating function of a given linear statistic $U_{N}$
corresponding to $u\in C_{b}(\C)$ has the following asymptotics:
\[
N^{-2}\log\E(e^{-N^{2}U_{N}})=N^{-2}\log\frac{Z_{N}[N(\phi+u)]}{Z_{N}[N\phi]}=-\mathcal{F}(\phi+u)+\mathcal{F}(\phi)+o(1)
\]
 where $\E$ denotes expectations wrt the Gibbs measure \ref{eq:Gibbs measure as det intro}
on $\C^{N}$ induced by $\phi$ at inverse temperature $\beta.$ The
error term $o(1)$ above tends to zero as $N\rightarrow\infty,$ but
depends on $u.$ The core of the present work is an inequality which
yields a quantitative refinement of the upper bound in the previous
asymptotics. 

\subsubsection{The case when $V=(N+p)\phi$}

The inequality in question turns out to be cleanest in the case when
the exterior potential $\phi$ is of the form $V=(N+p)\phi,$ for
$\phi$ of super logarithmic growth:
\begin{thm}
\emph{\label{thm:sharp non Gaussian bd intro} Assume that $\beta\leq1.$
For any $N\geq1$ and $u\in H^{1}(\C)$
\[
\frac{1}{\beta}\frac{1}{N(N+p)}\log\E(e^{-\beta(N+p)NU_{N}})\leq-\mathcal{F}(\phi+u)+\mathcal{F}(\phi)+\epsilon_{N,\beta}[\phi]
\]
where }$\epsilon_{N,\beta}[\phi]$ is a sequence, which only depends
on $\phi$ (and $\beta).$
\end{thm}

When $\beta=1$ the error sequence $\epsilon_{N,\beta}[\phi]$ is
minimized and vanishes for the spherical ensemble, i.e. for $\phi=\psi_{0}$
(formula \ref{eq:super log intro}). It is explicitly given by

\begin{equation}
\epsilon_{N,1}[\phi]:=-\frac{1}{\beta N(N+1)}\log Z_{N,\beta}[(N+1)\phi]-\mathcal{F}(\phi)+\frac{1}{N(N+1)}\log Z_{N,1}[(N+1)\psi_{0}]+\mathcal{F}(\psi_{0}),\label{eq:def of error in intro}
\end{equation}
(see formula \ref{eq:def of error sequence for V N plus p text} for
the expression for general $\beta$). An important feature of the
error term $\epsilon_{N,\beta}[\phi]$ is that it only depends on
a lower bound on $Z_{N,\beta}[(N+p)\phi],$ which can be estimated
explicitly. In particular, we have the following bounds, depending
on the regularity assumptions on $\phi$ (see Section \ref{subsec:Estimates-of-the}):
\begin{itemize}
\item If A0 holds then $\epsilon_{N,\beta}[\phi]\rightarrow0$ as $N\rightarrow\infty$
\item If\emph{ }A1 holds then $\epsilon_{N,\beta}[\phi]\leq\frac{\log N}{2}N^{-1}+\beta^{-1}C_{\phi}N^{-1}$
for an explicit constant $C_{\phi}.$ 
\item If $S$ is a domain with piece-wise $C^{1}-$boundary, $\phi\in C^{5}(S)$
and $\Delta\phi>0$ on $C,$ then $\epsilon_{N,\beta}[\phi]\leq AN^{-1}$
for a (non-explicit) constant $A$ (by \cite{l-s,b-b-n-y2}; see Section
\ref{sec:Sharpness}) 
\end{itemize}
The previous theorem  implies the following sharp sub-Gaussian bound 
\begin{thm}
\emph{\label{thm:sharp Gauss bound on E intro}(sub-Gaussian inequality)
Assume that $\beta\leq1.$ For any $N\geq1$ and $u\in H^{1}(\C)$
}
\[
\E(e^{(N+p)N\beta t\left(U_{N}-\bar{u}\right)})\leq e^{N(N+p)\beta\left(\frac{t^{2}}{2}\left\Vert u\right\Vert _{H^{1}}^{2}+\epsilon_{N,\beta}[\phi]\right)}
\]
 
\end{thm}

This can be viewed as a Coulomb gas generalization of the sharp Moser-Trudinger
inequality on the two-sphere. Indeed, when $\phi=\psi_{0}$ and $\beta=1$
the inequality is equivalent, under stereographic projection, to the
$N-$particular generalization of the Moser-Trudinger inequality for
the Coulomb gas on the two-sphere established in \cite{berm3}. In
this case the corresponding equilibrium measure $\mu_{\phi}$ identifies
with the uniform measure on the two-sphere and the error $\epsilon_{N,1}[\phi]$
vanishes, as discussed above. In general, as explained in \cite{berm15},
the previous theorem may be interpreted as a sub-Gaussian property
wrt the Laplacian of the Gaussian free field. 
\begin{cor}
\label{cor:Gaussian dev ineq}\emph{Assume that $\beta\leq1.$ For
any $N\geq1$ and $u\in H^{1}(\C)$ }
\[
\P\left(\left|U_{N}-\bar{u}\right|>\delta\right)\leq2e^{-\beta N(N+1)\left(\frac{\delta^{2}}{2}\left\Vert u\right\Vert _{H^{1}}^{-2}+\epsilon_{N,\beta}\right)}
\]
\end{cor}

\subsubsection{The case when $V=N\phi$}

We next turn to the case when the exterior potential is given by $V=N\phi$
for $\phi$ with strictly super logarithmic growth. This amounts to
considering the previous setting with $\phi$ replaced by

\[
\phi_{N}=\frac{N\phi}{N+p}.
\]
Applying Theorem \ref{thm:sharp non Gaussian bd intro} to $\phi_{N}$
we then arrive at the following inequality formulated in terms of
the rescaled error sequence 
\begin{equation}
\tilde{\epsilon}_{N,\beta}[\phi]:=(1+N^{-1}p)\epsilon_{N,\beta}[\phi_{N}]\label{eq:def of error tilde in intro}
\end{equation}
(which satisfies estimates analogous to the estimates on $\epsilon_{N,\beta}[\phi];$
see Section \ref{subsec:Estimates-of-the}).
\begin{thm}
\emph{\label{thm:sharp non-Gauss ine for strong}Assume that $\beta\leq1.$
For any $N\geq1$ and $u\in H^{1}(\C)$ satisfying the following integrability
property $\int\psi_{0}\mu_{\phi_{N}+u_{N}}<\infty$ we have}
\begin{equation}
\frac{1}{\beta N^{2}}\log\E(e^{-\beta N^{2}u})\leq-\mathcal{F}(\phi+u)+\mathcal{F}(\phi)+N^{-1}a_{N}(\phi,u)+\tilde{\epsilon}_{N,\beta},\label{eq:ineq in thm upper bound part for strictly}
\end{equation}
 where 
\[
a_{N}(\phi,u):=p\left(-E_{0}(\mu_{\phi_{N}+u_{N}})+E_{0}(\mu_{\phi})\right),
\]
 which finite
\end{thm}

Theorem \ref{thm:sharp non Gaussian bd intro} also implies the following 
\begin{thm}
\emph{\label{thm:gaussian ineq for E for strong}(Sub-Gaussian inequality).
Assume that $\beta\leq1.$ Then for any $N\geq1$ and $u\in H^{1}(\C)$
}
\[
\E(e^{\beta N^{2}t\left(U_{N}-\bar{u}\right)})\leq e^{\beta N^{2}\left(\frac{t^{2}}{2}\left\Vert u\right\Vert _{H^{1}}^{2}+\frac{p}{N}B+\tilde{\epsilon}_{N,\beta}\right)},
\]
where the constant $B$ only depends on $\phi.$ 
\end{thm}

The analog of the deviation inequality in Corollary \ref{cor:Gaussian dev ineq}
then follows as before.  

\subsubsection{Sharpness}

The leading term (of the order $O(N^{2})$) of the inequalities in
Theorem \ref{thm:sharp non Gaussian bd intro} and Theorem \ref{thm:sharp Gauss bound on E intro}
is, in general, sharp. Indeed, as discussed above the inequality in
Theorem \ref{thm:sharp Gauss bound on E intro} is an asymptotic equality,
as $N\rightarrow\infty$ and so is the inequality in Theorem \ref{thm:sharp non Gaussian bd intro}
for any sufficiently small $u$ supported in the interior of $S.$
Moreover, for $\phi$ sufficiently regular we show in Section \ref{sec:Sharpness}
that the order of magnitude $O(N^{-1})$ of the error terms is also
sharp. This is shown by making contact with the asymptotic expansions
of the corresponding Coulomb gas partition functions in \cite{l-s,b-b-n-y2}.
In the general setting of a compact Riemann surface $X$ we also show
that the corresponding error term $O(N^{-1})$ is of the lower order
$o(N^{-1})$ iff $\phi$ is the Kähler potential of a metric on $X$
with constant scalar curvature iff the error term is, in fact, exponentially
small.

In another direction is is shown Section \ref{subsec:The-inequality-in fails drast}
that the inequalities in Theorem \ref{thm:sharp non Gaussian bd intro}
and Theorem \ref{thm:sharp Gauss bound on E intro} fail rather drastically
when $\beta>2.$ The study of the range $\beta\in]1,2[$ is left as
an intriguing open problem for the future (on the other hand, as shown
in \cite{berm15}, the inequalities can be extended to all $\beta>0$
at the price of adding an error term involving $\left\Vert \nabla u\right\Vert _{L^{\infty}(\C)}).$
In the case of the 2D Coulomb gas restricted to the real line it is
shown in \cite{berm16} that the inequalities corresponding to the
ones in Theorem \ref{thm:sharp non Gaussian bd intro} and Theorem
\ref{thm:sharp Gauss bound on E intro} hold precisely for $\beta\leq1.$ 

It should be stressed that in case of the Coulomb gas in $\R^{d},$
for $d>2,$ the analog of the inequalities in Theorem \ref{thm:sharp Gauss bound on E intro}
and Corollary \ref{cor:Gaussian dev ineq} fail drastically (even
though the inequalities are still asymptotic equalities when $u$
is smooth \cite{c-g-z,l-s}). There are, in fact,  basic examples
of functions $u$ such that the right hand side in the inequality
in \ref{thm:sharp Gauss bound on E intro} is finite, while the left
hand side is infinite (for any $N\geq1).$ Indeed, by Trudinger's
embedding \cite{tr} the critical exponent in the implication $|\nabla u|^{p}\in L^{1}(\R^{d})\implies e^{u}\in L_{loc}^{1}(\R^{d})$
is $p=d$ and hence any $u$ which violates the previous implication
for $p=d>2$ also violates the corresponding inequality in Theorem
\ref{eq:def of error in intro} when $d>2$ (and by a standard approximation
argument this also implies that the inequalities in question cannot
hold for all smooth test functions $u$ either when $d>2).$ 

\subsection{\label{subsec:Comparison-with-previous}Comparison with previous
results}

Finally, we make some comparisons with previous result, starting with
the Coulomb gas in the plane (there is also an extensive literature
concerning the restriction of the Coulomb gas to the real line that
we only briefly touch on).

\subsubsection{\label{subsec:Comp Deviation-inequalities}Deviation inequalities
for the 2D Coulomb gas}

Recently, the following concentration of measure inequality was obtained
in \cite[Thm 1.5, Thm 1.9]{c-h-m} for any $\beta>0,$ assuming that
$\phi$ has strictly super logarithmic growth and under a growth assumption
on $\Delta\phi$ there exists positive constants $a$ and $c,$ depending
on $\phi,$ such that 

\begin{equation}
\P\left(d_{BL}(\delta_{N},\mu_{\phi})\geq\delta\right))\leq e^{-a\beta\delta^{2}N^{2}+\frac{N}{2}\log N+c(\beta)N}\label{eq:conc ineg for balls}
\end{equation}
where $d_{BL}$ is the bounded Lipschitz distance, i.e. $d_{BL}(\mu,\nu)$
is the sup of $\left\langle u,\mu-\nu\right\rangle $ over all Lipschitz
continuous functions on $\C$such that $u$ such that $\left\Vert u\right\Vert _{\text{Lip}}^{2}\leq1$
and $\left\Vert u\right\Vert _{L^{\infty}}\leq1.$ A similar inequality
was also established for higher dimensional Coulomb gases. The inequalities
were generalized to Coulomb gases on compact Riemannian manifolds
in \cite{ga-z} (see also \cite{d-n} for a related inequality for
$\beta-$ensembles on polarized compact complex manifolds). The proof
in \cite{c-h-m} is based on a Coulomb analog of the usual transport
inequality between the Wasserstein distance and the relative entropy.
As discussed in the introduction of \cite{c-h-m} the inequality can
be viewed as a global analog of a previous inequality in \cite{R-S},
proved using the notion of renormalized energy. The concentration
inequality \ref{eq:conc ineg for balls} implies, in particular, the
following deviation inequality for linear statistics defined by a
Lipschitz continuous function $u$ such that $\left\Vert u\right\Vert _{L^{\infty}(\C)}\leq1:$
\begin{equation}
\P\left(\left|U_{N}-\bar{u}\right|>\delta\right))\leq e^{-\left\Vert u\right\Vert _{\text{Lip}}^{-2}\delta^{2}a\beta N^{2}+\frac{N}{2}\log N+c(\beta)N}\label{eq:conc ineq for lip}
\end{equation}

This inequality is analogous to the inequality in Corollary \ref{cor:Gaussian dev ineq},
but with the $H^{1}-$norm replaced by the Lipschitz norm. Accordingly,
the inequality\ref{eq:conc ineq for lip} is not sharp in the limit
$N\rightarrow\infty.$ On the other hand the analog of the concentration
inequality \ref{eq:conc ineg for balls} drastically fails for the
$H^{1}-$norm, since the $H^{1}-$distance between and $\mu_{\phi}$
and any discrete measure is infinite (see however \cite{berm15b}
for concentration inequalities valid for dual $H^{(1+\epsilon)}-$Sobolev
norms and applications to Monte-Carlo integration). 

We also recall that deviation inequalities for general determinantal
point-processes have been established in \cite{p-p} and for projectional
determinantal point processes in \cite{b-d}. However, the inequalities
in \cite{p-p,b-d} concern deviations from the \emph{mean} $\E(U_{N})$
rather than its limit $\bar{u}.$ Moreover, the speed in the general
setting of \cite{p-p} is of the order $N$ (defined as expected number
of particles) in contrast to the optimal speed $N^{2},$ arises in
the present setting. Also note that the inequalities in \cite{p-p}
involve the Lipschitz norm of $u,$ while those in \cite{b-d} involve
the $L^{\infty}-$norm of $u.$ 

\subsubsection{\label{subsec:Comp QH}Quantum Hall states on Riemann surfaces}

We recall that the collective description of the fractional Quantum
Hall effect (QHE) in terms of powers of a Slater determinant - aka
a Quantum Hall state - was originally introduced by Laughlin in the
case of the plane \cite{Lau} and has given rise to an extensive physics
literature. The corresponding set $S$ is then usually referred to
as ``the droplet'' in the plasma analogy of the Quantum Hall state.
The power $\beta=m$ for an integer $m$ correspond to the fractional
Hall conductance $1/m$ (in fundamental units $e^{2}/h$) observed
in experiments. Considering Riemann surfaces $X$ with non-trivial
topology is crucial in the standard explanation of the quantization
of the Hall conductance, introduced in \cite{t-k-n-n} in the case
when $X$ is a torus. There is a recent and rapidly expanding literature
concerning Quantum Hall states on general compact Riemann surfaces
$X$ (see the survey \cite{kl}). In the case when the corresponding
metric $\phi$ on the line bundle $L\rightarrow X$ has globally positive
curvature and $\beta=1$ asymptotic expansions of the corresponding
partition functions $Z_{N}$ is given in \cite[Thm 1]{k-m-m-w} and
related to global anomalies and adiabatic transport on the moduli
space of $(X,L)$ (expressed in terms of universal transport coefficients
in the QHE). The strict positivity of the curvature of $\phi$ corresponds,
in physical terms, to a magnetic field with a definite orientation.
Here we follow the general setup of determinantal point process on
a compact complex manifold $X$ endowed with a positive line bundle
$L$ introduced in the series of papers \cite{berm1,berm2,berm3}.
The main new feature, compared to \cite{k-m-m-w} and \cite{berm1,berm2,berm3},
respectively, is that we allow metrics $\phi$ with non-positive curvature
and singularities, respectively. As a consequence, in our situation
the ``droplet'' $S$ does not cover all of $X$ (unless $\phi$
has strictly positive curvature). In particular, our results apply
to the setup of quasi-holes located at points $p_{1},...,p_{m}$ (or
magnetic impurities) in Quantum Hall states (see Section \ref{subsec:Quasi-holes-and-magnetic}),
where ``edge effects'', i.e. contributions from the boundary of
$S^{c}$ are important.

\subsection{Acknowledgments}

This work was supported by grants from the KAW foundation, the Göran
Gustafsson foundation and the Swedish Research Council. 

\subsection{Organization }

To fix ideas we start in Section \ref{sec:Outline-of-the} by outlining
the proof of Theorem \ref{thm:sharp non Gaussian bd intro} in the
particular case of the Coulomb gas in $\C.$ Then in Section \ref{sec:Complex-geometry-and}
we setup the complex geometric and potential theoretic framework in
the general setting of a polarized compact Riemann surface $(X,L).$
In Section \ref{sec:beta ensembles-on-Riemann} the main results for
general $\beta-$ensembles on $X$ are stated and proved. The main
results about Coulomb gases in $\C,$ stated in the introduction above,
are proved in the following section by compactifying $\C$ with the
Riemann sphere $X$ and using the results proved in Section \ref{sec:beta ensembles-on-Riemann}.
We also provide explicit estimate of the error terms appearing in
the main results in $\C.$ Then in section \ref{sec:Sharpness} we
make contact with the asymptotic expansions for Coulomb gases in \cite{l-s,b-b-n-y2}
in order to show that the main results in $\C$ are essentially sharp.
In Section \ref{sec:Outlook-on-relations} an outlook on relations
to Kähler geometry is provided. In the last section we state the main
results proved in the sequel \cite{berm15} of the present paper and
in the companion papers \cite{berm15b,berm16,berm17}.

The length of the paper is, at least partly, a result of an effort
to make the paper readable both to readers with background in mathematical
physics as well as in complex geometry. 

\section{\label{sec:Outline-of-the}Outline of the proof of Theorem \ref{thm:sharp non Gaussian bd intro}
in the special case of $\C$}

Theorem \ref{thm:sharp non Gaussian bd intro} will be deduced from
a more general result on a polarized compact Riemann surface $X$
applied to the Riemann sphere, viewed as the one-point compactification
of $\C.$ But it may be illuminating to outline the main ingredients
of the proof in the setting of $\C$ considered above. First, by an
approximation argument, it is enough to consider the case when $\phi$
is continuous (and has super logarithmic growth) and $u\in C_{b}(\C).$
Moreover, using Jensen's inequality, it turns out to be enough to
consider the case when $\beta=1.$ Now, introducing the ``error functional''
$\epsilon_{N}[\Phi]$ 
\begin{equation}
\epsilon_{N}[\Phi]:=-\frac{1}{N(N+1)}\log Z_{N,1}[(N+1)\Phi]-\mathcal{F}(\Phi)\label{eq:error functional in pf}
\end{equation}
on the space of all functions $\Phi$ in $\C$ with super logarithmic
growth, we can rewrite 

\emph{
\[
\frac{1}{N(N+1)}\log\E(e^{-(N+1)NU_{N}})=-\mathcal{F}(\phi+u)+\mathcal{F}(\phi)-\epsilon_{N}[\phi+u]+\epsilon_{N}[\phi]
\]
 }(in the notation of formula \ref{eq:def of error in intro}, $\epsilon_{N}[\Phi]:=\epsilon_{N,1}[\Phi]+\epsilon_{N,1}[\psi_{0}]$).
It is essentially well-known that 
\[
\epsilon_{N}[\Phi]=o(1)
\]
as $N\rightarrow\infty$ for $\Phi$ fixed, but the crux of the matter
is to obtained a uniform lower bound on $\epsilon_{N}[\Phi]$ as $N\rightarrow\infty.$
More precisely, it is enough to show that\emph{ for any $N$ the functional
$\epsilon_{N}[\Phi]$ is minimized for $\Phi=\psi_{0}$ }i.e. 
\begin{equation}
\epsilon_{N}[\Phi]\geq\epsilon_{N}[\psi_{0}],\,\,\,\,\psi_{0}(z):=\log(1+|z|^{2})\label{eq:minimiz wrr Phi in pf}
\end{equation}
Indeed, taking $\Phi:=\phi+u$ for a given $u\in C_{b}(\C),$ then
proves Theorem \ref{thm:sharp non Gaussian bd intro}. To prove the
minimization property \ref{eq:minimiz wrr Phi in pf} one first observes
that
\[
\mathcal{F}(\Phi)=\mathcal{E}(P\Phi),\,\,\,\,4\pi\mathcal{E}(\psi)=-\frac{1}{2}\int_{\C}|\nabla(\psi-\psi_{0})|^{2}d\lambda+\int_{\C}(\psi-\psi_{0})\Delta\psi_{0}d\lambda+C_{0}
\]
 for a constant $C_{0}\in\R,$ where $P$ is the operator defined
by the following Perron type envelope:
\[
P\Phi(z):=\sup_{\psi\in\mathcal{H}(\C)}\{\psi(z):\,\,\psi\leq\Phi\,\,\text{on\,\ensuremath{\C}}\},
\]
 where $\mathcal{\overline{\mathcal{H}}}(\C)$ denotes the space of
all subharmonic functions $\psi$ in $\C$ such that $\psi-\psi_{0}$
extends to a continuous function on the Riemann sphere $X.$ \footnote{In fact, up to this point the function $\psi_{0}$ could be taken
as any function in $\overline{\mathcal{H}}(\C)$ (for example, $\psi_{0}(z):=\log\max\{1,|z|\}$
gives $C_{0}=0).$ } The function $P\Phi$ is in $\mathcal{\overline{\mathcal{H}}}(\C)$
and hence $P$ is a projection operator: $P(P\Phi)=P\Phi.$ Now, using
also that, by construction, $P\Phi\leq\Phi,$ gives 
\[
\epsilon_{N}[\Phi]\geq\epsilon_{N}[P\Phi]
\]
This means that it is enough to show that, for any $N,$ the function
$\psi_{0}(z):=\log(1+|z|^{2})$ minimizes the functional $\epsilon_{N}[\psi]$
on the space $\mathcal{\overline{\mathcal{H}}}(\C):$ 
\begin{equation}
\inf_{\psi\in\overline{\mathcal{H}}(\C)}\epsilon_{N}[\psi]=\epsilon_{N}[\psi_{0}].\label{eq:minimiz proper of space H in pf}
\end{equation}
But this follows from results in \cite{berm3} applied to the Riemann
sphere $X.$ The result that we shall need is stated in Prop \ref{prop:lower bound on L funct in adj}
in the general setting of a compact Riemann surface $(X,L)$ endowed
with a positive line bundle $L.$ In the present setting $L$ is given
by the the $(N+1)$th tensor power of the hyperplane line bundle $\mathcal{O}(1)$
on $X,$ i.e. the unique holomorphic line bundle on $X$ of unit degree.
Moreover, $\mathcal{\overline{\mathcal{H}}}(\C)$ may be identified
with the space $\overline{\mathcal{H}}(X,L)$ of continuous metrics
$\left\Vert \cdot\right\Vert $ on $\mathcal{O}(1)\rightarrow X$
with positive curvature current $\omega.$ The correspondence is made
so that 
\[
\omega_{|\C}:=\frac{i}{2}\partial\bar{\partial}\psi=\frac{1}{4\pi}\Delta\psi_{0}d\lambda
\]
for $\mathcal{\psi\in\overline{\mathcal{H}}}(\C).$ A simplifying
feature in the case when $X$ is the Riemann sphere $X$ is that $X$
is homogeneous under the action of the group of biholomorphisms of
$X$ that lifts to $L.$ As a consequence, if the metric $\left\Vert \cdot\right\Vert $
on $L$ is homogeneous - which equivalently means that $\omega$ defines
a metric on $X$ with constant scalar curvature - then the corresponding
function $\psi$ is a critical point of the functional $\epsilon_{N}[\psi]$
on $\mathcal{\overline{\mathcal{H}}}(\C).$ As shown in \cite{berm3}
, the corresponding minimizing property \ref{eq:minimiz proper of space H in pf}
then follows from the convexity of the functional $\epsilon_{N}$
along, so called, weak Kähler geodesics in the space $\mathcal{\overline{\mathcal{H}}}(X,\mathcal{O}(1)).$
\begin{rem}
\label{rem:lsc in outline}Even if $\Phi(z)$ is assumed smooth and
subharmonic in $\C$ (for example, $\Phi(z)=|z|^{2})$ the corresponding
metric on $\mathcal{O}(1)\rightarrow X$ is not, in general, locally
bounded over $X$ and its curvature current is not positive (see Example
\ref{exa:phi as sing metric on X}). This is one of the reasons that
we will work with a general setup of metrics on a line bundle $L\rightarrow X$
which (in our additive notation) are merely lower semi-continuous
and may take the value $+\infty.$ Another good reason to work with
lsc metrics is that they also naturally appear in geometric, as well-as
physical situations (see Section \ref{subsec:Quasi-holes-and-magnetic}).
\end{rem}

\subsection{The convexity of the functional $\epsilon_{N}$ along weak Kähler
geodesics}

For the convenience of the reader we outline the proof of the convexity
in question from \cite{berm3} in the present setting of $\C.$ The
weak Kähler geodesic $\psi_{t}\in\mathcal{H}(\C)$ connecting given
$\psi_{0}$ and $\psi_{1}$ in $\mathcal{H}(\C)$ may be expressed
as $\psi_{t}(z)=\Psi(t,z)$ where $\Psi(\tau,z)$ is the continuous
plurisubharmonic function on the strip $([0,1]+i\R)\times\C$ defined
as the following envelope: 

\[
\Psi(\tau,z)=\sup\left\{ \widetilde{\Psi}(\tau,z):\,\,\widetilde{\Psi}(0,z)=\psi_{0}(z)\,\,\widetilde{\Psi}(1,z)=\psi_{1}(z)\right\} ,
\]
 where the sup runs over all continuous plurisubharmonic functions
$\widetilde{\Psi}(\tau,z)$ on $[0,1]+i\R)\times\C$ such that $\widetilde{\Psi}$
is independent of the imaginary part of $\tau.$ It follows from well-known
results in Kähler geometry that $\epsilon_{N}[\psi_{t}]$ is affine
and all that remains is to verify that 
\[
\mathcal{L}_{N}(\psi_{t}):=-\frac{1}{N(N+1)}\log Z_{N,1}[(N+1)\psi_{t}]
\]
 is convex in $t.$ To this end one first uses $\beta=1$ to rewrite
$Z_{N,1}[(N+1)\psi]$ as the determinant of Gram matrix:
\[
Z_{N,1}[(N+1)\psi]=N!\det_{1,i,j\leq N}\left(\int_{\C}\Psi_{i}(z)\overline{\Psi_{j}(z)}e^{-(N+1)\psi}d\lambda\right),
\]
 where $\Psi_{i}(z)=z^{i-1}$ is the standard base of monomials in
the space $\mathcal{P}_{N-1}(\C)$ of all polynomials of degree at
most $N-1.$ The convexity of $\mathcal{L}_{N}(\psi_{t})$ then follows
from general positivity results in \cite{bern} for direct images
of adjoint line bundles, using that $\mathcal{P}_{N-1}(\C)$ may be
identified with the space $H^{0}(X,(N+1)\mathcal{O}(1)+K_{X})$ of
all global holomorphic one-forms on $X$ with values in the $N+1$th
tensor power of $\mathcal{O}(1).$ 

\section{\label{sec:Complex-geometry-and}Complex geometry and potential theory
on compact Riemann surfaces }

In this section we provide preliminary material on the complex geometry
and potential theory on a polarized Riemann surface $(X,L).$ We will
mainly follow the notation in \cite{b-b} (which concerns the considerably
more involved case of an $n-$dimensional complex manifold). The main
new feature compared to \cite{b-b} (and \cite{berm1}) is that we
allow non-continuous metrics $\phi$ on $L:$ first lower semi-continuous
(lsc) metrics and then $H^{1}-$metrics i.e. metrics whose local gradient
is in $L^{2}.$ For the latter we exploit the link to the intrinsic
Hilbert space $H^{1}(X)/\R$ and its dual, which is specific for the
case $n=1.$ This link also underlies the inequality \ref{lem:smaller than dirichlet}
that will be used to obtain sub-Gaussian estimates.

\subsection{The complex geometric setup}

Let $(X,L)$ be a\emph{ polarized compact Riemann surface}, i.e. a
compact complex manifold $X$ of complex dimension $n=1$ endowed
with a holomorphic line bundle of positive degree: 
\[
\deg L:=\int_{X}c_{1}(L)>0,
\]
 where $L$ is the first Chern class of $L$ in the integral part
of the second de Rham cohomology group $H^{2}(X,\R).$  We will use
additive notation for tensor powers and Hermitian metrics on line
bundles. We recall that a holomorphic line bundle $L$ over $X$ is
determined by fixing a covering of $X$ by a finite number of open
sets $U_{i}$ and trivializing (i.e. non-vanishing) holomorphic sections
$e_{U_{i}}$ of $L$ over $U_{i}.$ We denote by $t_{ij}$ the corresponding
transition functions, i.e. the holomorphic invertible functions on
$U_{i}\cap U_{j}$ defined by $t_{ij}=e_{U_{i}}/e_{U_{j}}.$

\subsubsection{\label{subsec:Metrics-on}Metrics on $L$}

In our additive notation a metric $\left\Vert \cdot\right\Vert $
on $L$ (that will be assumed to be smooth to start with) will be
denoted by the symbol $\phi,$ where $\phi$ may be locally represented
by a smooth function. More precisely, given a covering of $X$ and
local trivializations of $L$ as above, a metric $\left\Vert \cdot\right\Vert $
on $L$ is represented by the collection $\phi:=\{\phi_{U_{i}}$\}
of functions $\phi_{U_{i}}$ on $U_{i}$ defined by 
\[
\left\Vert e_{U_{i}}\right\Vert ^{2}:=e^{-\phi_{U_{i}}}
\]
Hence, $\phi_{U_{i}}=\phi_{U_{j}}-\log\left|t_{ij}\right|^{2}.$ The
object $\phi$ is often called a \emph{weight} (cf. \cite{b-b}),
but will here, abusing notation slightly, be identified with the corresponding
metric on $L.$ The (normalized) curvature form of a metric $\phi$
on $L$ is the globally well-defined two-form on $X$ defined on $U$
by 
\[
\omega^{\phi}:=\frac{i}{2\pi}\partial\bar{\partial\phi_{U_{i}}}=:dd^{c}\phi_{U_{i}},
\]
 where $d$ is the exterior derivative and $d^{c}$ is the real operator
defined by $d^{c}:=-\frac{1}{4\pi}J^{*}d$ in terms of the complex
structure $J$ on the real tangent bundle $TX.$ Accordingly, the
curvature form of a metric $\phi$ is usually symbolically written
as $dd^{c}\phi$ (abusing notation slightly). 
\begin{rem}
The normalizations have been chosen so that in $\C$ $dd^{c}\phi=\frac{1}{4\pi}\Delta\phi d\lambda$
and hence $dd^{c}\log|z|^{2}=\delta_{0}$ where $\delta_{0}$ denotes
the Dirac mass at $0$ in $\C.$ Also note that compared with standard
notation in gauge theory $dd^{c}\phi=\frac{i}{2\pi}F_{A}$ where $F_{A}$
denotes the curvature two-form of the Chern connection $A$ on the
principle $U(1)-$bundle corresponding to the Hermitian holomorphic
line bundle $(L,\phi).$ 
\end{rem}

The induced metric on the $k$ th tensor powers of $L,$ written as
$kL,$ in our additive notation, is represented by $k\phi.$ Taking
the difference between two metrics $\phi_{1}$ and $\phi_{2}$ on
$L$ yields a globally well-defined function on $X$ (i.e. the space
of metrics on $L$ is an affine space modeled on $C^{\infty}(X)).$
We will denote by $\mathcal{H}(L)$ the space of all smooth metrics
$\psi$ with strictly positive curvature: 
\[
\mathcal{H}(L):=\left\{ \psi:\,dd^{c}\psi>0\right\} 
\]
We denote by $\psi_{0}$ a fixed reference element $\mathcal{H}(L),$
which in the present Riemann surface case can (and will) be taken
so that the Riemann metric defined by the curvature form 
\begin{equation}
\omega_{0}:=dd^{c}\psi_{0}\label{eq:def of mu noll}
\end{equation}
 has constant scalar curvature, when identified with an Hermitian
metric on $X.$ The probability measure 
\[
\mu_{0}:=\frac{\omega_{0}}{\int\omega_{0}}(=\frac{\omega_{0}}{\deg L})
\]
 will be referred to as the\emph{ canonical volume form} on $X$ and
the corresponding metric $\psi_{0}$ on $L$ will be referred to as
the\emph{ canonical reference metric} on $L$ (abusing notation slightly,
since $\psi_{0}$ is only determined up to an additive constant when
$X$ has genus at least one and up to automorphisms when $X$ is the
Riemann sphere). Anyway, in the case when $X$ is the Riemann sphere
we will make a particular choice of constant.

\subsubsection{Holomorphic sections of $L$}

The complex vector space of all global holomorphic sections of $L$
is denoted by $H^{0}(X,L)$ and its dimension by $N_{L}:$ 

\[
N_{L}:=\dim H^{0}(X,L)
\]
Given an element $\Psi$ and a metric $\phi$ on $L$ the squared
point-wise norm of $\left\Vert \Psi\right\Vert _{\phi}^{2}$ is a
global function on $X,$ symbolically expressed as $|\Psi|^{2}e^{-\phi}.$
More precisely, fixing a covering $U_{i}$ and local trivializations
$e_{i}$ of $L,$ as above, an element $\Psi\in H^{0}(X,L)$ may be
locally written as 
\[
\Psi=\Psi_{U_{i}}e_{U_{i}}
\]
for a holomorphic function $\Psi_{U_{i}}$ on $U_{i}$ and 
\[
\left\Vert \Psi\right\Vert _{\phi}^{2}:=|\Psi|^{2}e^{-\phi}:=|\Psi_{U_{i}}|^{2}e^{-\phi_{U_{i}}}
\]
on $U_{i}.$ Fixing a continuous volume form $dV$ on $X$ we will
denote by $\left\langle \cdot,\cdot\right\rangle _{(dV,\phi)}$ the
corresponding scalar product on $H^{0}(X,L),$ i.e. 

\[
\left\langle \Psi,\Psi\right\rangle _{(dV,\phi)}:=\int_{X}\Psi\bar{\Psi}e^{-\phi}dV.
\]

\subsubsection{\label{subsec:The-asymptotic-setting and the fixed references}The
asymptotic setting and the canonical reference basis}

We will be interested in the limit where $L$ is replaced by 
\[
L_{k}:=kL+F
\]
 and $k\rightarrow\infty$ for fixed line bundles $L$ and $F.$ By
the Riemann-Roch theorem on a Riemann surface this corresponds to
the limit where the dimension $N_{L_{k}}$ tends to infinity. Indeed,
for $k$ sufficently large (so that $\dim H^{1}(X,kL+F)=0),$ 
\[
N_{k}:=N_{kL+F}=\mbox{\ensuremath{\deg}(}L)k-(g(X)-1)+\mbox{\ensuremath{\deg}\ensuremath{(F)}},
\]
 where $g(X)$ is the genus of $X.$ In general, a subindex $k$ will
indicate that the object in question is defined with respect to the
line bundle $kL+F$ and a metric on $kL+F$ of the form $k\phi+\phi_{F}.$ 

We fix, once and for all, a ``reference basis'' $\Psi_{1}^{(k)},...\Psi_{N_{k}}^{(k)}$
in $H^{0}(X,kL+F)$ which is orthonormal wrt the scalar product determined
by the reference volume form $\mu_{0}$ on $X,$ the reference metric
$\psi_{0}$ on $L$ and a fixed continuous reference metric $\phi_{F}$
on $F.$ In the particular case when $F=K_{X},$ the canonical line
bundle, there is a canonical choice of metric $\phi_{F},$ as explained
in the next section. 

\subsubsection{\label{subsec:The-adjoint-setting}The adjoint setting}

In the special case when $F=K_{X}$ there is a natural $L^{2}-$norm
on $H^{0}(X,kL+K_{X})$ induced by a metric $\phi$ on $L:$

\begin{equation}
\left\langle \Psi,\Psi\right\rangle _{k\phi}:=\frac{i}{2}\int_{X}\Psi\wedge\bar{\Psi}e^{-k\phi}\label{eq:herm norm in adjoint setting}
\end{equation}
identifying the elements of $H^{0}(X,kL+K_{X})$ with holomorphic
$1-$forms with values in $kL.$ Equivalently, in terms of the general
setting above, this corresponds to fixing any volume form $dV$ on
$X$ and taking $\phi_{F}$ as the metric on $F=K_{X}$ induced by
$dV.$ Indeed, fixing local coordinates $z=(z_{1},...,z_{n})$ induces
a trivialization $dz:=dz_{1}\wedge\cdots\wedge dz_{n}.$ By definition
$\phi_{F}:=-\log(i^{n^{2}}dz\wedge d\bar{z}/dV)$ and hence

\[
\frac{i}{2}\Psi\wedge\bar{\Psi}e^{-k\phi}=|\Psi|^{2}e^{-k\phi}e^{-\phi_{F}}dV
\]
This will be called the ``adjoint setting'' (in the algebraic geometry
literature the line bundle $kL+K_{X}$ is usually said to be the line
bundle which is adjoint to the line bundle $kL$). We will 

\subsection{\label{subsec:Potential-theoretic-setup}Potential-theoretic setup}

\subsubsection{Admissible singular metrics on $L$ and the spaces $PSH(L)$ and
$PSH(X,\omega_{0})$}

It will be important to allow metrics $\phi$ on $L$ which are singular,
i.e. the local weights of $\phi$ take values in $[-\infty,\infty].$
The corresponding point-wise norm of a given section $\Psi$ in $H^{0}(X,L)$
thus defines a function $X\rightarrow[0,\infty].$ To simplify the
notation we will drop the adjective singular and simply call $\phi$
a\emph{ metric. }
\begin{defn}
A metric $\phi$ on $L$ will be said to be\emph{ admissible }if $\phi$
is lsc, taking values in $]-\infty,\infty]$ and $\{\phi<\infty\}$
is not polar. Equivalently, this means that we can write $\phi=\psi_{0}+v$
where $v$ is a globally defined function on $X$ taking values in
$]-\infty,\infty]$ such that $\{v<\infty\}$ is not polar.
\end{defn}

We will denote by $PSH(L)$ the space of metrics $\psi$ on $L\rightarrow X$
such that $\psi$ is locally \emph{plurisubharmonic (psh)}. In other
worst, $\psi$ is locally represented by a plurisubharmonic function,
i.e. a function in $L_{loc}^{1}$which is strongly upper semi-continuous
function and whose curvature defines a positive current on $X.$ In
the present Riemann surface setting this simply means that $\psi$
is locally \emph{subharmonic.} The subspace of all locally bounded
metrics in $PSH(L)$ will be denoted by $PSH(L)_{b}.$
\begin{rem}
$\psi\in PSH(L)$ is not admissible, unless $\psi$ is continuous. 

In general, the map 
\[
\psi\mapsto\varphi:=\psi-\psi_{0},\,\,\,PSH(L)\rightarrow PSH(X,\omega_{0})
\]
 yields a bijection between the space $PSH(L)$ and the space $PSH(X,\omega_{0})$
of all $\omega_{0}-$psh functions, i.e. all strongly usc functions
$\varphi$ such that $dd^{c}\varphi+\omega_{0}\geq0$ holds in the
sense of currents. In the present Riemann surface setting the latter
condition simply means that $-\frac{1}{4\pi}\Delta_{\omega_{0}}\varphi\geq-1,$
in the sense of distributions, where $\Delta_{\omega_{0}}$ denotes
the ``positive Laplacian'' determined by $\omega_{0},$ i.e. $\Delta_{\omega_{0}}v:=-dd^{c}v/\omega_{0}.$ 
\end{rem}

\begin{lem}
\label{lem:sup estimate}Suppose that $\deg L=1$ and let $\mu$ be
a probability measure on $X$ of the form $\mu:=dd^{c}\varphi+\omega_{0}$
for $\varphi\in L^{\infty}(X).$ Then, for any $v\in PSH(X,\omega_{0})$
we have
\[
\sup_{X}v\leq\int_{X}v\mu+\sup_{X}\varphi-\inf_{X}\varphi+C_{0},\,\,\,C_{0}:=\sup_{X\times X}(-G_{0}),
\]
 where $G_{0}$ denotes the Green functions of $\Delta_{\omega_{0}}$
(see formula \ref{eq:def of Green f}).
\end{lem}

\begin{proof}
By approximation we may as well assume that $v$ is smooth. Integrating
by parts then gives
\[
\int_{X}v\mu=\int_{X}v\omega_{0}+\int\varphi(dd^{c}v+\omega_{0}-\omega_{0})\geq\int_{X}v\omega_{0}+(\inf_{X}\varphi-\sup_{X}\varphi)
\]
Finally, we have $\sup v\leq\int_{X}v\omega_{0}\leq C_{0}.$ Indeed,
we may as well assume that $\int v\omega_{0}=0$ and then $v(x)=-\int G_{0}(x,\cdot)(dd^{c}v+\omega_{0})\leq C_{0}.$ 
\end{proof}

\subsubsection{The projection operator $P$ and the equilibrium measure $\mu_{\phi}$}
\begin{defn}
Let $\phi$ be a metric on $L.$ Then $P\phi$ is the metric on $L$
defined as the upper semi-continuous regularization of a Perron type
envelope:
\begin{equation}
P\phi:=\underline{\left\{ \sup\psi\in PSH(L):\,\psi\leq\phi\right\} },\label{eq:def of P}
\end{equation}
where we use (the non-standard) notation $\underline{f}$ for the
upper semi-continuous regularization of a function $f$ on $X$ and
similarly for a metric on $L.$ \footnote{The notation $f^{*}$ is usually used in the complex analysis literature,
but we will reserves the upper star for the Legendre-Fenchel transform,
introduced below{]}}
\end{defn}

\begin{lem}
\label{lem:prop of p phi}Assume that $\{\phi<\infty\}$ is not polar. 
\begin{itemize}
\item $P\phi$ is in $PSH(X,L)$ and $P\phi\leq\phi$ quasi-everywhere,
i.e. on the complement of a polar subset.
\item If $\phi$ is admissible, then $P\phi\in PSH(X,L)_{b}.$
\item If $\phi_{j}$ is a sequence of metrics decreasing to $\phi$ then
$P(\phi_{j})$ decreases to $P(\phi).$
\item If $\phi_{j}$ is a sequence of admissible metrics increasing to the
admissible metric $\phi$ then $P(\phi_{j})$ increases quasi-everywhere
to $P(\phi)$ in $PSH(X,L)_{b}$
\end{itemize}
\end{lem}

\begin{proof}
This is proved in essentially the same was as in \cite[Prop 2.2]{glz}.
\end{proof}
\begin{defn}
\label{def: eq measure as P}Let $\phi$ be an admissible metric on
$L.$ The corresponding \emph{equilibrium measure }$\mu_{\phi}$ is
the probability measure on $X$ defined by 
\begin{equation}
\mu_{\phi}:=dd^{c}(P\phi)/\text{deg \ensuremath{L}}\label{eq:def of muphi}
\end{equation}
Note that the measure $\mu_{\phi}$ does not charge polar subsets
(since its potential $P\phi$ is bounded). We will denote by $S$
the support of $\mu_{\phi.}$
\end{defn}

\begin{prop}
\label{prop:og relation etc}Let $\phi$ be an admissible metric.
Then the support $S$ of the corresponding equilibrium measure $\mu_{\phi}$
is contained in the coincidence set $D,$ i.e. the closed subset of
$X$ defined by 
\begin{equation}
D=\left\{ P\phi=\phi\right\} \label{eq:def of coinc set}
\end{equation}
In particular, the Laplacian $\Delta$ vanishes on the open set $\left\{ P\phi<\phi\right\} $
and hence the following ``orthogonality relation'' holds 
\begin{equation}
\int_{X}(\phi-P\phi)dd^{c}(P\phi)=0\label{eq:og relation}
\end{equation}
\end{prop}

\begin{proof}
First note that $D$ is indeed closed as $P\phi-\phi$ is usc. In
the case when $\phi$ is continuous the orthogonality relation \ref{eq:og relation}
is a special case of \cite[Prop 2.10]{b-b}. In the general case we
take a sequence of continuous $\phi_{j}$ increasing to $\phi.$ By
Lemma \ref{lem:prop of p phi} $\psi_{j}:=P(\phi_{j})$ increases
a.e. to $\psi:=P\phi$ and hence $\int_{X}(\phi_{j}-\psi_{j})dd^{c}\psi_{j}(=0)$
converges towards $\int_{X}(\phi-\psi)dd^{c}\psi$ (as follows, for
example, from Prop \ref{prop:cont of beatiful E under monotone} below).
\end{proof}
Even if $\phi$ is assumed smooth, $P\phi$ is not smooth, in general,
but only only $C^{1,1}-$smooth. In general, we have the following
\begin{prop}
\label{prop:regularity}Let $\phi$ be a metric on $L$ satisfying
the assumptions A0 and A1 appearing in the beginning of Section \ref{sec:beta ensembles-on-Riemann}
below. Then $\Delta(P\phi)\in L_{loc}^{\infty}.$ As a consequence,
\begin{equation}
\Delta(P\phi)=1_{\left\{ P\phi=\phi\right\} }\Delta\phi\label{eq:Laplac of P phi}
\end{equation}
 locally, and
\begin{equation}
\mu_{\phi}=1_{S}dd^{c}\phi\label{eq:eq measure for phi smooth}
\end{equation}
where $S$ denotes the support of $\mu_{\phi.}$ 
\end{prop}

\begin{proof}
In the setting where $\phi$ corresponds to a function $\phi$ in
$\C$ with strictly super logarithmic growth (which, as explained
in Section \ref{sec:The-Coulomb-gas}, is a special case of the general
complex geometric setup) a simple proof was given in \cite{berm 1 komma 1}.
In the case when $\phi$ is in $C_{loc}^{2}$ the case of a Riemann
surface $X$ is the one-dimensional case of \cite[Thm 1.1]{berm4}
and a straight-forward modification of the proof of \cite[Thm 1.1]{berm4}
also applies to $\phi$ satisfying A0 and A1.
\end{proof}

\subsubsection{The functional $\mathcal{E}$ on $PSH(L)$ and the free energy functional
$\mathcal{F}$ on $C(X)$}

The functional $\mathcal{E}$ 
\[
\mathcal{E}:\,PSH(L)_{b}\rightarrow\R
\]
is defined as the primitive of the operator $\psi\mapsto dd^{c}\psi/\text{deg\ensuremath{\,L}}$
taking values in the space of probability measures on $X,$ i.e. the
differential of $\mbox{\emph{\ensuremath{\mathcal{E}} }}$on the convex
space $PSH(L)_{b}$ is given by 
\begin{equation}
d\mathcal{E}_{|\psi}:=dd^{c}\psi/\text{deg}(L)\label{eq:diff of beat E is ddc}
\end{equation}

Accordingly, the functional $\mathcal{E}$ is only defined up to an
additive constant. The constant will be fixed by the normalization
condition 
\[
\mathcal{E}(\psi_{0})=0
\]
 on the reference element $\psi_{0}$ in $PSH(L)_{b}.$ Integrating
along the affine line in $PSH(L)_{b}$ connecting $\psi$ and $\psi_{0}$
reveals that we may equivalently make the following
\begin{defn}
The functional $\mathcal{E}$ on $PSH(L)_{b}$ is defined by 
\begin{equation}
\mathcal{E}(\psi)=\frac{1}{2\text{deg}(L)}\int_{X}(\psi-\psi_{0})(dd^{c}\psi+dd^{c}\psi_{0}),\label{eq:explic formula for beaut E}
\end{equation}
(formula \ref{eq:diff of beat E is ddc} then follows from integrating
by parts). More generally, the functional $\mathcal{E}$ may be extended
to all of $PSH(L)$ as the smallest upper semi-continuous extension
from $PSH(L)_{b}.$ 
\end{defn}

\begin{rem}
Occasionally, when emphasizing the dependence of $\mathcal{E}$ on
a fixed reference element $\psi_{0}$ we will write $\mathcal{E}_{\psi_{0}}.$
By the property \ref{eq:diff of beat E is ddc} the following cocycle
relation holds
\begin{equation}
\mathcal{E}_{\psi_{1}}=\mathcal{E}_{\psi_{0}}+\mathcal{E}_{\psi_{1}}(\psi_{0})\label{eq:cocycl relation for primitive}
\end{equation}
and differences $\mathcal{E}(\phi_{1})-\mathcal{E}(\phi_{0})$ are
independent of the choice of reference. 
\end{rem}

It follows immediately from the defining relation \ref{eq:diff of beat E is ddc}
that $\mathcal{E}$ is increasing on $PSH(L)_{b}$ and satisfies the
scaling relation 
\begin{equation}
\mathcal{E}(\psi+C)=\mathcal{E}(\psi)+C\label{eq:scaling of energy func}
\end{equation}
Moreover, the functional $\mathcal{E}$ has the following basic continuity
properties \cite[Prop 4.3]{b-b}:
\begin{prop}
\label{prop:cont of beatiful E under monotone}Let $\psi_{j}$ be
a sequence in $PSH(X,L)$ decreasing (increasing almost everywhere)
to $\psi\in PSH(X,L).$ Then $\mathcal{E}(\psi_{j})$ decreases (increases)
to $\mathcal{E}(\psi).$ 
\end{prop}

Integrating by parts reveals that, when $\psi$ is smooth, 
\begin{equation}
-\mathcal{E}(\psi)=J(\psi-\psi_{0})-\int(\psi-\psi_{0})dd^{c}\psi_{0}/\deg L,\label{eq:energy functional in terms of Dirichlet}
\end{equation}
 where $J$ denotes the half of the normalized squared Dirichlet norm:
\[
J(u):=\frac{1}{2\deg L}\int du\wedge d^{c}u
\]
The following proposition shows that the previous formula still holds
when $\psi$ is singular or more precisely when $u:=\psi-\psi_{0}$
has a gradient in $L^{2}(X),$ i.e. when $u$ is an element of the
intrinsic Hilbert space $H_{0}^{1}(X),$ whose definition is recalled
in Section \ref{subsec:The-functionals- in terms of intr}. 
\begin{prop}
Let $X$ be a compact Riemann surface. Then $\psi\in PSH(X,L)$ satisfies
$-\mathcal{E}(\psi)<\infty$ iff the $\omega_{0}-$psh function $u:=\psi-\psi_{0}$
has a gradient in $L^{2},$ i.e. 
\[
H^{1}(X)\cap PSH(X,\omega_{0})=\{u\in PSH(X,\omega):\,\mathcal{E}(\psi_{0}+u)>-\infty\}.
\]
 Moreover, formula \ref{eq:energy functional in terms of Dirichlet}
holds for any $\psi\in PSH(X,L)$ such that $\mathcal{E}(\psi)>-\infty,$
when $J$ is interpreted as the (normalized) $L^{2}-$norm of $du.$
\end{prop}

\begin{proof}
This is well-known and can, for example, be shown using Prop \ref{prop:cont of beatiful E under monotone}
together with the basic fact that $H^{1}(X)/\R$ is the completion
of $C^{\infty}(X)/\R$ wrt the $L^{2}-$norm defined by the functional
$J$ (the proof is similar to the proof of the first point in Prop
\ref{prop:prop of energy when =00005Cphi is cont}). 
\end{proof}
Next, we consider the\emph{ free energy functional} $\mathcal{F}$
on $C(X)$ defined by 
\begin{equation}
\mathcal{F}_{\phi}(u):=\mathcal{E}(P(\phi+u))\label{eq:def of free energy text}
\end{equation}

(that we will occasionally identify as a functional on the space of
all metrics on $L$ and then write $\mathcal{F}_{\phi}(u)=\mathcal{F}(\phi+u))$ 
\begin{prop}
\label{prop:beatif F is diff}Let $\phi$ be an admissible metric
on $L.$ Then the free energy functional $\mathcal{F}$ is Gateaux
differentiable on $C^{0}(X)$ and its differential at $u\in C^{0}(X)$
is represented by the measure $dd^{c}P(\phi+u)(:=\mu_{\phi+u}).$ 
\end{prop}

\begin{proof}
When $\phi$ is continuous this is a special case of \cite[Thm B]{b-b}.
We recall that the proof is based on the the orthogonality relation
\ref{eq:og relation} and the same proof thus applies when $\phi$
is admissible. 

The following inequality is specific for the case of one complex dimension:
\end{proof}
\begin{lem}
\label{lem:smaller than dirichlet}Let $\phi$ a smooth metric on
$L.$ Then 
\[
\mathcal{E}(P(\phi))-\mathcal{E}(P(\phi+u)\leq J(u)-\int_{X}udd^{c}(P\phi)/\deg L
\]
 and, equivalently, 
\[
J\left(P(\phi+u)-P(\phi)\right)\leq J(u)
\]
\end{lem}

\begin{proof}
First observe that setting $\psi:=P\phi$ we have 
\begin{equation}
\mathcal{E}(P(\phi))-\mathcal{E}(P(\phi+u)\leq\mathcal{E}(\psi)-\mathcal{E}(P(\psi+u)\label{eq:proof of smaller than dirichlet}
\end{equation}
Indeed, since $\psi\leq\phi$ we get
\[
P(\psi+u)\leq P(\phi+u)
\]
Since $\mathcal{E}$ is increasing on the space $PSH(X,L)$ it follows
that 
\[
\mathcal{E}(P(\psi+u)\leq\mathcal{E}(P(\phi+u)
\]
proving the inequality \ref{eq:proof of smaller than dirichlet} Hence,
the lemma follows from the inequality 
\[
\mathcal{E}(\psi)-\mathcal{E}(P(\psi+u)\leq J(u)-\int_{X}udd^{c}\psi/\deg L
\]
for any locally bounded psh-metric $\psi$ on $L.$  The latter inequality
is equivalent to \cite[Prop 2.3]{berm3} applied to $\omega:=dd^{c}\psi.$
For the conveniens of the reader we recall the proof (where, to simplify
the notation, we have assumed that $\deg L=1)$. 
\[
\mathcal{E}(\psi)-\mathcal{E}(P(\psi+u)+\int_{X}udd^{c}\psi=\int(\psi+u-P(\psi+u))(dd^{c}\psi+dd^{c}(P(\psi+u))=
\]
\[
=\int(\psi+u-P(\psi+u))(dd^{c}\psi-dd^{c}(P(\psi+u))
\]
 using the orthogonality relation \ref{eq:og relation} applied to
$\psi+u$ in the last equality. But the latter expression is nothing
but the Dirichlet norm of $\psi+u-P(\psi+u)$ and hence non-negative,
which concludes the proof of the first inequality. The second one
the follows from the cocycle property \ref{eq:cocycl relation for primitive}
together with formula \ref{eq:energy functional in terms of Dirichlet}.
\end{proof}

\subsubsection{\label{subsec:The-energy-functional E phi as Leg tr}The energy functional
$E_{\phi}$ on $\mathcal{P}(X)$ as a Legendre-Fenchel transform }

We first recall that if $f$ is a function on a topological vector
space $V,$ then its Legendre-Fenchel transform is defined as following
convex lower semi-continuous function $f^{*}$ on the topological
dual $V^{*}$ 
\[
f^{*}(w):=\sup_{v\in V}\left\langle v,w\right\rangle -f(v)
\]
in terms of the canonical pairing between $V$ and $V^{*}.$ In the
present setting we will, to start with, take $V=C^{0}(X),$ endowed
with the sup norm. Then $V^{*}=\mathcal{M}(X)$ is the space of all
signed Borel measures on $X.$ 
\begin{defn}
Given a continuous metric $\phi$ on $L\rightarrow X$ the corresponding
energy functional $E_{\phi}$ on the space of signed measures $\mathcal{M}(X)$
is defined as the (convex) Legendre-Fenchel transform of the functional
$u\mapsto-\mathcal{F_{\phi}}(-u):=-\mathcal{E}\circ P(\phi-u)$ on
$C^{0}(X):$
\[
E_{\phi}(\mu):=\sup_{u\in C^{0}(X)}\mathcal{F_{\phi}}(u)-\left\langle u,\mu\right\rangle 
\]
A probability measure $\mu$ on $X$ is said to have\emph{ finite
energy} if $E_{\phi}(\mu)<\infty$ for some (or equivalently any)
continuous metric $\phi$ on $L.$
\end{defn}

\begin{rem}
Occasionally, when emphasizing the dependence of $E_{\phi}$ on the
fixed reference element $\psi_{0}$ we will write $E_{\psi_{0},\phi}.$
By the cocycle property \ref{eq:cocycl relation for primitive} 
\[
E_{\psi_{1},\phi}(\mu)=E_{\psi_{0},\phi}(\mu)+\mathcal{E}(\psi_{1})-\mathcal{E}(\psi_{2})
\]
\end{rem}

It follows directly from the definition that, for any given function
$v$ 
\[
E_{\phi+v}(\mu)=E_{\phi}(\mu)+\int_{X}v\mu
\]

Applying the standard duality theory on locally convex topological
vector spaces gives the following result (see \cite[Prop 2.10]{berm2 komma 5}): 
\begin{prop}
\label{prop:diff of composed energy}Let $\phi$ be a continuous metric
on $L.$ 

\begin{itemize}
\item If $\mu$ is a signed measure on $X,$ i.e. $\mu\in\mathcal{M}(X),$
then $E_{\phi}(\mu)=\infty$ unless $\mu$ is in $\mathcal{P}(X).$
\item The functional $u\mapsto-\mathcal{F}_{\phi}(-u)$ on $C^{0}(X)$ is
the Legendre-Fenchel transform of $E_{\phi}:$ 
\begin{equation}
\mathcal{F}_{\phi}(u)=-E_{\phi}^{*}(-u):=\inf_{\mu\in\mathcal{M}(X)}\left(E_{\phi}(\mu)+\left\langle u,\mu\right\rangle \right)\label{eq:leg transf relations between energy f}
\end{equation}
\item The equilibrium measure $\mu_{\phi}$ (Definition \ref{def: eq measure as P})
is the unique minimizer of $E_{\phi}$ on $\mathcal{P}(X).$
\end{itemize}
\end{prop}

\subsubsection{\label{subsec:The-functionals- in terms of intr}The intrinsic Hilbert
space $H^{1}(X)/\R$ and $H_{0}^{-1}(X)$}

Denote by $H^{1}(X)$ the Sobolev space of all functions $u\in L_{loc}^{2}(X)$
such that $\nabla u\in L_{loc}^{2}$ (in the sense of distributions).
The quadratic form 
\[
\left\Vert u\right\Vert _{H^{1}}^{2}:=\int_{X}du\wedge d^{c}u
\]
 induces an intrinsic Hilbert space structure on the quotient space
$H^{1}(X)/\R$ (only depending on the complex structure on $X).$
By standard Sobolev space theory $C^{\infty}(X)/\R$ is dense in $H^{1}(X)/\R.$
Fixing a Riemannian metric $g$ compatible with the complex structure
on $X$ we have
\[
\left\Vert u\right\Vert _{H^{1}}^{2}:=\frac{1}{4\pi}\int|\nabla_{g}u|^{2}dV_{g},
\]
 We denote by $H^{1}(X,g)$ the vector space $H^{1}(X)$ endowed with
the Hilbert norm 
\[
\left\Vert u\right\Vert _{H^{1}(X,g)}^{2}:=\int_{X}|u|^{2}dV_{g}+\frac{1}{4\pi}\int|\nabla_{g}u|^{2}dV_{g},
\]
(where only the first term depends on $g).$ We endow $H^{1}(X)$
with the corresponding topology (which is independent of the choice
of $g).$ In the following we will usually take $g$ to be the metric
with volume (area) form $dV_{g}:=dd^{c}\psi_{0},$ where $\psi_{0}$
is a fixed reference element in the space $PSH(X,L).$ Then $H^{1}(X)/\R$
may be identified with subspace $H_{0}^{1}(X)$ of $H^{1}(X,g)$ of
all functions $u$ such that $\int_{X}udV_{g}=0.$ We endow $H_{0}^{1}(X)$
with the topology induced by the Hilbert norm. This is also called
the \emph{strong topology} on $H_{0}^{1}(X),$ while the \emph{weak
topology} on $H_{0}^{1}(X)$ is defined as the weak topology in the
sense of of Hilbert spaces, i.e. $u_{j}$ converges to $u_{\infty}$
weakly in $H_{0}^{1}(X)$ iff $\left\langle u_{j},u\right\rangle \rightarrow\left\langle u_{\infty},u\right\rangle $
for any $u\in H_{0}^{1}(X).$ 

We denote by$H^{-1}(X)$ the Sobolev space of distributions on $X$
which is dual to $H^{1}(X)$ and by $H_{0}^{-1}(X)$ the subspace
consisting of all distributions $\nu$ in $H^{-1}(X)$ with zero mean,
i.e. $\left\langle \nu,1\right\rangle =0.$\footnote{We view a distribution $\nu$ on the two-dimensional manifold $X$
as a current of dimension $0,$ i.e. of degree $2,$ so that the pairing
$\left\langle \nu,u\right\rangle $ is intrinsically defined for any
$u\in C^{\infty}(X).$ } We endow $H_{0}^{-1}(X)$ with the intrinsic Hilbert norm which is
dual to $H^{1}(X)/\R:$

\begin{equation}
\left\Vert \nu\right\Vert _{H^{-1}(X)}^{2}:=\sup_{u\in C^{\infty}(X)}\frac{|\int\nu u|^{2}}{\left\Vert u\right\Vert _{H^{1}}^{2}},\label{eq:def of dual norm}
\end{equation}
 which, by definition, is finite precisely on $H^{-1}(X)$ (if $\left\Vert \nu\right\Vert _{H^{-1}(X)}^{2}<\infty$
the sup may as well be taken over $H^{1}(X)).$ It follows from basic
Sobolev space theory that the operator $dd^{c}$ induces a canonical
isometry between $H^{1}(X)/\R$ and $H_{0}^{-1}(X):$ 
\begin{equation}
u\mapsto\nu:=dd^{c}u,\,\,\,H^{1}(X)\rightarrow H_{0}^{-1}(X)\label{eq:ddc is isometry}
\end{equation}
We will write $\nu\mapsto u_{\nu}$ for the inverse map. In particular,
if $\left\Vert \nu\right\Vert _{H^{-1}(X)}^{2}<\infty,$ then 
\[
\left\Vert \nu\right\Vert _{H^{-1}(X)}^{2}=\int du_{\nu}\wedge d^{c}u_{\nu}
\]
\begin{prop}
\label{prop:prop of energy when =00005Cphi is cont}Let $\phi$ be
a continuous metric on $L.$

\begin{itemize}
\item Assume that $\mu\in\mathcal{P}(X).$ Then $E_{\phi}(\mu)<\infty$
iff $\mu-\mu_{0}\in H^{-1}(X)$ and then 
\[
E_{\phi}(\mu)=E_{\psi_{0}}(\mu)+\left\langle \phi-\psi_{0},\mu\right\rangle =\frac{1}{2}\left\Vert \mu-\mu_{0}\right\Vert _{H^{-1}(X)}^{2}+\left\langle \phi-\psi_{0},\mu\right\rangle 
\]
In particular, if $\mu\in\mathcal{P}(X)$ then 
\begin{equation}
E_{\psi_{0}}(\mu)=\frac{1}{2}\left\Vert \mu-\mu_{0}\right\Vert _{H^{-1}(X)}^{2}=\frac{1}{2}\int G(x,y)\mu\otimes\mu,\label{eq:expressions for E psi not}
\end{equation}
 where $G$ is the integral kernel of the ``positive Laplacian''
$-dd^{c}/\mu_{0}$ restricted to the orthogonal complement of the
constants in $L^{2}(X,\mu_{0}),$ i.e. $G$ is the Green function
determined by the property that $G$ is symmetric and 
\begin{equation}
(i)\,-dd_{x}^{c}G(x,y)=\delta_{x}-\mu_{0},\,\,\,(ii)\,\int_{X}G(x,y)\mu_{0}(y)=0\label{eq:def of Green f}
\end{equation}
\item If $u\in H^{1}(X)$ is normalized so that $\int u\mu_{0}=0,$ then
\[
\left|\int_{X}u\mu\right|\leq\left\Vert u\right\Vert _{H^{1}}E_{\psi_{0}}(\mu)
\]
and $u$ acts continuously on the sublevelsets $\{\mu\in\mathcal{P}(X):\,E_{\psi_{0}}(\mu)\leq C\}$
endowed with the weak topology on $\mathcal{P}(X).$
\item The following formula holds for any $u\in C^{0}(X):$
\begin{equation}
\mathcal{F}_{\phi}(u)=\inf_{\mu\in\mathcal{P}(X)\cap H^{-1}(X)}\left(E_{\phi}(\mu)+\left\langle u,\mu\right\rangle \right)\label{eq:expression for leg transform in term sof H^1-1}
\end{equation}
and the right hand side above is finite for any $u\in H^{1}(X).$
\end{itemize}
\end{prop}

\begin{proof}
First assume that $\mu$ is a volume form in $\mathcal{P}(X).$ Then
\[
E_{\psi_{0}}(\mu)=\mathcal{E}(\psi_{\mu})-\left\langle \psi_{\mu}-\psi_{0},\mu\right\rangle ,
\]
 where $\psi_{\mu}\in PSH(L)$ solves $dd^{c}\psi_{\mu}=\mu$ (as
follows, for example, from the concavity of the functional $\mathcal{E}$;
see \cite{bbgz}). Using formula \ref{eq:energy functional in terms of Dirichlet}
this means that 
\[
E_{\psi_{0}}(\mu)=-J(\psi_{\mu}-\psi_{0})+\left\langle \psi_{\mu}-\psi_{0},\mu-\mu_{0}\right\rangle =J(\psi_{\mu}-\psi_{0}),
\]
 using integration by parts in the last equality. This proves formula
\ref{eq:expressions for E psi not} when $\mu$ is a volume form.
In the general case we take a sequence of smooth metrics $\psi_{j}\in PSH(L)$
decreasing to $\psi_{\mu}$ and set $\mu_{j}:=dd^{c}\psi_{j}.$ Letting
$j\rightarrow\infty$ and using the previous case then proves formula
\ref{eq:expressions for E psi not} in the case when $\phi=\psi_{0},$
using that both $E$ and $J$ are continuous along such sequences
(as follows from Prop \ref{prop:cont of beatiful E under monotone}).The
case of a general continuous $\phi$ then follows immediately by linearity,
which concludes the proof of the first point in the proposition. The
inequality in the second point now follows immediately from replacing
$\mu$ with the signed measure $\nu:=\mu-\mu_{0}\in H^{-1}(X)$ and
using that $|\left\langle u,\nu\right\rangle |\leq\left\Vert u\right\Vert _{H^{1}}\left\Vert \nu\right\Vert _{H^{-1}},$
since $H^{-1}(X)$ is the Hilbert space dual of $H^{1}(X).$ Moreover,
if $E_{\psi_{0}}(\mu_{j})\leq C$ and $\mu_{j}\rightarrow\mu$ weakly
in $\mathcal{P}(X).$ Then, by the compactness of the unit-ball in
$H^{-1}(X)$ endowed with the weak (Hilbert) topology, $\mu_{j}-\mu_{0}\rightarrow\mu-\mu_{0}$
in the weak topology on $H^{-1}(X)$ and hence $\left\langle u,\mu_{j}-\mu_{0}\right\rangle \rightarrow\left\langle u,\mu-\mu_{0}\right\rangle ,$
concluding the proof of the second point. The last point now follows
immediately from the definitions.
\end{proof}
Next we turn to the singular setting where $\phi$ is an admissible
metric on $L$ and $u\in H^{1}(X).$ We then define $E_{\phi+u}(\mu)$
by imposing linearity wrt $u:$
\begin{defn}
Let $\phi$ be an admissible metric on $L$ and $u\in H^{1}(X).$
Then 
\begin{equation}
E_{\phi+u}(\mu):=E_{\psi_{0}}(\mu)+\left\langle (\phi-\psi_{0}),\mu\right\rangle +\left\langle u,\mu\right\rangle \label{eq:def of E in sing setting}
\end{equation}
 when $\mu\in\mathcal{P}(X)$ has finite energy and otherwise we set
$E_{\phi+u}(\mu)=\infty.$ We also define $\mathcal{F}_{\phi}(u)$
by formula \ref{eq:expression for leg transform in term sof H^1-1}.

Accordingly, we define the equilibrium measure $\mu_{\phi+u}$ as
the unique minimizer of $E_{\phi+u}$ on $\mathcal{P}(X).$ The existence
and uniqueness is guaranteed by the first point in the next theorem. 
\end{defn}

\begin{thm}
\label{thm:conv of E wrt phi and u}Let $\phi$ be an admissible metric
on $L$ and $u\in H^{1}(X).$ Then

\begin{itemize}
\item The functional $\mu\mapsto E_{\phi+u}(\mu)$ is lsc and convex on
$\mathcal{P}(X)$ and strictly convex on the set where it is finite. 
\item Let $\phi_{N}$ be a sequence of admissible metrics on $L$ increasing
to an admissible metric $\phi$ and $u_{N}$ a sequence in $H^{1}$
converging to $u\in H^{1}$ in the $H^{1}-$topology. Then $\mathcal{F}_{\phi_{N}}(u_{N})\rightarrow\mathcal{F}_{\phi}(u).$ 
\end{itemize}
\end{thm}

\begin{proof}
Without loss of generality we may as well assume that $u_{N}$ and
$u$ are normalized in the sense that their integral against $\mu_{\psi_{0}}$
vanishes. Set $v_{N}:=\phi_{N}-\psi_{0},$ defining a sequence of
increasing lsc functions on $X.$ Let us first observe that the functionals
\[
I_{N}(\mu):=(E_{\phi_{N}+u_{N}})(\mu):=(E_{\psi_{0}})(\mu)+\int_{X}v_{N}\mu+\int_{X}u_{N}\mu
\]
 on $\mathcal{P}(X)$ are uniformly coercive on $H^{-1}(X)$ in the
sense that there exist positive constant $\epsilon$ and $C$ such
that 
\begin{equation}
I_{N}(\mu)\geq\epsilon E_{\psi_{0}}(\mu)-C.\label{eq:unif coerc}
\end{equation}
 Indeed, since $v_{1}$ is lower-semicontinuous there is a constant
$C_{1}$ such that 
\[
\int_{X}v_{N}\mu\geq\int_{X}v_{1}\mu\geq-C_{1}
\]
 by the weak compactness of $\mathcal{P}(X).$ Moreover, using the
Cauchy-Schwartz inequality on normalized functions in $H^{1}$ gives
(by the second point in the previous proposition) 
\[
\left|\int_{X}u_{N}\mu\right|\leq\left\Vert u_{N}\right\Vert _{H^{1}}E_{\psi_{0}}(\mu)^{1/2}\leq C_{2}E_{\psi_{0}}(\mu)^{1/2}
\]
Combining the previous two inequalities thus proves the uniform coercivity
\ref{eq:unif coerc}. The same argument shows that the functional
$I:=E_{\phi+u}$ is coercive on on $H^{-1}(X):$ 
\[
I(\mu)\geq\epsilon E_{\psi_{0}}(\mu)-C
\]
Now, to prove the first point in the lemma we let $\mu_{j}$ be a
sequence converging to $\mu$ weakly in $\mathcal{P}(X).$ We may
as well assume that $I(\mu_{j})$ is uniformly bounded from above
(otherwise there is nothing to prove). But then the previous coercivity
bound implies that $E_{\psi_{0}}(\mu_{j})$ is uniformly bounded from
above, i.e. $\left\Vert \mu_{j}-\mu_{0}\right\Vert _{H^{-1}(X)}$
is uniformly bounded from above (using the previous proposition).
By the compactness of the unit-ball in the $H^{-1}(X)$ wrt the weak
topology we may thus assume that the convergence of $\mu_{j}$ holds
in the weak topology on $H^{-1}(X).$ In particular, 
\begin{equation}
\lim_{j\rightarrow\infty}\left\langle u,\mu_{j}\right\rangle =\left\langle u,\mu\right\rangle \label{eq:conv of pairings in proof}
\end{equation}
Moreover, since $v_{N}$ is increasing we have that $v_{N}\geq v_{N_{0}}$for
any fixed index $N_{0}$ for $N\geq N_{0}$ and hence $I(\mu)\geq(E_{\phi_{N_{0}}+u})(\mu)$
so that 
\[
(E_{\phi_{N_{0}}+u})(\mu)\leq\liminf_{j\rightarrow\infty}I(\mu_{j})
\]
 using \ref{eq:conv of pairings in proof} and the lower semi-continuity
of $v_{N_{0}}.$ Finally, letting $N_{0}\rightarrow\infty$ and using
the monotone convergence theorem of integration theory concludes the
proof of the lower semi-continuity stated in the first point in the
lemma. Next, in order to establish the convexity of $E_{\phi+u}$
on $\mathcal{P}(X)$ it is, by linearity, enough to consider the case
when $\phi=\psi_{0}$ and $u=0.$ But, by definition, the functional
$E_{\psi_{0}}$ on $\mathcal{P}(X)$ is the restriction to $\mathcal{P}(X)$
of the Legendre-Fenchel transform of the functional $-(\mathcal{E}\circ P)(\phi-\cdot)$
which is Gateaux differentiable on $C^{0}(X)$ (by Prop \ref{prop:diff of composed energy}).
In particular, $E_{\psi_{0}}$is convex and moreover its strict convexity,
on the subset where it is finite, as follows from basic convex duality.

To prove the second point we fix a sequence $\mu_{N}$ realizing the
inf of $I_{N}$ on $\mathcal{P}(X)$ (whose existence follows from
the lsc of $I_{N}$ and the compactness of $\mathcal{P}(X)).$ By
construction 
\[
I_{N}(\mu_{N}):=\inf_{\mu}I_{N}\leq I_{N}(\mu)
\]
 for any fixed measure $\mu$ on $X.$ Hence, using the Cauchy-Schwartz
inequality and the monotone convergence theorem again and then taking
the inf over all $\mu$ in $\mathcal{P}(X)$ we get
\[
\limsup_{N\rightarrow\infty}I_{N}(\mu_{N})\leq\inf_{\mu\in\mathcal{P}(X)}I(\mu)
\]
In particular, by the uniform coercivity \ref{eq:unif coerc} we have
that $E_{\psi_{0}}(\mu_{N})\leq C'.$ Since $u_{N}$ is assumed to
converges to $u$ strongly in $H^{1}$ the Cauchy-Schwartz inequality
thus gives 
\[
\lim_{N\rightarrow\infty}\left\langle u_{N},\mu_{N}\right\rangle =\lim_{N\rightarrow\infty}\left\langle u,\mu_{\infty}\right\rangle .
\]
 Hence we get, using that $v_{N}\geq v_{N_{0}}$ as before, that 
\[
(E_{\phi_{N_{0}}+u})(\mu_{\infty})\leq\lim\inf_{N\rightarrow\infty}I_{N}(\mu_{N})
\]
for any fixed index $N_{0}.$ Finally, letting $N_{0}\rightarrow\infty$
gives 
\[
I(\mu_{\infty})\leq\lim\inf_{N\rightarrow\infty}I_{N}(\mu_{N})\leq\limsup_{N\rightarrow\infty}I_{N}(\mu_{N})\leq\inf_{\mu\in\mathcal{P}(X)}I(\mu)
\]
But then all inequalities above have to be equalities showing that
\[
\lim_{N\rightarrow\infty}\inf_{\mu\in\mathcal{P}(X)}I_{N}=\inf_{\mu\in\mathcal{P}(X)}I(\mu),
\]
 which concludes the proof of the second point. 
\end{proof}

\section{\label{sec:beta ensembles-on-Riemann}$\beta-$ensembles on Riemann
surfaces}

In this section we start by specializing the general setup of determinantal
point processes on a polarized complex manifold $X$ introduced in
\cite{berm2,berm1} to the case when $X$ is a compact Riemann surface.
From place to place we will use the following regularity assumptions
on a given singular metric $\phi$ on $L,$ where $S$ denotes the
support of the equilibrium measure $\mu_{\phi}$ and $dV$ is a fixed
volume form on $X:$ 
\begin{itemize}
\item \textbf{(A0)} \emph{$\phi$ is continuous on the complement of a closed
polar subset. }
\item \textbf{(A1)}\emph{ $A_{0}$ holds and there exists a constant $C$
such that $dd^{c}\phi\leq CdV$ on a neighborhood of $S$ }
\end{itemize}

\subsection{\label{subsec:Determinantal-point-process}Determinantal point process
on polarized Riemann surfaces and $\beta-$ensembles}

Let $L\rightarrow X$ be a positive line bundle over compact Riemann
surface $X.$ The geometric data $(dV,\phi)$ of a continuous volume
form $dV$ on $X$ and an admissible metric $\phi$ on $L$ induces
a determinantal random point process on $X$ with $N:=N_{L}$ particles,
defined by the following symmetric probability measure on the $N-$fold
product $X^{N}:$ 
\[
d\mathbb{P}_{N,L}:=\frac{|\Psi(x_{1},...,x_{N})|^{2}e^{-\left(\phi(x_{1})+...+\phi(x_{N})\right)}}{Z_{N}[dV,\phi]}dV{}^{\otimes N},
\]
 where $\Psi(x_{1},...,x_{N})$ is a generator of the top exterior
power $\Lambda^{N}H^{0}(X,L),$ viewed as a one-dimensional subspace
of $H^{0}(X^{N},L{}^{\boxtimes N})$ under the usual isomorphism between
$H^{0}(X^{N},L{}^{\boxtimes N})$ and the $N$ fold tensor product
of $H^{0}(X,L).$ The number $Z_{N}[dV,\phi],$ called the \emph{partition
function}, is the normalizing constant ensuring that $d\mathbb{P}_{N,L}$
is a probability measure:
\[
Z_{N}:=Z_{N}[dV,\phi]:=\int_{X^{N}}|\Psi|_{\phi}^{2}dV^{\otimes N},\,\,\,|\Psi|_{\phi}^{2}:=|\Psi(x_{1},...,x_{N})|^{2}e^{-\left(\phi(x_{1})+...+\phi(x_{N})\right)}
\]
Since $\Psi(x_{1},...,x_{N})$ is determined up to multiplication
by a non-zero complex number the probability measure $d\mathbb{P}_{N,L}$
is canonically attached to the geometric data $(dV,\phi).$ Concretely,
fixing a basis $\Psi_{1},...,\Psi_{N}$ in $H^{0}(X,L)$ we can (and
will) take $\Psi(x_{1},...,x_{N})$ to be 
\begin{equation}
\Psi(x_{1},...,x_{N}):=\det(\Psi_{i}(x_{j})):=\sum_{\sigma\in S_{N}}(-1)^{\mbox{sign\ensuremath{(\sigma)}}}\Psi_{1}(x_{\sigma(1)})\cdots\Psi_{N}(x_{\sigma(N)})\label{eq:def of det S}
\end{equation}
which we will call the \emph{Slater determinant}, following the terminology
used in the physics literature. 
\begin{rem}
If $\widetilde{\Psi}_{i}$ denotes another basis in in $H^{0}(X,L)$
then 
\end{rem}

\begin{equation}
\widetilde{\Psi}(x_{1},...,x_{N})=C\Psi(x_{1},...,x_{N})\,\,\,C\in\C^{*}\label{eq:scaling relation for slater}
\end{equation}
 where $C$ is the determinant of the corresponding change of basis
matrix. 

In the adjoint setting (Section \ref{subsec:The-adjoint-setting})
the probability measure $d\mathbb{P}_{N,L}$ only depends on $\phi$
and can be expressed as 
\[
d\mathbb{P}_{N,L}=(\frac{i}{2})^{N}\frac{\Psi(x_{1},...,x_{N})\wedge\overline{\Psi(x_{1},...,x_{N})}e^{-\left(\phi(x_{1})+...+\phi(x_{N})\right)}}{Z_{N}},
\]
 by identifying $\Psi$ with a holomorphic top form on $X^{N}$ with
values in $L^{\boxtimes N}$

\subsubsection{$\beta-$ensembles and Quantum Hall states on $X$}

Given a pair $(dV,\phi)$ as above and $\beta>0$ and the corresponding\emph{
$\beta-$ensemble }is the random point process with $N_{L}$ particles
on $X$ corresponding to the following symmetric probability measure
$\mu_{\beta}^{(N)}$ on $X^{N}:$ 
\begin{equation}
d\mathbb{P}_{N,L\beta}:=\frac{|\Psi|_{\phi}^{2\beta}}{Z_{N,\beta}[dV,\phi]}dV{}^{\otimes N},\,\,\,Z_{N,\beta}[dV,\phi]:=\int_{X^{N}}|\Psi|_{\phi}^{2\beta}dV^{\otimes N}\label{eq:Gibbs measure for beta ens text}
\end{equation}
 As recalled in Section \ref{subsec:Comp QH} for $\beta=1$ this
point processes represents an integer Quantum Hall state, while $\beta=m>0$
represents a fractional Quantum Hall state with ``filling fraction''
$1/m.$ Moreover, the curvature form $dd^{c}\phi$ coincides with
the magnetic two-form corresponding to an exterior magnetic field. 

\subsubsection{\label{subsec:The-asymptotic-setting of seq beta ens}The asymptotic
setting of a sequence of $\beta-$ensembles on $X$}

Assume given $(\phi,dV,\phi_{F}),$ where $\phi$ is an admissible
metric on the positive line bundle $L,$ $dV$ is a continuous volume
form on $X$ and $\phi_{F}$ is a continuous metric on a line bundle
$F.$ To this data we attach, for any given $\beta>0,$ a sequence
of $\beta-$ensembles obtained by replacing $(L,\phi)$ in the previous
section with $(L_{k},\phi_{k}):=(kL+F,k\phi+\phi_{F}).$ 

We then take $\Psi^{(k)}(x_{1},...,x_{N_{k}})$ to be the Slater determinant
\ref{eq:def of det S} defined wrt the reference bases $\Psi_{1}^{(k)},...\Psi_{N_{k}}^{(k)}$
in $H^{0}(X,kL+F).$ 
\begin{lem}
\label{lem:beta-ens as Gibbs measure}The probability measure $d\mathbb{P}_{N_{k},\beta}$
may be expressed as the following Gibbs measure: 
\[
d\mathbb{P}_{N_{k},\beta}=\frac{e^{-\beta H_{\phi}^{(N_{k})}}dV^{\otimes N_{k}}}{\int_{X^{N_{k}}}e^{-\beta H_{\phi}^{(N_{k})}}dV^{\otimes N_{k}}},\,\,\,H_{\phi}^{(N)}=H_{\psi_{0}}^{(N)}+k\sum_{i=1}^{N_{k}}(\phi-\psi_{0})(x_{i}),
\]
 where $H_{\psi_{0}}^{(N)}(x_{1},...,x_{N_{k}})=-\log|\Psi^{(k)}(x_{1},...,x_{N_{k}})|^{2}e^{-k(\psi_{0}+\phi_{F})}+\log N_{k}!$
(using the reference $\phi_{F_{0}}=\phi_{F}$ on $F$ in the definition
of $\Psi^{(k)}).$
\end{lem}

\begin{proof}
This follows directly from the basic observation that, in general,
\begin{equation}
\int_{X^{N}}|\det(\Psi_{i}(x_{j}))|^{2}e^{-\phi}dV=N!\label{eq:integral of det is N factorial}
\end{equation}
 if $\Psi_{1},...,\Psi_{N}$ is a bases in $H^{0}(X,L)$ which is
orthonormal wrt the $L^{2}-$norm induced by $(\phi,dV).$ Indeed,
this follows from integrating the sum in formula \ref{eq:def of det S}
and using that all $N!$ terms are equal to one, by permutation symmetry.
\end{proof}
\begin{rem}
In the literature on the Quantum Hall effect the ``twisting bundle''
$F$ is usually taken as $sK_{X}$ for $s$ a given non-negative number
in $\Z/2$ (when $s$ is not an integer the definition of $sK_{X}$
depends on the choice of a spin structure on $X$ i.e. a square root
of $K_{X}$\cite{k-m-m-w,kl}).
\end{rem}

\subsubsection{\label{subsec:Quasi-holes-and-magnetic}Quasi-holes, magnetic impurities
and Hele-Shaw flow}

Given a polarized Riemann surface $(X,L)$ and $(dV,\phi)$ as above
fix the additional data of an effective divisor $Z$ on $X,$ i.e.
$Z$ is the formal weighted sum $\sum_{i=1}^{m}a_{i}P_{i}$ for given
points $P_{i}$ in $X$ and positive integers $a_{i}>0$. Following
\cite[Secion 5.5]{berm1} another determinantal point process is obtained
by replacing the space $H(X,L)$ in section \ref{subsec:Determinantal-point-process}
with its sub-space $H_{Z}$ of sections vanishing to order at least
$a_{i}$ at the points $P_{i}.$ Denote by dimension $N$ the dimension
of $H_{Z}$ and by $\mu_{(L,Z)}^{(N)}$ the corresponding probability
measure on $X^{N}.$ In Laughlin's terminology \cite{Lau} $\mu_{(L,Z)}^{(N)}$
represents a \emph{quasi-hole} excited Quantum Hall state (see also
\cite[Section 9]{l-c-w}).\emph{ }In the asymptotic situation $L$
is replaced by $kL$ (or more generally by $kL+F)$ and $Z$ by $kZ,$
assuming that $\deg(L)>\deg(Z):=\sum_{i}a_{i}$ (which ensures that
$N_{k}\rightarrow\infty$ as $k\rightarrow\infty).$ 

The probability measure $\mu_{(L,Z)}^{(N)}$ can also be described
from an alternative point of view. Denote by $L_{Z}$ the line bundle
on $X$ determined by $Z$ and by $\phi_{Z}$ the corresponding singular
metric on $L$ (which is uniquely determined up to scaling). This
means that there exists $s_{Z}\in H(X,L_{Z})$ vanishing to order
$a_{i}$ at $P_{i}$ and $\phi_{Z}$ is the metric locally represented
by the weight $\log|s_{Z}|^{2}.$ As explained in \cite[Secion 5.5]{berm1}
the probability measure $\mu_{(L,Z)}^{(N)}$ coincides with the probability
measure $\mu_{(L-L_{Z})}^{(N)}$ obtained by replacing the line bundle
$L$ in Section \ref{subsec:Determinantal-point-process} with the
line bundle $L-L_{Z}$ and $\phi$ with the metric $\phi-\phi_{Z}$
on $L-L_{Z}.$ Since the curvature form of $\phi-\phi_{Z}$ is given
by $dd^{c}\phi-\sum a_{i}\delta_{p_{i}}$ the corresponding point
process represents a Quantum Hall state on $X$ subject to an external
magnetic field $dd^{c}\phi$ and a singular magnetic field of opposite
direction created by \emph{magnetic impurities} (or vortices) of charge
$a_{i}$ inserted at the points $P_{i}.$ The analytic advantage of
this point of view is that $\phi-\phi_{Z}$ is a lsc metric which
is continuous on the complement of a closed polar set so that the
all the results proven in the present paper apply. Moreover, it follows
readily from the definitions that the support $S_{Z}$ of the corresponding
equilibrium measure $\mu_{\phi-\phi_{Z}}$ is compactly supported
in the complement of $Z.$ Hence, by Prop \ref{prop:regularity},
the corresponding equilibrium measure is given by
\[
\mu_{\phi-\phi_{Z}}=1_{S_{Z}}dd^{c}\phi
\]
and if $\phi$ is smooth then the assumption A1 holds. 

For a given positive integer $k$ the previous setup applies more
generally when the coefficients $a_{i}$ of $Z$ are taken to be non-negative
real numbers. Then $H_{kZ}$ is defined by sections vanishing at the
points $P_{i}$ to order at least the round down of $ka_{i}.$ It
follows from \cite[Thm 1.10]{r-n} that if the coefficients $a_{i}$
are sufficiently small then $X-S$ is diffeomorphic to a disjoint
union of discs centered at the points $P_{i}.$ However, when $a_{i}$
is increased singularities of $S$ are typically formed and the topology
of $S$ changes. Indeed, fixing a divisor $Z$ the map $t\mapsto X-S_{tZ}$
describes an evolution of increasing domains coinciding with the Hele-Shaw
flow on $X$ with sources at $p_{i}$ and injection rates $a_{i}$
\cite{h-s,r-n0}, which is known to typically become singular in a
finite time. 

\subsection{Upper bounds on partition functions }

The key ingredient in the upper bounds on the partition functions
is the following result from \cite{berm3} :
\begin{prop}
\label{prop:lower bound on L funct in adj} Let $\psi$ be a continuous
metric on $L$ with positive curvature current and $\psi_{0}$ a metric
on $L$ such that $dd^{c}\psi_{0}$ defines a Riemannian metric on
the Riemann surface $X$ with constant scalar curvature. Then 

\[
\frac{1}{N}\log Z_{N,1}[\psi]-\frac{1}{N}\log Z_{N,1}[\psi_{0}]\leq-\mathcal{E}_{L}(\psi)+\mathcal{E}_{L}(\psi_{0})-\delta_{L}\left(\mathcal{E}_{L}(\psi)-\mathcal{E}_{L}(\psi_{0})-\sup_{X}(\psi-\psi_{0})\right)
\]
where $\delta_{L}$ is a non-negative constant such that $\delta=0$
when $X$ is the Riemann sphere and in the case when $X$ has genus
at least one, $\delta_{kL}\leq Ce^{-k/C}$ for an explicit positive
constant $C.$ 
\end{prop}

\begin{proof}
This follows directly from combining \cite[Prop 2.1, Prop 2.2]{berm3}.
\end{proof}
We also have the following result, which follows from \cite{b-b,berm1}: 
\begin{prop}
\label{prop:-asympt of L func} Fix a continuous volume form $dV$
on $X$ and let $L$ and $F$ be line bundles over $X$ endowed with
continuous metrics $\phi$ and $\phi_{F},$ respectively and assume
that $L$ is positive. Then 
\[
\lim_{k\rightarrow\infty}\frac{1}{kN_{k}}\log Z_{N_{k},\beta}[(dV,k\phi+\phi_{F})]=-\mathcal{E}(P(\phi))=:\mathcal{F}(\phi)
\]
More generally, for any lower semi-continuous $\phi$ such that $\{\phi<\infty\}$
is not polar
\begin{equation}
\limsup_{k\rightarrow\infty}\frac{1}{kN_{k}}\log Z_{N_{k},\beta}(dV,k\phi+\phi_{F})\leq-\mathcal{E}(P(\phi))<\infty\label{eq:limif L in prop}
\end{equation}
\end{prop}

\begin{proof}
When $F$ is the trivial line bundle the first convergence in the
proposition is shown in \cite[Equation 4.23]{berm1} in any dimension
(the proof reduces to the case $\beta=1$ which is covered by \cite[Thm A]{b-b}).
The general case is proved in essentially the same way (moreover,
in the present Riemann surface case a different proof can be given;
see Theorem \ref{prop:Lower bd on part func riemann surf}). In the
case when $\phi$ is merely assumed to be lsc we take a sequence $\phi_{j}$
increasing to $\phi.$ Then $\log Z_{N,\beta}[(dV,k\phi+\phi_{F}]\leq\log Z_{N,\beta}[dV,k\phi_{j}+\phi_{F})]$
and hence the left hand side in formula \ref{eq:limif L in prop}
is bounded from above by $-\mathcal{E}(P(\phi_{j}))$ for any fixed
$j.$ The proof is thus concluded by letting letting $j\rightarrow\infty$
and using Lemma \ref{lem:prop of p phi} and Prop \ref{prop:cont of beatiful E under monotone}.
\end{proof}

\subsection{Lower bounds on partition functions and relative entropy}

We will denote by $D_{\mu_{0}}(\mu)$ the \emph{entropy} (aka the
information divergence) of a measure $\mu$ on $X$ relative a measure
$\mu_{0},$ i.e. 
\begin{equation}
D_{\mu_{0}}(\mu):=\int\log\frac{\mu}{\mu_{0}}\mu\label{eq:def of entropy}
\end{equation}
 if $\mu$ is absolutely continuous wrt $\mu_{0}$ and otherwise $D_{\mu_{0}}(\mu):=\infty.$
For the lower bound on the partition function we will use Gibbs variational
principle in the following form (which is a simple consequence of
Jensen's inequality):
\begin{lem}
\label{lem:Gibbs}Let $(X,dV)$ be a measure space and$H^{(N)}$ be
a symmetric function in $L^{1}(X^{N}).$ Then
\[
\frac{1}{\text{\ensuremath{\beta}}}\log\int e^{-\beta H^{(N)}}dV^{\otimes N}\geq-\int H^{(N)}\nu^{\otimes N}-\beta^{-1}D_{dV}(\nu)
\]
for any for any $\nu\in\mathcal{P}(X)$ such that the integrals above
are well-defined and finite. 
\end{lem}

We will also use the following lemma, which is elementary when $X$
is the Riemann sphere (the proof for general $X$ is given in the
companion paper \cite{berm17}):
\begin{lem}
\label{lem:asy of mean energy}Let $H_{\psi_{0}}^{(N)}$ be the Hamiltonian
defined in Lemma \ref{lem:beta-ens as Gibbs measure}. Then
\begin{equation}
\lim_{N_{k}\rightarrow\infty}k^{-1}N_{k}^{-1}\int_{X^{N}}H_{\psi_{0}}^{(N_{k})}\nu^{\otimes N_{k}}=E_{\psi_{0}}(\nu)\label{eq:exact mean energy on sphere}
\end{equation}
for any volume form $\nu$ in $\mathcal{P}(X).$ Moreover, there exists
a constant $C_{(X,L)}$ only depending on $X$ and $L$ such that
\[
-k^{-1}N_{k}^{-1}\int_{X^{N}}H_{\psi_{0}}^{(N_{k})}\nu^{\otimes N_{k}}\geq-E_{\psi_{0}}(\nu)-C_{(X,L)}\frac{1}{N_{k}}\log N_{k}
\]
where $C_{(X,L)}=1/2$ when $(X,L)=(\P^{1},\mathcal{O}(1)).$ 
\end{lem}

\begin{proof}
Consider the case when $X$ is the Riemann sphere and $L$ is the
line bundle of degree one and assume for simplicity that $F$ and
$\phi_{F}$ are trivial (a similar proof applies in the general case).
Then we have $N=k+1$ and the corresponding Hamiltonian may be expressed
as follows on $\C\subset X$ (using the identifications in Section
\ref{subsec:Compactification-of-}):
\begin{equation}
H_{\psi_{0}}^{(N)}(z_{1},...,z_{N})=\frac{1}{2}\sum_{i\neq j}w(z_{i},z_{j})+F(N+1)+\log N!\,\,\,\,-w(z,z'):=\log|z-z'|^{2}-\psi_{0}(z)-\psi_{0}(z')\label{eq:H in terms of w log in pf}
\end{equation}
(see Lemma \ref{lem:slater as vandermonde and adjoint}). Fubini's
theorem thus gives
\[
k^{-1}N_{k}^{-1}\int_{X^{N}}H^{(N_{k})}\nu^{\otimes N_{k}}=\frac{1}{2}\int_{X}w\nu^{\otimes2}+\frac{F(N+1)+\log N!}{N(N-1)}=E_{\psi_{0}}(\nu)+\left(C_{0}+\frac{F(N+1)+\log N!}{N(N-1)})\right)
\]
using that $(w-G):=2C_{0}$ is constant in the last equality (see
the proof of Lemma \ref{lem:difference energies plane sphere}). By
the asymptotic expansion in Lemma \ref{lem:def and exp of F} this
means that there exists a constant $C_{1}$ such that
\[
k^{-1}N_{k}^{-1}\int_{X^{N}}H^{(N_{k})}\nu^{\otimes N_{k}}=E_{\psi_{0}}(\nu)+C_{1}+\frac{1}{2}N^{-1}\log N+O(N^{-1})
\]
All that remains is to verify that
\begin{equation}
C_{1}=0\label{eq:constant C zero is}
\end{equation}
 This is essentially well-known and can be checked by explicit calculation,
but here we instead provide the following ``high-brow proof''. By
general results \cite{bermldpgibbs} about large deviation principles
for Hamiltonians $H_{W}^{(N)}$ of the form \ref{eq:Hamilt for W and V intro}
(with $V=0)$ 
\[
-\lim_{N\rightarrow\infty}\frac{1}{N^{2}}\log\int_{X}e^{-H_{W}^{(N)}}\mu_{0}^{\otimes N}=\inf_{\nu\in P(X)}E_{W}(\nu),\,\,\,\,E_{W}(\nu):=\lim_{N\rightarrow\infty}N^{-2}\int_{X^{N}}H_{W}^{(N)}\nu^{\otimes N}
\]
But in the present case $\int_{X}e^{-H_{W}^{(N)}}\mu_{0}^{\otimes N}=N!$
(by the orthonormality condition) and hence the infimum of the corresponding
functional $E_{W}$ vanishes. Since the infimum of $E_{\psi_{0}}$
vanishes this forces the equality \ref{eq:constant C zero is}. The
case when $X$ has genus at least one is shown in the companion paper
\cite{berm17}, using the bosonization formula as a replacement for
formula \ref{eq:H in terms of w log in pf}. 
\end{proof}
\begin{prop}
\label{prop:Lower bd on part func riemann surf}Let $(X,L)$ be a
polarized Riemann surface, $\phi$ a metric on $L$ satisfying $A0$
and $u\in H^{1}(X).$ Then, for any $\beta>0,$ 

\[
\liminf_{N\rightarrow\infty}\frac{1}{\beta kN_{k}}\log Z_{N_{k},\beta}[dV,k(\phi+u)]\geq-\mathcal{F}(\phi+u)
\]
\end{prop}

\begin{proof}
We need to show that 
\begin{equation}
\liminf_{N\rightarrow\infty}\frac{1}{\beta kN_{k}}\log Z_{N_{k},\beta}[\mu_{0},k(\phi+u)]\geq-E_{\phi+u}(\mu_{\phi+u}),\label{eq:lower bound in proof lower bound}
\end{equation}
 To prove \ref{eq:lower bd on part} we decompose $\phi=\psi_{0}+v$.
We can then write and
\begin{equation}
Z_{N}:=Z_{N_{k},\beta}[dV,k(\phi+u)]:=\int e^{-\beta H^{(N)}}\mu_{0}^{\otimes N},\,\,\,H^{(N)}:=H_{\psi_{0}}^{(N)}+\sum_{i=1}^{N}k(v(x_{i})+u(x_{i})),\,\,\,\label{eq:N particle energy sphere in proof ldp sing}
\end{equation}
 where $H_{\psi_{0}}^{(N)}$ is the Hamiltonian defined in Lemma \ref{lem:beta-ens as Gibbs measure}.
For any fixed volume form $\nu\in\mathcal{P}(X)$ combining Lemmas
\ref{lem:Gibbs} and \ref{lem:asy of mean energy} gives

\begin{equation}
\limsup_{N\rightarrow\infty}-\frac{1}{\beta kN}\log Z_{N}\leq\int_{X^{2}}E_{\psi_{0}}(\mu)+\int_{X}v\nu+\int_{X}u\nu=E_{\phi+u}(\nu)\label{eq:lower bd on part}
\end{equation}
 Next, we recall that for any measure $\mu\in\mathcal{P}(X)$ such
that $E_{\psi_{0}}(\mu)<\infty,$ i.e. $\mu\in H^{-1}(X),$ there
exists a family $\nu_{\epsilon}$ of volume forms in $\mathcal{P}(X)$
such that $\nu_{\epsilon}\rightarrow\mu$ in $H^{-1}(X)$ as $\epsilon\rightarrow0$
(this follows, for example, from properties of the heat-kernel, as
in \cite[Lemma 33]{z-z}). In particular, 
\[
E_{\psi_{0}}(\nu_{\epsilon})\rightarrow E_{\psi_{0}}(\mu),\,\,\,\int_{X}u\nu_{\epsilon}\rightarrow\int_{X}u\mu,
\]
 using that $u\in H^{1}(X).$ In case when when $v\in C(X)$ the corresponding
convergence trivially also holds when integrating against $v.$ Taking
the inf over all volume forms $\nu$ in the inequality \ref{eq:lower bd on part}
thus proves the desired lower bound on $Z_{N}$ in the case when $v\in C(X).$

In the case when $v$ is merely continuous on $U:=X-K,$ where $K$
is a compact polar subset of $X,$ the previous argument gives 
\[
\limsup_{N\rightarrow\infty}-\frac{1}{\beta kN}\log Z_{N}\leq\inf_{\nu\in\mathcal{P}(U)\cap H^{-1}}\int_{X^{2}}E_{\psi_{0}}(\mu)+\int_{X}v\nu+\int_{X}u\nu.
\]
 But since $K$ is polar we have that $\mu(K)=0$ for any measure
such that $E_{\psi_{0}}(\mu)<\infty.$ In particular, this is the
case for the the minimizer $\mu_{\phi+u}$ of $E_{\phi+u}$ and we
can thus identify $\mu_{\phi+u}$ with an element in $\mathcal{P}(U)$
showing that the inf in the right hand side in formula \ref{eq:lower bd on part}
may as well be taken over $\mathcal{P}(X)\cap H^{-1}(X).$ 
\end{proof}
\begin{rem}
\label{rem:contr to general lsc phi}Even if the equilibrium measure
$\mu_{\phi}$ is well-defined for a general lsc $\phi$ with the property
that $\{\phi<\infty\}$ is non-polar, the lower bound in the previous
theorem does not hold in this generality, as opposed to the upper
bound \ref{eq:limif L in prop} (a concrete example is provided in
\cite{berm15}). 
\end{rem}

\subsection{\label{subsec:Main-results-for beta ens on X}Main results for $\beta-$ensembles }

In this section we fix the geometric data $(\phi,dV,\phi_{F})$ as
in Section \ref{subsec:The-asymptotic-setting of seq beta ens} and
consider, for a given $\beta\in]0,1]$ the corresponding sequence
of $\beta-$ensembles on $X.$ Given a function $u\in H^{1}(X)$ we
will denote by $U^{N_{k}}$ the corresponding random variable on the
$\beta-$ensemble with $N_{k}$ particles, defined by formula \ref{eq:def of linear stat}.
We will obtain bounds on the corresponding moment generating function
$\E(e^{-\beta kN_{k}U_{N_{k}}}),$ involving the ``error sequence''
defined by 

\begin{equation}
\epsilon_{N_{k},\beta}[\phi,dV,\phi_{F}]:=-\frac{1}{\beta N_{k}k}\log Z_{N_{k},\beta}[dV,k\phi+\phi_{F}]-\mathcal{F}(\phi)\label{eq:def of epsion error in ineq}
\end{equation}
(where the dependence on $dV$ and $\phi_{F}$ has been suppressed
in the left hand side). We recall that in the adjoint setting $\epsilon_{N_{k},\beta}[\phi,dV,\phi_{F}]$
only depends on the metric $\phi$ on $L$ and $\beta=1.$ Accordingly,
the corresponding error sequence will then be denoted by $\epsilon_{N_{k}}[\phi].$
It follows directly from the definition $\epsilon_{N_{k},\beta}[\phi]=0$
if $\phi=\psi_{0},$ i.e. if the curvature form $\omega^{\phi}$ defines
a metric on $X$ with constant scalar curvature. In the statements
of the main results below below we will suppress the explicit dependence
on the geometric data from the notation for the corresponding error
sequences. 

The starting point of the proofs is the basic formula

\[
\E(e^{-\beta kN_{k}U_{N_{k}}})=\frac{Z_{N_{k},\beta}[dV,k(\phi+u)+\phi_{F}]}{Z_{N_{k},\beta}[dV,k\phi+\phi_{F}]}
\]

We start with the adjoint setting, which gives the cleanest inequality:
\begin{thm}
\label{thm:upper bd on part funct riem surf adj}If $\phi$ is an
admissible metric on $L$ and the following inequalities hold in the
adjoint setting (in particular, $\beta=1):$
\[
\frac{1}{kN_{k}}\log\E(e^{-kN_{k}U_{N_{k}}})\leq-\mathcal{F}(\phi+u)+\mathcal{F}(\phi)+\epsilon_{N_{k}}+\delta_{k}(J(u)+C_{0})
\]
where $\delta_{kL}$ is the exponentially small sequence appearing
in Prop \ref{prop:lower bound on L funct in adj} (which only depends
on $(X,kL)$ and vanishes if $X$ is the Riemann sphere) and

\[
C_{0}:=\sup_{X\times X}(-G_{0})
\]
where $G_{0}$ is the Green function of the canonical metric on $X.$
\end{thm}

\begin{proof}
First assume that $\phi$ and $u$ are smooth. Given a metric $\Phi$
on $L$ we we will write 
\[
\mathcal{L}_{k}(\Phi)=-\frac{1}{kN_{k}}\log Z_{N,1}[k\Phi]
\]
 in terms of the partition function appearing in the adjoint setting.

\emph{Step one:} let $\delta_{k}:=\delta(X,kL)$ be as in the statement
of Prop \ref{prop:lower bound on L funct in adj}, then

\[
\mathcal{L}_{k}(\phi+u)-\mathcal{L}_{k}(\phi)\geq\left(-\mathcal{E}(\psi)-\mathcal{E}(\psi_{0})\right)-\left(\mathcal{E}(P\phi)-\mathcal{E}(\psi_{0})\right)-\delta_{k}R(\phi,u),
\]
\[
R(\phi,u):=\mathcal{E}_{L}(\psi_{0})-\mathcal{E}_{L}(P(\phi+u))+\sup_{X}(P(\phi+u)-\psi_{0})
\]

The inequality to be proved may be rewritten as 
\begin{equation}
k^{-1}\mathcal{L}_{kL+K_{X}}(k\phi+u)-k^{-1}\mathcal{L}_{kL+K_{X}}(\phi)\geq-k^{-1}\mathcal{E}_{kL}(P(k(\phi+u)))+k^{-1}\mathcal{E}_{kL}(P(k\phi))-o(1),\label{eq:ineq in adj setting in proof of thm upper bound part}
\end{equation}
 using that $\mathcal{F}_{\phi}(u):=\mathcal{E}(P(\phi+u)$  and $k^{-1}\mathcal{E}_{kL}(k\psi)=\mathcal{E}_{L}(\psi),$
as follows immediately from the definition. Setting $\mathcal{L}_{k}(\psi):=k^{-1}\mathcal{L}_{kL+K_{X}}(k\psi)$
we rewriting the left hand side in the inequality above as 
\[
\mathcal{L}_{k}(\phi+u)-\mathcal{L}_{k}(\phi)=\left(\mathcal{L}_{k}(\phi+u)-\mathcal{L}_{k}(\psi_{0})\right)-\left(\mathcal{L}_{k}(\phi)-\mathcal{L}_{k}(\psi_{0})\right)
\]
Next, we observe that since
\[
\psi:=P(\phi+u)\leq\phi+u
\]
 and $\phi\mapsto\mathcal{L}_{k}(\phi)$ is increasing we get
\[
\mathcal{L}_{k}(\phi+u)-\mathcal{L}_{k}(\phi)\geq\left(\mathcal{L}_{k}(\psi)-\mathcal{L}_{k}(\psi_{0})\right)-\left(\mathcal{L}_{k}(\phi)-\mathcal{L}_{k}(\psi_{0})\right)
\]
Hence, applying the lower and upper bounds in Prop \ref{prop:lower bound on L funct in adj}
and Prop \ref{prop:-asympt of L func} to the first and second term
above, respectively, concludes the proof of Step one. 

Step two: the following inequality holds:
\[
R(\phi,u)\leq J(u)+C
\]
 for a constant $C$ only depending on $\phi.$ 

To prove this first note that it for any $v\in PSH(X,\mu_{0})$ 
\[
\sup_{X}v\leq\int_{X}v(dd^{c}P\phi)+C_{1},\,\,C_{1}:=-\sup_{X}(P\phi-\psi_{0})-\inf_{X}(P\phi-\psi_{0})+C_{0},
\]
as follows from Lemma \ref{lem:sup estimate}. In particular taking
$v:=P(\phi+u)-\psi_{0}$ gives 

\[
\sup_{X}(P(\phi+u)-\psi_{0})\leq\int_{X}P(\phi+u)-\psi_{0})dd^{c}(P\phi)+C_{1}\leq
\]
\[
\int_{X}P(\phi+u)-P\phi)dd^{c}(P\phi)+B_{\phi},\,\,\,C_{2}:=C_{0}+\inf_{X}(P\phi-\psi_{0})
\]
 Hence, 
\[
R(\phi,u)\leq\mathcal{E}_{L}(P\phi)-\mathcal{E}_{L}(P(\phi+u))+\int_{X}P(\phi+u)-P\phi)dd^{c}(P\phi)+C_{\phi},
\]
 where 
\[
C_{\phi}:=C_{0}+\inf_{X}(P\phi-\psi_{0})-\left(\mathcal{E}(P\phi)-\mathcal{E}(\psi_{0})\right)\leq C_{0},
\]
 using that, in general, if $\psi\in PSH(X,L)$ is locally bounded,
then it follows readily from the definition of the functional $\mathcal{E}$
that
\[
\inf_{X}(\psi-\psi_{0})\leq\mathcal{E}(\psi)-\mathcal{E}(\psi_{0})\leq\sup_{X}(\psi-\psi_{0})
\]

Now, using that $P(\phi+u)\leq\phi+u$ and applying the orthogonality
relation \ref{eq:og relation} to replace $P\phi$ in the second term
above with $\phi$ gives 
\[
R(\phi,u)\leq\mathcal{E}_{L}(P\phi)-\mathcal{E}_{L}(P(\phi+u))+\int_{X}udd^{c}(P\phi)+C_{0}
\]
 and applying Lemma \ref{lem:smaller than dirichlet} thus proves
the inequality \ref{eq:ineq in adj setting in proof of thm upper bound part}
in the case when $\phi$ and $u$ are smooth. The general case now
follows by approximation, writing $\phi$ as a decreasing limit of
smooth metrics and $u$ as a strong limit in $H^{1}(X)$ of smooth
functions (and then using Theorem \ref{thm:conv of E wrt phi and u}). 
\end{proof}
Next, we turn to the general setting in $\beta-$ensembles on polarized
Riemann surfaces. This can be viewed as a special case of the previous
setting, but with $L$ replaced by a $\Q-$line bundle $L_{k}$ depending
on $k.$ Accordingly, one can get an inequality as in the previous
theorem with $\phi$ replaced by $\phi_{k},$ coinciding with the
right hand side in the inequality in question, to the leading order.
Here we content ourselves with an asymptotic result which will be
sufficient for the applications in \cite{berm15}.
\begin{thm}
\label{thm:non-gaussian ineq general setting}Assume given an admissible
metrics $\phi$ on $L$ and a sequence $u^{(k)}\in H^{1}(X)$ such
that $u^{(k)}\rightarrow u$ strongly in $H^{1}(X).$ Denote by $U_{N_{k}}^{(k)}$
the corresponding normalized linear statistics. If $\beta\leq1,$
then 

\[
\limsup_{k\rightarrow\infty}\frac{1}{kN_{k}}\log\E(e^{-kN_{k}U_{N_{k}}^{(k)}})\leq-\mathcal{F}(\phi+u)+\mathcal{F}(\phi)
\]
\end{thm}

\begin{proof}
\emph{Step 1: The case $\beta=1$:}

We can rewrite 
\[
kL+F=kL_{k}+K_{X},\,\,\,L_{k}=L+k^{-1}(F-K_{X})
\]
 where $E=F-K_{X}$ (the equality above holds in the usual sense of
$\Q-$line bundles). Hence, we are in the previous adjoint setting
with $L$ replaced by $L_{k}$ and $\phi$ replaced by 
\[
\phi_{k}:=\phi+k^{-1}\phi_{E},
\]
 where $\phi_{E}$ depends on $\phi_{F}$ and the volume form $dV.$
We thus get 
\[
\frac{1}{kN_{k}}\log\frac{Z_{N_{k}}[dV,k(\phi+u)]}{Z_{N_{k}}[dV,k\phi]}\leq-E_{\phi+u}(\mu_{\phi+u})+E_{\phi}(\mu_{\phi})+\epsilon_{k}+\delta_{k}J(u),
\]
 where $\epsilon_{k}\rightarrow0$ and $\delta_{k}$ is exponentially
small (and vanishes when $X$ is the Riemann sphere). Letting $k\rightarrow\infty$
then concludes the proof of Step one (using a simple variant of Theorem
\ref{thm:conv of E wrt phi and u} to replace $\phi_{k}$ with $\phi$). 

\emph{Step 2: The case $\beta\leq1$}

By Jensen's inequality we have, when $\beta\leq1,$ that

\[
\frac{1}{\beta kN_{k}}\log\frac{Z_{N_{k},\beta}[dV,k(\phi+u)]}{Z_{N_{k}}[dV,k\phi]}\leq\frac{1}{kN_{k}}\log Z_{N_{k},1}[\mu_{0},k(\phi+u)]-\frac{1}{\beta kN_{k}}\log Z_{N_{k}}[\mu_{0},k\phi]
\]
Hence, we can use the same upper bound on $Z_{N_{k},1}[\mu_{0},k(\phi+u)]$
as in the previous step and conclude exactly as before.

The next theorem gives a sub-Gaussian bound:
\end{proof}
\begin{thm}
\label{thm:gaussian ineq adj and general}Let $\phi$ be an admissible
metric on $L$ and $u\in H^{1}(X).$ In the adjoint setting with $\beta=1$
the following inequalities hold: 
\begin{equation}
\frac{1}{kN_{k}}\log\E(e^{kN_{k}(u-\bar{u})})\leq\frac{1}{2\deg L}\int_{X}du\wedge d^{c}u(1+\delta_{k})+\delta_{k}C_{0}+\epsilon_{N_{k}}\label{eq:gaussian ineq adjoint}
\end{equation}
For general $\beta-$ensembles the following inequality holds when
$\beta\leq1:$
\[
\frac{1}{\beta kN_{k}}\log\E(e^{\beta kN_{k}(u-\bar{u})})\leq\frac{1}{2\deg L}\int_{X}du\wedge d^{c}u(1+k^{-1}B)+k^{-1}B+\epsilon_{N_{k},\beta},
\]
 for a constant $B$ not depending on $u.$ 
\end{thm}

\begin{proof}
The first inequality follows directly from Theorem \ref{thm:upper bd on part funct riem surf adj}
combined with Lemma \ref{lem:smaller than dirichlet}. Next consider
the general setting. Using Jensen's inequality as in the proof of
Theorem \ref{thm:non-gaussian ineq general setting} it is enough
to consider the case when $\beta=1.$ Then we can apply Theorem \ref{thm:upper bd on part funct riem surf adj}
with $L$ replaced by $L_{k}$ (using the notation in the proof of
Theorem \ref{thm:non-gaussian ineq general setting}). This gives
the same right hand side as in the inequality \ref{eq:gaussian ineq adjoint}
up to an extra error term 
\[
\int\left(dd^{c}P(\phi_{k})-dd^{c}P(\phi)\right)u
\]
which comes from computing the mean $\bar{u}$ with respect to $\phi_{k}$
and $\phi,$ respectively. To estimate this term we rewrite $P(\phi_{k})=\left(P(\phi_{k})-k^{-1}\phi_{E}\right)+k^{-1}\phi_{E}$
so that the term in question may be expressed as 
\begin{equation}
k^{-1}\left(\int dd^{c}v_{k}u+\int udd^{c}\phi_{E}\right),\,\,\,v_{k}:=k\left(\left(P(\phi_{k})-k^{-1}\phi_{E}\right)-P(\phi)\right)\label{eq:pf of Gaussian ineq adj}
\end{equation}
Since $v_{k}$ defines a bona fide function on $X$ integrating by
parts and using the Cauchy-Schwartz inequality and the inequality
$ab\leq a^{2}/2+b^{2}/2$ gives
\[
\left|\int dd^{c}v_{k}u\right|\leq\frac{1}{2}\int dv_{k}\wedge d^{c}v_{k}+\frac{1}{2}\int_{X}du\wedge d^{c}u
\]
Now decompose
\[
-\int dv_{k}\wedge d^{c}v_{k}=\int v_{k}dd^{c}v_{k}=\int v_{k}dd^{c}P(\phi_{k})-\int v_{k}dd^{c}P(\phi)-k^{-1}\int v_{k}dd^{c}\phi_{E}
\]
Using that $\phi_{E}$ is locally bounded it is not hard to see that
$|v_{k}|\leq C_{1}$ (depending on $\phi)$ and hence 
\[
\int dv_{k}\wedge d^{c}v_{k}\leq C_{1}2\deg L+k^{-1}C_{1}C_{2}
\]
Finally, to estimate the second term in the left hand side of formula
\ref{eq:pf of Gaussian ineq adj}, we first trivially estimate 
\[
\int udd^{c}\phi_{E}\leq C_{2}\int|u|\mu_{0}\leq C_{2}(\int u^{2}\mu_{0})^{1/2}.
\]
Next, by scaling invariance we may as well assume that $\int u\mu_{0}=0.$
Then 
\[
(\int u^{2}\mu_{0})^{1/2}\leq C_{3}^{1/2}(\int du\wedge d^{c}u)^{1/2}\leq C_{3}/2+\frac{1}{2}\int_{X}du\wedge d^{c}u,
\]
 where $1/C_{3}$ is the smallest strictly positive eigenvalue of
the ``positive Laplacian'' $\Delta:=-dd^{c}/\mu_{0},$ which concludes
the proof. 
\end{proof}
Finally, we show that the inequality in Theorem \ref{thm:non-gaussian ineq general setting}
is, in fact, an asymptotic equality, in the following sense: 
\begin{thm}
\label{thm:asymptot of log E on Riemann }Assume that $\phi$ satisfies
A0. If $u\in H^{1}(X)$ and $\beta\leq1$ then 

\[
\lim_{N\rightarrow\infty}\frac{1}{\beta kN_{k}}\log\E(e^{-kN_{k}U_{N_{k}}})=-\mathcal{F}(\phi+u)+\mathcal{F}(\phi)
\]
If $u\in C(X)$ then the convergence holds for any $\beta>0.$ 
\end{thm}

\begin{proof}
By Prop\ref{prop:Lower bd on part func riemann surf} the lower bound
holds for any $u\in H^{1}(X)$ and $\beta>0.$ Moreover, if if $u\in H^{1}(X)$
then the upper bound also holds, by Theorem \ref{thm:non-gaussian ineq general setting}.
Finally, if $u\in C(X)$ then $\Phi:=\phi+u$ is lsc and hence the
upper bound follows from the inequality \ref{eq:limif L in prop}.
As a consequence, if $u\in C(X)$ then the equality holds above as
a true limit (using the inequality \ref{eq:limif L in prop}) and
if $u\in H^{1}(X)$ this is the case if $\beta\leq1$ 
\end{proof}

\subsection{\label{subsec:Quantative-bounds-on}Quantitative upper bounds on
the error sequence $\epsilon_{N,\beta}$}

If $\phi$ satisfies A0, then the error sequence $\epsilon_{N_{k},\beta}[\phi,\phi_{F},dV]$
(defined by formula \ref{eq:def of epsion error in ineq}) tends to
zero, as $N_{k}\rightarrow\infty.$ Indeed, this follows from Theorem
\ref{thm:asymptot of log E on Riemann } with $(\phi+u,\phi)$ replaced
by $(\phi,\psi_{0}).$  If moreover A1 holds then a quantitative upper
bound on $\epsilon_{N_{k},\beta}[\phi,\phi_{F},dV]$ is provided by
the following
\begin{prop}
\label{prop:error estim on compact Riemann}Let $L$ be a positive
line bundle over a compact Riemann surface $X$ and $\phi$ a lsc
metric on $L$ such that $\mu_{\phi}$ has finite entropy. Then there
exists a constant $C$ only depending on an upper bound on the entropy
$D_{\mu_{0}}(\mu_{\phi})$ and $(\phi_{F},dV)$ and a constant $C_{(X,L)}$
only depending on $X$ and $L$ such that
\[
\frac{1}{\beta N_{k}k}\log Z_{N_{k},\beta}[dV,k\phi+\phi_{F}]\geq-NkE_{\phi}(\mu_{\phi})-C_{(X,L)}N\log N-N\beta^{-1}C
\]
$(C_{(X,L)}=1/2$ when $(X,L)=(\P^{1},\mathcal{O}(1)).$ As a consequence,
\[
\epsilon_{N_{k},\beta}[\phi,\phi_{F},dV]\leq C_{(X,L)}N\log N+N\beta^{-1}C
\]
The assumptions on $\phi$ are, in particular, satisfied if A1 holds. 
\end{prop}

\begin{proof}
This is a direct consequence of Gibbs variational principle (Lemma
\ref{lem:Gibbs}) combined with Lemma \ref{lem:asy of mean energy}. 

In the case when $\beta=1,$ $\phi$ is smooth on $X$ and $dd^{c}\phi>0$
it follows from Bergman kernel asymptotics that $\epsilon_{N_{k}}$
admits a complete asymptotics expansion
\[
\epsilon_{N_{k}}=C_{1}N_{k}^{-1}+C_{2}N_{k}^{-2}+...,
\]
(see \cite{k-m-m-w} and the companion paper \cite{berm17}). Hence,
there exists a constant $A_{\phi}$ such that $\epsilon_{N}\leq A_{\phi}N^{-1}.$
However, it appears to be non-trivial to estimate $A_{\phi}$ explicitly
in terms of $\phi$ (and $A_{\phi}$ appears to depend on more than
two derivatives of $\phi).$ 
\end{proof}

\section{\label{sec:The-Coulomb-gas}The Coulomb gas in the plane }

Consider $\R^{2}$ identified with $\C$ with complex coordinate $z.$
We will use the notation introduced in Section \ref{subsec:The-2D-Coulomb}.
Given a function $V$ on $\C,$ called the exterior potential, the
corresponding $N-$particle Hamiltonian is defined by 

\[
H_{V}^{(N)}(z_{1},...,z_{N}):=-\frac{1}{2}\sum_{i\neq j\leq n}\log|z_{i}-z_{j}|^{2}+\sum_{i\leq N}V(z_{i})
\]
We will denote by $\mu_{\beta,V}^{(N)}$ and $Z_{N,\beta}[V]$ the
corresponding Gibbs measure and partition function at inverse temperature
$\beta:$ 
\begin{equation}
d\mathbb{P}_{N,\beta}:=\frac{e^{-\beta H^{(N)}}d\lambda^{\otimes N}}{Z_{\beta,N}[V]},\,\,\,Z_{N,\beta}[V]:=\int_{\C^{N}}e^{-\beta H_{V}^{(N)}}d\lambda^{\otimes N}=\int_{\C^{N}}|D^{(N)}|^{2}(e^{-V}d\lambda)^{\otimes N},\label{eq:def of Gibbs measure Coul text}
\end{equation}
 where $D^{(N)}$ is the corresponding Vandermonde determinant appearing
in formula \ref{eq:Gibbs measure as det intro}. In order to ensure
that $Z_{N,\beta}[V]$ is finite and non-zero we introduce the following
\begin{defn}
\label{def:adm}A function $\phi(z)$ on $\C$ will be said to be
\emph{strongly admissible} if it is (a) is lower-semi continuous (b)
has the property that $\{\phi<\infty\}$ is non-polar and (c) has
\emph{strictly super logarithmic growth} in the sense that there exist
strictly positive constants $\epsilon$ and $C$ such that 
\[
\phi(z)\geq(1+\epsilon)\log((1+|z|^{2})-C
\]
More generally, $\phi$will be said to be \emph{admissible }if $\epsilon$
is allowed to vanish.

If $N=V\phi$ for a strongly admissible function or $V=(N+p)\phi$
for an admissible function $\phi$ and $p=2/\beta-1,$ then $Z_{N,\beta}[V]$
is non-zero and finite (see Lemma \ref{lem:growth gives lsc metric}). 
\end{defn}

\subsection{\label{subsec:Compactification-of-}Compactification of $\C$ by
the Riemann sphere}

In the following $X$ will denote the Riemann sphere and $L$ the
line bundle over $X$ of degree one. We identify $\C$ with a open
set $U_{0}$ of $X$ in the usual way, so that $X=\C\cup\{\infty\}$
and fix the standard trivialization $e_{U_{0}}$ of $L$ over $\C\subset X.$
For further background we refer to the survey \cite{berm 14 comma 5},
where the general case of $\C^{n}$ and its compactification by the
$n-$dimensional projective space $\P^{n}$ is recalled and where
the role of $L$ is played by the hyperplane line bundle $\mathcal{O}(1)$
over $\P^{n}.$ Concretely, in the present one-dimensional setting
$X$ is covered by two charts $U_{0}$ and $U_{\infty}$ biholomorphic
to $\C.$ Denoting by $z$ the coordinate on $U_{0}$ the coordinate
on $U_{\infty}$ is given by $z^{-1}$ on $U_{0}\cap U_{\infty}.$
The line bundle $L$ is defined by the transition function 
\begin{equation}
t_{\infty0}=z\label{eq:trans function for L over Riemann sp}
\end{equation}
 on $U_{0}\cap U_{\infty}.$ With these identifications the space
$H^{0}(X,kL)$ of all holomorphic sections $\Psi_{k}$ of $L$ gets
identified with the space of all polynomials $p_{k}(z)$ on $\C$
of degree at most $k:$ 
\[
\Psi_{k}=p_{k}e_{U_{0}}
\]
 Moreover, the space of all locally bounded metrics $\phi$ on $L$
gets identified with the space of functions $\phi(z)$ in $L_{loc}^{\infty}(\C)$
on $\C$ with logarithmic growth, i.e. such that 

\begin{equation}
\sup_{\C}\left|\phi(z)-\log(1+|z|^{2})\right|\leq C\label{eq:def of log growth}
\end{equation}
 for some constant $C$ (see \cite[Lemma 3.13]{berm 14 comma 5}).
The squared point-wise norm $\left\Vert \Psi_{k}\right\Vert ^{2}$
of an element $\Psi_{k}\in H^{0}(X,kL)$ may be expressed as follows
over $\C:$ 
\[
\left\Vert \Psi_{k}\right\Vert ^{2}(x)=|p_{k}(z)|^{2}e^{-k\phi(z)},
\]
 where $p_{k}$ is the polynomial $p_{k}(z)$ representing $\Psi_{k}$
and $\phi(z)$ is the function representing the metric $\left\Vert \cdot\right\Vert .$
Note that, since $p_{k}$ is a polynomial of degree $k$ the growth
property \ref{eq:def of log growth} ensures that the point-wise norm
$\left\Vert \Psi_{k}\right\Vert ^{2}$ is globally bounded on $\C,$
as it should (since it, a priori, defines a bounded function on $X).$ 

We will denote by $\psi_{0}$ the Fubini-Study metric, which defines
a canonical reference metric $\left\Vert \cdot\right\Vert _{FS}$
on $L\rightarrow X$ in the sense of Section \ref{subsec:Metrics-on}.
Its curvature form coincides with the usual invariant probability
measure on $X$ that we shall denote by $\mu_{0}$ (as in Section
\ref{subsec:Metrics-on}). In the standard trivialization over $\C$
we have the representation 
\[
\psi_{0}(z):=-\log\left\Vert e_{U_{0}}\right\Vert _{FS}^{2}(z)=\log(1+|z|^{2})
\]
The corresponding probability $\mu_{0}$ may be expressed as
\[
\mu_{0}=dd^{c}\psi_{0}=\frac{1}{\pi}e^{-2\psi_{0}}d\lambda
\]
\begin{lem}
\label{lem:growth gives lsc metric}Under the identifications above
the following holds:

\begin{itemize}
\item A lsc function $\phi(z)$ on $\C$ with the growth property 
\[
\inf_{\C}\left(\phi-\psi_{0}\right)>-\infty
\]
 corresponds to a lsc metric $\phi$ on $L\rightarrow X.$ Moreover,
if $u\in C_{b}(\C),$ then $\phi+u$ defines a lsc metric on $L\rightarrow X.$ 
\item Set $N_{k}:=k+1.$ Under the standard trivializations of $L$ and
$K_{X}$ over $\C\Subset X$ an element $\Psi_{k}$ in $H^{0}(X,(N_{k}+1)L+K_{X})$
corresponds to a polynomial $p_{k}$ of degree at most $k$ and the
squared adjoint $L^{2}-$norm of $\Psi_{k},$ for a metric $\phi$
on $L$ as in the previous point, is given by 
\begin{equation}
\int_{\C}|p_{k}(z)|^{2}e^{-(N_{k}+1)\phi(z)}d\lambda\label{eq:weighted norm with N+1 in lemma}
\end{equation}
In particular, the latter norm is always finite and non-zero if $\phi$
is admissible and $p_{k}$ does not vanish identically.
\item If $\phi(z)$ has strictly super logarithmic growth then $\int_{\C}|p_{k}(z)|^{2\beta}e^{-N\beta\phi(z)}d\lambda$
is finite for any polynomial $p_{k}$ of degree at most $k$ and $\beta>0.$ 
\end{itemize}
\end{lem}

\begin{proof}
The first point about $\phi$ is a special case of \cite[Lemma 3.13]{berm 14 comma 5}
and is shown by noting that $v:=\phi-\psi_{0}$ is a lsc function
on the dense subset $U_{0}$ of $X$ which is bounded from below.
Hence, identifying $v$ with its greatest lower semi-continuous extension
to $X$ concludes the proof. If $u\in C_{b}(\C)$ then $\phi+u$ is
lsc and has the same growth as $\phi$ and thus defines a lsc metric
on $L\rightarrow X,$ by the previous argument. The second point follows
from the fact that $K_{X}$ is isomorphic to $-2L$ and hence $H^{0}(X,kL)$
may be identified with the space $H^{0}(X,(k+2)L+K_{X}).$ For readers
lacking background in complex geometry we provide the following proof.
First to see that $K_{X}$ is isomorphic to $-2L$ note that the holomorphic
one-form $dz$ trivializes $K_{X}$ over $U_{0},$ since $z$ is a
holomorphic coordinate on $U_{0}.$ Similarly, since $z^{-1}$ is
a holomorphic coordinate on $U_{\infty}$ $d(z^{-1})$ trivializes
$K_{X}$ over $U_{1}.$ But 
\[
d(z^{-1})=-z^{-2}dz
\]
which proves the isomorphism in question (using that $L$ is defined
by the transition function in formula \ref{eq:trans function for L over Riemann sp}).
Since, as recalled above, $H^{0}(X,kL)$ identifies with the space
of all polynomials $p_{k}$ of degree at most $k$ this means that
$\Psi_{k}\in H^{0}(X,(N_{k}+1)L+K_{X}$ iff there exists $p_{k}$
such that
\[
\Psi_{k}=p_{k}e_{U_{0}}^{\otimes(N_{k}+1)}\otimes dz
\]
over $U_{0}.$ Hence, the squared adjoint $L^{2}-$norm of $\Psi_{k}$
is given by 
\[
\frac{i}{2}\int p_{k}dz\wedge\overline{p_{k}dz}e^{-(N_{k}+1)\phi}
\]
which coincides with the weighted norm in formula \ref{eq:weighted norm with N+1 in lemma}.
The finiteness stated in the lemma then follows directly from the
fact that the adjoint norm defined by a lsc metric is always bounded
(just using that the function $\psi_{0}-\phi$ is bounded from below).
Moreover, since the local density $|p_{k}(z)|^{2}e^{-(N+1)\phi_{N}}$
vanishes on a polar subset (which in particular has vanishing Lebesgue
measure) the norm in question is non-degenerate. To prove the last
point observe that 
\[
|p_{k}(z)|^{2}e^{-N\phi}=\left(|p_{k}(z)|^{2}e^{-k\phi}\right)e^{-\phi}\leq Ae^{-(1+\epsilon)\psi_{0}}
\]
 for a constant $A$ and $\epsilon>0.$ Since the right hand side
above is integrable this shows that $(|p_{k}(z)|^{2}e^{-N\phi})^{\beta}$
is integrable for $\beta\geq1.$ For $\beta<1$ we can instead apply
Hölder's inequality to reduce the integrability to the case $\beta=1.$ 
\end{proof}
\begin{example}
\label{exa:phi as sing metric on X}Consider a subharmonic function
$\phi(z)$ on $\C$ such that $\phi(z)-(1+\epsilon)\log(|z|^{2})$
is bounded for some $\epsilon>0$ when $|z|$ is sufficiently large.
Then the curvature current of the corresponding lsc metric on $L\rightarrow X$
is non-negative on $\C\subset X,$ but has a negative Dirac mass (of
mass $-\epsilon)$ at the point $x_{\infty}$ at infinity in $X.$
Indeed, in the standard trivialization of $L$ centered at $x_{\infty}$
the metric on $L$ is represented by $\phi+\log|w|^{2}$ which grows
like $(1+\epsilon)\log|w|^{-2}+\log|w|^{2}=-\epsilon\log|w|^{2}$
as $w=1/z\rightarrow0.$ 
\end{example}

\begin{lem}
\label{lem:slater as vandermonde and adjoint}Let $\Psi(x_{1},...,x_{N})$
be the Slater determinant induced by the standard multinomial basis
in $H^{0}(X,kL),$ where $N=k+1.$ Then the zero locus of $\Psi$
in $X^{N}$ is given by the set of all $(x_{1},...,x_{N})$ such that
there exists $i$ and $j$ such that $x_{i}=x_{j}.$ Moreover, in
the standard trivialization of $L$ over $\C$ the section $\Psi$
is represented by the Vandermonde determinant 
\[
D^{(N)}(z):=\prod_{i<j\leq N}(z_{i}-z_{j})=\det_{0\leq i,j\leq N-1}(z_{i}^{j})
\]
In particular, given a lsc metric $\phi$ on $L$ the squared point-wise
norm of $\Psi$ may be expressed as follows over $\C^{N}:$

\[
|\Psi(z_{1},...,z_{N})|^{2}e^{-\left(k\phi(z_{1})+...+k\phi(z_{N})\right)}=\prod_{i<j\leq N}|z_{i}-z_{j}|^{2}\prod_{i\leq N}e^{-k\phi(z_{i})}
\]
on $\C^{N}.$ As a consequence, the \emph{$N-$particle Coulomb gas
on $\C$ with exterior potential $(N+1)\phi$ at inverse temperature
$\beta=1$ coincides with the adjoint determinantal point process
on $X$ associated to the lsc metric $\phi$ on $L\rightarrow X.$ }
\end{lem}

\begin{proof}
In the standard trivialization of $L$ over $\C\subset\P^{1}$ the
multinomial base $\Psi_{i}$ is represented by the monomials $z^{i}.$
Hence, the corresponding Slater determinant on $\C^{N}$ may be identified
with the Vandermonde determinant $D^{(N)}(z_{1},...,z_{N}),$ which
satisfies the following classical identity: 
\begin{equation}
D^{(N)}(z_{1},...,z_{N}):=\det_{1\leq i,j\leq N-1}(z_{i}^{N-1})=\prod_{1\leq i<j\leq N}(z_{i}-z_{j})\label{eq:vanderm id}
\end{equation}
The right hand side vanishes precisely when $z_{i}=z_{j}$ for some
pair $(i,j),$ which proves the statement about the zero locus of
$\Psi$ in the subset $\C^{N}\Subset X^{N}.$ Repeating the same argument
in the standard chart $U_{\infty}$ around $x_{\infty}$ identifies
$\Psi$ with $D^{(N)}(w_{1},...,w_{N})$ over $U_{\infty},$ which
shows that vanishing locus of $\Psi$ in $U_{\infty}^{N}$ also has
the desired form. 

It follows from Lemma \ref{lem:growth gives lsc metric}, applied
to each factor of $\C^{N},$ that $Z_{N,1}[(N+1)\phi]$ is finite
if $\phi$ has logarithmic growth. In fact, a similar argument shows
that $Z_{N,1}[(N+\epsilon)\phi]$ is finite for any $\epsilon>0.$ 
\end{proof}
\begin{rem}
\label{rem:diverges}Using the identifications above it is not hard
to see that $Z_{N,1}[N\phi]=\infty$ as soon as $\phi$ has logarithmic
growth, i.e. 1\ref{eq:def of log growth} holds. Indeed, we can rewrite
\begin{equation}
d\P_{N}:=|\Psi(z_{1},...,z_{N})|^{2}e^{-\left((N\phi(z_{1})+...+N\phi(z_{N})\right)}d\lambda^{\otimes N}=|\Psi|_{k\phi}^{2}\nu^{\otimes N}\label{eq:rewritten density in remark diverges}
\end{equation}
 where $\nu=1_{\C}e^{-\phi}\frac{i}{2}dz\wedge d\bar{z}$ defines
a measure on $X$ and 
\[
|\Psi|_{k\phi}^{2}(z_{1},...,z_{N}):=|\Psi(z_{1},...,z_{N})|^{2}e^{-\left(k\phi(z_{1})+...+k\phi(z_{N})\right)}
\]
is the squared point-wise norm of the section $\Psi$ wrt the metric
induced by $\phi.$ In particular, $|\Psi|_{k\phi}^{2}$ is globally
bounded. But, $\nu$ gives infinity volume to any neighborhood of
the point $x_{\infty}$ at infinity in $X$ (indeed, in local coordinates
$w(=1/z)$ centered at $x_{\infty}$ the density of $\nu$ is comparable
to $1/|w|^{2})$. The integral of $d\P_{N}$ will thus be equal to
infinity on the neighborhood of any configuration $(x_{1},x_{2},...,x_{N})\in X^{N}$
such that $x_{1}$ is the point at infinity and $x_{i}\neq x_{j}$
for any $j$ (using that $|\Psi|_{k\phi}^{2}\neq0$ there, by Lemma
\ref{lem:Slater wrt FS vs Vandermonde}). Concretely, making the change
of variables $w=1/z$ on $\C-\{0\}$ the measure $\mu^{(N)}$ in formula
\ref{eq:rewritten density in remark diverges} becomes, for example
when $N=2,$ comparable to 
\[
|w_{1}-w_{2}|^{2}\frac{1}{|w_{1}|^{2}|w_{2}|^{2}}idw_{1}\wedge d\bar{w}_{1}\wedge idw_{2}\wedge d\bar{w}_{2}
\]
on the region of all $(w_{1},w_{2})$ such that $|w_{1}|$ and $|w_{2}-1|$
are sufficiently small (say bounded by $1/4).$ But the integral over
the $w_{1}-$ variable, for any fixed $w_{2}\neq0,$ clearly diverges.
\end{rem}

\subsubsection{\label{subsec:Partition-functions-in}Partition functions in $\C$
vs partition functions on the Riemann sphere $X$}

By Lemma \ref{lem:slater as vandermonde and adjoint} we can write
\[
Z_{N,1}[(N+1)\phi]=(\frac{i}{2})^{N}\int_{X^{N}}\Psi(x_{1},...,x_{N})\wedge\overline{\Psi(x_{1},...,x_{N})}e^{-(N+1)\left(\phi(x_{1})+...+\phi(x_{N})\right)}
\]
where $\Psi$ is the Slater determinant defined wrt the monomial base
in $H^{0}(X,(N+1)L+K_{X}),$ identified with the space of polynomials
on $\C\subset X$ of degree at most $N-1.$ This bases is different
then the fixed reference bases used in Section \ref{subsec:The-asymptotic-setting and the fixed references}.
We will denote by $\Psi_{FS}$ the latter Slater determinant, since
it is defined wrt the Fubini-Study metric. However, by the scaling
relation \ref{eq:scaling relation for slater}, variations wrt $\phi$
of the correponding partition functions are independent of the choice
of bases. Accordingly, differences of the form 
\[
\log Z_{N}[(\phi+u)]-\log Z_{N}[\phi]
\]
 may as well be computed with respect to the reference basis used
in Section \ref{subsec:The-asymptotic-setting and the fixed references}.
Accordingly, abusing notation slightly, we shall occasionally use
the same symbol $Z_{N}$ for the partition function in $\C$ (formula
\ref{eq:def of Gibbs measure Coul text}) and the partition function
on the Riemann sphere $X$ (formula \ref{eq:Gibbs measure for beta ens text}).
The precise constant linking the different Slater determinants is
provided by the following 
\begin{lem}
\label{lem:Slater wrt FS vs Vandermonde}Let $\Psi_{FS}^{(N)}$ be
the Slater determinant for the $N-$dimensional space $H^{0}(X,k\mathcal{O}(1)$
(where $N=k+1)$ induced by the basis $\Psi_{i}^{(N)}$ which is orthonormal
wrt the $L^{2}-$norm induced by the adjoint $L^{2}-$norm corresponding
to the Fubini-Study metric $\psi_{0}.$ Then 
\[
-\log|\Psi_{FS}^{(N)}|^{2}=-\log|D^{(N)}|^{2}+F(N+1),
\]
 where $F(k)$ is defined in Lemma \ref{lem:def and exp of F} below.
Moreover, given any admissible function $\phi$ and measure $dV$
we have, for any $\beta>0,$
\[
\frac{1}{N!}\left(\int\left(|\Psi_{FS}^{(N)}|^{2}e^{-k\phi}\right)^{\beta}dV^{\otimes N}\right)^{1/\beta}=\frac{\left(\int\left(|D^{(N)}|^{2}e^{-k\phi}\right)^{\beta}dV\right)^{1/\beta}}{\int|D^{(N)}|^{2}e^{-(N+1)\psi_{0}}d\lambda}
\]
 
\end{lem}

\begin{proof}
Making a change of bases gives that $|\Psi_{FS}^{(N)}|^{2}\det\left\langle \int z^{i}\bar{z}^{j}e^{-(N+1)\psi_{0}}d\lambda\right\rangle =|D^{(N)}|^{2}.$
Since the corresponding matrix is diagonal we have
\[
\det\left\langle \int z^{i}\bar{z}^{j}e^{-(N+1)\psi_{0}}d\lambda\right\rangle =\prod_{j=0}^{N-1}c_{j},\,\,c_{j}:=\int_{\C}|z|^{2j}e^{-(N+1)\psi_{0}}d\lambda,
\]
 which proves the first formula. To prove the second formula it will,
by the scaling relation \ref{eq:scaling relation for slater}, be
enough to show that
\[
N!=\int_{X^{N}}|\Psi_{FS}^{(N)}|^{2}e^{-(N+1)\psi_{0}}d\lambda^{\otimes N}
\]
 But this follows from formula \ref{eq:integral of det is N factorial}.
\end{proof}
\begin{lem}
\label{lem:def and exp of F}Set $c_{j}:=\int_{\C}|z|^{2j}e^{-(k+2)\psi_{0}}d\lambda$
for $j=0,1,...k.$ Then
\[
F(k):=\sum_{j=0}^{k}\log c_{j}=-\frac{1}{2}k^{2}-\frac{k}{2}\log k+O(k)
\]
and hence setting $N:=k+1$
\[
\log\int_{\C^{N}}|D^{(N)}|^{2}e^{-(N+1)\psi_{0}}d\lambda^{\otimes N}=\log N!+F(N+1)=-\frac{1}{2}N^{2}+\frac{N}{2}\log N+O(N)
\]
 
\end{lem}

\begin{proof}
This follows from the results in \cite[Section 3.2]{d-g-i-l}. In
fact, an explicit estimate for the error term $O(N)$ can be given
\cite{berm16}. 
\end{proof}

\subsubsection{\label{subsec:Energies-in-}Energies in $\C$ vs energies on the
Riemann sphere $X$}

Let $\mu$ be a probability measure on $\C$ and assume first that
$\mu$ has compact support in $\C.$ We then denote by $E_{0}(\mu)$
the usual logarithmic energy in $\C$
\[
E_{0}(\mu):=-\frac{1}{2}\int\log|z-w|^{2}\mu\otimes\mu\in]-\infty,\infty]
\]
More generally, if $\phi$ is an admissible function on $\C$ we define
the weighted logarithmic energy
\[
E_{0,\phi}(\mu):=E_{0}(\mu)+\int\phi\mu\in]-\infty,\infty]
\]
For a general probability measure we define 
\[
E_{0,\phi}(\mu):=-\frac{1}{2}\int\left(\log|z-w|^{2}-\psi_{0}(z)-\psi_{0}(w)\right)\mu\otimes\mu+\int(\phi-\psi_{0})\mu\in]-\infty,\infty]
\]
which is well-defined since both terms in the right hand side above
take values in $]-\infty,\infty].$ 

Next, identifying $\mu$ with a probability measure on the Riemann
sphere $X$ and $\phi$ with a metric on $\mathcal{O}(1)\rightarrow X$
we also get the energy $E_{\phi}(\mu),$ defined in Section \ref{subsec:The-energy-functional E phi as Leg tr}
wrt the reference metric $\psi_{0}$ defined by the Fubini-Study metric
(which has the property that $E_{\psi_{0}}(\mu_{0})=0).$
\begin{lem}
\label{lem:difference energies plane sphere}There exists a constant
$C$ such that, for any admissible function $\phi$ on $\C,$ 
\begin{equation}
E_{\phi}=E_{\phi,0}+C\label{eq:E phi is E not phi plus C}
\end{equation}
 In particular, if $u$ is a bounded continuous function on $\C,$
then 
\begin{equation}
\inf_{\mu\in\mathcal{P}(\C)}E_{0,\phi+u}(\mu)-\inf_{\mu\in\mathcal{P}(\C)}E_{0,\phi}(\mu)=\inf_{\mu\in\mathcal{P}(X)}E_{\phi+u}(\mu)-\inf_{\mu\in\mathcal{P}(X)}E_{\phi}(\mu)\label{eq:diff of energies in the plane}
\end{equation}
where $u$ has been identified with the lsc function on the Riemann
sphere $X$ obtained as the greatest lsc extension of $u$ from $\C.$ 
\end{lem}

\begin{proof}
First observe that the Green function $G$ on $X,$ defined by formula
\ref{eq:def of Green f}, may be expressed as 
\[
G(z,w)=G_{0}(z,w)+2C,\,\,\,-G_{0}(z,w)=\log|z-w|^{2}-\psi_{0}(z)-\psi_{0}(w)
\]
Indeed, function $G_{0}(z,w)$ satisfies $(i)$ in formula \ref{eq:def of Green f}
and hence $(ii)$ also holds if
\begin{equation}
C:=-\frac{1}{2}\int_{X}G_{0}(x,\cdot)\omega_{0}=-\frac{1}{2}\int_{X\times X}G_{0}\omega_{0}\otimes\omega_{0}\label{eq:def of constant C}
\end{equation}
Thus 
\[
E_{\psi_{0}}(\mu)=-\frac{1}{2}\int\left(\log|z-w|^{2}-\psi_{0}(z)-\psi_{0}(w)\right)\mu\otimes\mu+C=E_{0,\psi_{0}}(\mu)+C
\]
and hence formula \ref{eq:E phi is E not phi plus C} follows. Next,
take $u\in C_{b}(\C).$ Then the functions $\Phi:=\phi+u$ and $\phi$
are admissible on $\C$ and may hence, by Lemma \ref{lem:growth gives lsc metric},
be identified with admissible metrics on $L.$ Now take $\mu\in\mathcal{P}(X).$
If $\mu$ charges the point at infinity in $X$ then it has infinite
energy (since any point defines a polar subset of $X).$ Hence, when
considering the infima in the right hand side of formula \ref{eq:diff of energies in the plane}
we may as well assume that $\mu\in\mathcal{P}(\C).$ But then \ref{eq:diff of energies in the plane}
follows from formula \ref{eq:E phi is E not phi plus C}. 

Note that in Section \ref{subsec:Statement-of-the} the notation $E_{\phi}$
was used for the functional $E_{0,\phi},$ but this abuse of notation
should be harmless as we will mainly consider differences as in the
previous lemma. 
\end{proof}
\begin{rem}
The constant $C$ defined by formula \ref{eq:def of constant C} is,
in fact, given by $-1/2,$ as revealed by an explicit calculation.
Alternatively, by Lemma \ref{lem:def and exp of F}
\[
\frac{1}{2}=\lim_{N\rightarrow\infty}N^{-2}\log\int_{\C^{N}}|D^{(N)}|^{2}e^{-(N+1)\psi_{0}}d\lambda^{\otimes N}=\inf_{\mathcal{P}(\C)}E_{0,\psi_{0}},
\]
 using the large deviation results in \cite{berm1} in the last equality.
But the right hand side above coincides with $-C.$
\end{rem}

Denote by $H^{1}(\C)$ the space of all $u\in L_{loc}^{2}(\C)$ such
that $\nabla u\in L^{2}(\C)$ in the sense of distributions and set
\[
\left\Vert u\right\Vert _{H^{1}(\C)}^{2}:=\int_{\C}du\wedge d^{c}u=\frac{1}{4\pi}\int_{\C}|\nabla u|^{2}d\lambda
\]
 We will say that $u_{j}\rightarrow u$ in the $H^{1}-$topology if
$u_{j}\rightarrow u$ in $L_{loc}^{2}$ and $\left\Vert u_{j}-u\right\Vert _{H^{1}(\C)}\rightarrow0.$
The space $H^{1}(\C)$ may be identified with the corresponding space
$H^{1}(X)$ on the Riemann sphere according to the following basic
\begin{lem}
\label{lem:restr is a homeo}The restriction map $H^{1}(X)\rightarrow H^{1}(\C)$
is a homeomorphism. 
\end{lem}

\begin{proof}
Injectivity is immediate. To see that surjectivity holds take $u\in H^{1}(\C),$
i.e. $u\in L_{loc}^{2}$ and $\int_{\C}du\wedge d^{c}u<\infty.$ Consider
the linear operator on $H^{1}(X)/\R$ defined by $\Lambda[v]:=\int_{\C}du\wedge d^{c}v.$
By the Cauchy-Schwartz inequality $\Lambda$ is continuous on $H^{1}(X)/\R$
and hence there exists $\tilde{u}\in H^{1}(X)$ such that $\Lambda[v]=\int_{X}d\tilde{u}\wedge d^{c}v.$
It follows that $d(\tilde{u}-u)=0$ in the sense of distributions
on $\C$ and hence $\tilde{u}-u$ is constant, showing surjectivity
of the restriction map. All that remains is to show that it is continuous.
To this end assume that $u_{j}\rightarrow u$ in $H^{1}(\C).$ By
the previous steps we may identify $u_{j}$ and $u$ with elements
in $H^{1}(X).$ Hence, $u_{j}$ and $u\in L^{2}(\mu_{0}),$ where
$\mu_{0}=dV_{g}$ for the standard metric $g$ on $X.$ Set $v_{j}:=u_{j}-\int u_{j}\mu_{0}$
and $v:=v-\int u\mu_{0}.$ By assumption $\int_{X}dv_{j}\wedge d^{c}v\rightarrow0$
and hence it it follows from the Poincaré inequality for $\Delta_{g}$
on $X$ that $v_{j}\rightarrow v$ in $L^{2}(X,g).$ Since both $u_{j}\rightarrow u$
and $v_{j}\rightarrow v$ as distributions on $\C$ it follows that
$\int u_{j}\mu_{0}\rightarrow\int u\mu_{0}$ and hence $u_{j}\rightarrow u$
in $L^{2}(X,dV_{g}),$ as desired. 
\end{proof}
\begin{rem}
It follows from the previous lemma that if $u\in H^{1}(\C),$ then
$u/(1+|z|^{2})\in L^{2}(\C,d\lambda),$ but, in general, $u$ is not
in $L^{2}(\C,d\lambda).$ It also follows from the proof that it is
enough to assume that $u$ is a distribution such that $\nabla u\in L^{2}(\C)$
and that $u_{j}\rightarrow u$ in the sense of distributions on $\C$
in the definition of the $H^{1}(\C)-$topology.
\end{rem}

If $\phi$ is an admissible function on $\C$ and $u\in H^{1}(\C)$
we define the corresponding free energy 
\[
\mathcal{F}(\phi+u):=\inf_{\mu}\left(E_{0,\phi}(\mu)+\int_{\C}u\mu\right)
\]
where the infimum runs over all $\mu\in\mathcal{P}(\C)$ such that
$E_{0,\phi}(\mu)<\infty.$ Combining Lemma \ref{lem:restr is a homeo}
and Theorem \ref{thm:conv of E wrt phi and u} then gives the following
\begin{thm}
\label{thm:conv of free energy in C}Let $\phi$ be an admissible
function in $\C$ and $u\in H^{1}(\C).$ Then
\begin{itemize}
\item $\mathcal{F}(\phi+u)$ is finite 
\item Let $\phi_{N}$ be a sequence of admissible functions in $\C$ increasing
to an admissible function $\phi$ and $u_{N}$ a sequence in $H^{1}(\C)$
converging to $u\in H^{1}(\C)$ in the $H^{1}-$topology. Then $\mathcal{F}(\phi_{N}+u_{N})\rightarrow\mathcal{F}(\phi+u).$ 
\end{itemize}
\end{thm}

\subsection{The case when $V=(N+p)\phi$}

Given $\beta>0$ we set 
\[
p:=\frac{2}{\beta}-1
\]
In this section we will assume that $\beta\leq1,$ i.e. $p\geq1.$
Given a function $\phi$ in $\C$ with super logarithmic growth we
define the an\emph{ error sequence} \emph{$\epsilon_{N,\beta}[\phi]:$}
\[
\epsilon_{N,\beta}[\phi]=\epsilon_{N,\beta}^{(s)}[\phi]-\epsilon_{N,\beta}^{(u)},
\]
 decomposed in the \emph{specific error $\epsilon_{N}^{(s)}[\phi]$
}and the \emph{universal error $\epsilon_{N}^{(u)}$ }(which independent
of $\phi),$ where 

\begin{equation}
\epsilon_{N,\beta}^{(s)}[\phi]:=-\left(\frac{1}{\beta N(N+p)}\log Z_{N,\beta}[(N+p)\phi]+\mathcal{F}(\phi)\right)\label{eq:def of error sequence for V N plus p text}
\end{equation}
\[
\epsilon_{N,\beta}^{(u)}:=-\frac{N+1}{N+p}\left(\frac{1}{N(N+1)}\log Z_{N,1}[(N+1)\psi_{0}]+\mathcal{F}(\psi_{0})\right)-\frac{\frac{1}{\beta}-1}{N+p}\left(\log\int e^{-2\psi_{0}}d\lambda+2\mathcal{F}(\psi_{0})\right)
\]
Rearranging the terms gives the following alternative expression:
\[
\epsilon_{N,\beta}[\phi]:=-\left(\frac{1}{N(N+p)}\log\frac{\left(Z_{N,\beta}[(N+p)\phi]\right)^{1/\beta}}{Z_{N,1}[(N+1)\psi_{0}]}+\mathcal{F}(\phi)-\mathcal{F}(\psi_{0})\right)+\frac{\frac{1}{\beta}-1}{N+p}\log\pi
\]
\begin{rem}
\label{rem:error is minimal}When $\beta=1$ we have $\epsilon_{N,\beta}[\phi]=\epsilon_{N,\beta}^{(s)}[\phi]-\epsilon_{N,\beta}^{(s)}[\psi_{0}],$
which is minimal and vanishes for $\phi=\psi_{0}$ (by the proof of
Theorem \ref{thm:upper bd on part funct riem surf adj}). More generally,
$\epsilon_{N,\beta}[\phi]=0$ if $dd^{c}\phi=F(dd^{c}\psi_{0})$ for
some Möbius transformation $F$ (i.e. a biholomorphism of the Riemann
sphere). Indeed, Theorem \ref{thm:upper bd on part funct riem surf adj}
only requires that $dd^{c}\phi$ has constant scalar curvature, viewed
as a Kähler metric on the Riemann sphere. 
\end{rem}

\subsubsection{Proof of Theorem \ref{thm:sharp non Gaussian bd intro}}

\emph{Step 1: the case $\beta=1.$ }

Set $N:=k+1$ and identify $H^{0}(X,kL)$ with $H^{0}(X,(N+1)L+K_{X})$
as in Lemma \ref{lem:growth gives lsc metric}. Thanks to the identifications
used in Lemma \ref{lem:growth gives lsc metric} and the discussions
in Sections \ref{subsec:Energies-in-} and \ref{subsec:Partition-functions-in}
we can apply Theorem \ref{thm:upper bd on part funct riem surf adj}
and deduce that 
\[
\frac{1}{N(N+1)}\log\frac{Z_{N}[(N+1)(\phi+u)]}{Z_{N}[(N+1)\phi]}\leq-\mathcal{F}(\phi+u))+\mathcal{F}(\phi)+\epsilon[(N+1)\phi],
\]

\emph{Step 2: the case $\beta\leq1.$ }

Set $\Phi:=\phi+u.$ First observe that
\begin{equation}
\frac{1}{\beta}\log Z_{N,\beta}[(N+p)\Phi]\leq\log Z_{N}[(N+1)\Phi]+(1/\beta-1)N\log\int e^{-2\Phi}\label{eq:ineq in proof Thm V adjoint}
\end{equation}

To this end rewrite the adjoint norm in formula \ref{eq:weighted norm with N+1 in lemma}
as 
\[
\int|p_{k}|^{2}e^{-(N+1)\phi}d\lambda=\int|p_{k}|^{2}e^{-k\phi}dV,\,\,\,\,dV=e^{-2\phi}d\lambda
\]
Since $\phi$ has super logarithmic growth the measure $dV$ has finite
mass. Hence, by Hölder's inequality 
\[
\left(\int\left(|p_{k}|^{2}e^{-(k+2\beta^{-1})\phi}\right)^{\beta}d\lambda\right)^{1/\beta}=\left(\int(|p_{k}|^{2}e^{-k\phi})^{\beta}dV\right)^{1/\beta}\leq\left(\int|p_{k}|^{2}e^{-k\phi}dV\right)(\int dV)^{(1/\beta-1)}
\]
Hence, replacing $p_{k}$ with $D^{(N_{k})}$ and $\phi$ with $\Phi$
and using the previous inequality on each factor of $\C^{N}$ proves
the inequality \ref{eq:ineq in proof Thm V adjoint}. As a consequence,

\[
\frac{1}{\beta N(N+p)}\log Z_{N,\beta}[(N+p)\Phi]\leq\frac{N+1}{N+p}\frac{1}{N(N+1)}\log Z_{N}[(N+1)\Phi]+\frac{1}{N+p}(1/\beta-1)\log\int e^{-2\Phi}
\]
Hence, applying Step one, i.e. the case when $\beta=1,$ for $N$
and $N=1$ (see Remark \ref{rem:upper bound on Z in C}) gives 
\[
\frac{1}{\beta N(N+p)}\log Z_{N,\beta}[(N+p)\Phi]\leq
\]

\[
\leq\frac{N+1}{N+p}(-\mathcal{F}(\Phi))+\frac{N+1}{N+p}\left(\log Z_{N,1}[(N+1)\psi_{0}]+\mathcal{F}(\psi_{0})\right)+\frac{1}{N+p}(1/\beta-1)\left(-2\mathcal{F}(\Phi)+\log Z_{1,1}[2\psi_{0}]+2\mathcal{F}(\psi_{0})\right)
\]
But $\frac{N+1}{N+p}+\frac{2}{N+p}(1/\beta-1)=1$ and hence 
\[
\frac{1}{\beta N(N+p)}\log Z_{N,\beta}[(N+p)\Phi]\leq-\mathcal{F}(\Phi)+\frac{N+1}{N+p}\left(\log Z_{N,1}[(N+1)\psi_{0}]+\mathcal{F}(\psi_{0})\right)+\frac{1}{N+p}(1/\beta-1)\left(\log Z_{1,1}[2\psi_{0}]+2\mathcal{F}(\psi_{0})\right)
\]
By the very definition of the error sequence \ref{eq:def of error sequence for V N plus p text}
this concludes the proof. 
\begin{rem}
\label{rem:upper bound on Z in C}Note that Theorem \ref{thm:sharp non Gaussian bd intro}
is (for $\beta=1)$ equivalent to the upper bound
\[
\frac{1}{N(N+1)}\log\frac{Z_{N}[(N+1)\Phi]}{Z_{N}[(N+1)\psi_{0}]}\leq-\mathcal{F}(\Phi)+\mathcal{F}(\psi_{0})
\]
for any admissible $\Phi.$ Indeed, the error term $\epsilon[(N+1)\phi]$
is defined so that the previous inequality translates into an upper
bound on $\frac{1}{N(N+1)}\log\E(^{-(N+1)\Phi})/\E(^{-(N+1)\phi})$
for $\Phi=\phi+u.$ 
\end{rem}

\subsubsection{Proof of Theorem \ref{thm:sharp Gauss bound on E intro}}

Theorem \ref{thm:sharp Gauss bound on E intro} follows directly from
applying Theorem \ref{thm:sharp non Gaussian bd intro} to $u-\bar{u}$
and using Lemma \ref{lem:smaller than dirichlet}. 

\subsubsection{Proof of Corollary \ref{cor:Gaussian dev ineq}}

The corollary follows from Theorem \ref{thm:sharp non Gaussian bd intro}
using the following standard observation: let $Y$ be a, say continuous,
function on a topological space $\mathcal{X}$ endowed with a (Borel)
measure $\nu$ such that $\int e^{-\lambda tY}\nu\leq Ce^{\lambda Bt^{2}/2}.$
Then $\int_{\{Y\geq s\}}\nu\leq Ce^{-\lambda B^{-1}s^{2}/2}.$ Indeed,
for any fixed $t$ we have $\int_{\{Y\geq s\}}\nu=\int_{\{Y\geq s\}}e^{tY}e^{-tY}\nu\leq e^{sY}\int e^{-tY}\nu\leq Ce^{Bt^{2}/2-st}.$
Hence, taking the inf over all $t$ and completing the square proves
the observation when $\lambda=1$ and the general case when follows
from rescaling $t.$

\subsection{The setting where $V=N\phi$ for $\phi$ with strictly logarithmic
growth}

Given $\beta\in]0,1]$ and $\phi$ with strictly super logarithmic
growth we set 
\[
\phi_{N}=\frac{N}{N+p},\,\,\,p:=2/\beta-1.
\]
 We can then rewrite

\begin{equation}
Z_{N,\beta}[N\phi]=Z_{N,\beta}[N(\phi_{N}+2/\beta-1)]\label{eq:ZN in terms of rescaled phi}
\end{equation}
\begin{rem}
In order to ensure that $\phi_{N}$ has logarithmic growth $N$ has
to be taken sufficiently large. But, in fact, it is not hard to see
that all the arguments below are valid for any $N.$ The point is
that, in the adjoint setting (and hence also the setting where $V=(N+p)\phi)$
one could more generally work with metrics which are locally of the
form $\phi+c_{k}\log|z|^{2}$ for $\phi$ lsc as long as $c_{k}k<1$
(which ensures that $e^{-k\phi}$ is locally integrable). 
\end{rem}

\begin{lem}
\label{lem:estimates on energy for phi stongly ad}Assume that $\phi$
is strongly admissible. Then

\[
-N(N+p)E_{\phi_{N}}(\mu_{\phi_{N}})\leq-N^{2}E_{\phi}(\mu_{\phi})-NpE_{0}(\mu_{\phi_{N}})
\]
and 
\[
N(N+p)E_{\phi_{N}}(\mu_{\phi_{N}})\leq N^{2}E_{\phi}(\mu_{\phi})+NpE_{0}(\mu_{\phi})
\]
As a consequence, $E_{0}(\mu_{\phi_{N}})\leq E_{0}(\mu_{\phi}).$
\end{lem}

\begin{proof}
\emph{Step 1: The first inequality}

Since $\mu_{\phi}$ minimizes $E_{\phi}$ we have $E_{\phi}(\mu_{\phi})\leq E_{\phi}(\mu_{\phi_{N}})$
and hence
\[
-E_{\phi_{N}}(\mu_{\phi_{N}})\leq-E_{\phi}(\mu_{\phi})+\int(\phi-\phi_{N})\mu_{\phi_{N}}
\]
Multiplying both sides with $N^{2}$ thus gives
\[
-N(N+p)E_{\phi_{N}}(\mu_{\phi_{N}})\leq-N^{2}E_{\phi}(\mu_{\phi})-pNE_{\phi_{N}}(\mu_{\phi_{N}})+N\int N(\phi-\phi_{N})\mu_{\phi_{N}}
\]
Since $N(\phi-\phi_{N})=p\phi_{N}$ the last integral above may be
expressed as
\[
\int N(\phi-\phi_{N})\mu_{\phi_{N}}=p\int\phi_{N}\mu_{\phi_{N}}=p(E_{\phi_{N}}(\mu_{\phi_{N}})-E(\mu_{\phi_{N}})),
\]
 which proves the first inequality. 

\emph{Step 2: The second inequality}

By Prop \ref{prop:beatif F is diff} the function $f(t):=\mathcal{F}(t\phi):=E_{t\phi}(\mu_{t\phi})$
is differentiable. Moreover, since it can be realized as the infimum
of affine functionals it is also concave. Hence, $f(t)\leq f(1)+f'(1)(t-1)$
which, applied to $t=N/(N+p)$ translates into 
\[
E_{\phi_{N}}(\mu_{\phi_{N}})\leq E_{\phi}(\mu_{\phi})+(\frac{N}{N+p}-1)\int\mu_{\phi}\phi=E(\mu_{\phi})+\frac{N}{N+p}\int\mu_{\phi}\phi=\frac{N}{N+p}E_{\phi}(\mu_{\phi})+\frac{p}{N+p}E(\mu_{\phi})
\]
\end{proof}

\subsubsection{Proof of Theorem \ref{thm:sharp non-Gauss ine for strong}}

Applying Theorem \ref{thm:sharp non Gaussian bd intro} to $\phi_{N}:=\phi N/(N+p)$
and $u_{N}:=uN/(N+p)$ gives
\begin{equation}
\frac{1}{\beta}\log\frac{Z_{N}[N(\phi+u)]}{Z_{N}[N\phi]}\leq N(N+p)\left(-E_{\Phi_{N}}(\mu_{\Phi_{N}}))+E_{\phi_{N}}(\mu_{\phi_{N}})\right)+N(N+p)\epsilon_{N,\beta}[\phi_{N}]\label{eq:ineq in proof ine for E strong}
\end{equation}
Applying Lemma \ref{lem:estimates on energy for phi stongly ad} thus
gives
\[
\frac{1}{N^{2}}\frac{1}{\beta}\log\frac{Z_{N}[N(\phi+u)]}{Z_{N}[N\phi]}\leq\left(-E_{\Phi}(\mu_{\Phi})+E_{\phi_{N}}(\mu_{\phi_{N}})\right)+pN^{-1}\left(-E_{0}(\mu_{\Phi_{N}})+E_{0}(\mu_{\phi})\right)+(1+N^{-1}p)\epsilon_{N,\beta}[\phi_{N}]
\]

\subsubsection{Proof of Theorem \ref{thm:gaussian ineq for E for strong}}

By definition we can express $\E_{N\phi}(e^{N^{2}\beta tU})=\E_{(N+p)\phi_{N}}(e^{N(N+p)\beta tU_{N}},$
where $U_{N}$ is the linear statistic of $u_{N}:=u(N+p)/N.$ Hence,
applying the inequality in Theorem \ref{thm:sharp Gauss bound on E intro}
to $\phi_{N}$ and $u_{N}$ gives
\[
\beta^{-1}N^{-2}\log\E(e^{N^{2}t\beta\left(U-\bar{u}\right)})\leq\frac{1}{2}\frac{N}{N+p}t^{2}J(u_{N})+|t\int u(\mu_{\phi_{N}}-\mu_{\phi})|+(1+N^{-1}p)\epsilon_{N,\beta}[\phi_{N}]
\]

Now, $\frac{N+p}{N}J(u_{N})=\frac{N}{N+p}J(u)$ and the Cauchy-Schwartz
inequality gives, 
\[
|\int u(\mu_{\phi_{N}}-\mu_{\phi})|\leq\left\Vert u\right\Vert _{H^{1}}\left\Vert P(\phi_{N})-P(\phi)\right\Vert _{H^{1}}\leq\frac{p}{N+p}b_{\phi}\left\Vert u\right\Vert _{H^{1}}
\]
 using Prop \ref{prop:H one conv for H-S derivative} in the last
step. Finally, decomposing $\frac{N}{N+p}J(u_{N})=J-Jp/N+p$ and using
the inequality $ab\leq a^{2}/2+b/2$ with $a=\left\Vert u\right\Vert _{H^{1}}$
and $b=b_{\phi}$ concludes the proof (with $B_{\phi}:=b_{\phi}^{2}/2$).

\subsection{\label{subsec:Estimates-of-the}Estimates of the error terms}

In this section we show how to explicitly estimate the error terms,
in terms of $\phi,$ appearing in the main results. 

\subsubsection{The case when $V=(N+p)\phi$}
\begin{prop}
\label{prop:bound on error for V adj}Assume that $\phi$ has logarithmic
growth. Then 
\[
\frac{1}{N(N+p)}\log\frac{\left(Z_{N,\beta}[(N+p)\phi]\right)^{1/\beta}}{Z_{N,1}[(N+1)\psi_{0}]}+E_{\phi}(\mu_{\phi})\geq\frac{p+1}{N+p}\inf_{S}(\phi-\psi_{0})-N^{-1}\beta^{-1}\log\sup_{S}\frac{dd^{c}\phi}{e^{-2\phi}}
\]
As a consequence, 
\[
\epsilon_{N,\beta}[\phi]\leq\beta^{-1}N^{-1}C_{\phi}+\frac{1}{2}N^{-1}\log N+O(N)
\]
 where $C_{\phi}$ only depends on a lower bound on $\inf_{S}(\phi-\psi_{0})$
and an upper bound on $\sup_{S}\frac{dd^{c}\phi}{e^{-2\phi}}$ (which
in turn only depends on upper bounds on $dd^{c}\phi/\omega_{0}$ and
$\sup(\phi-\psi_{0})$)
\end{prop}

\begin{proof}
Applying Lemma \ref{lem:Gibbs} with $\mu_{0}=e^{-2\phi}d\lambda$
and $\nu=\mu_{\phi}$ together with Lemma \ref{lem:asy of mean energy}
with $k=N-1$ gives
\[
\frac{1}{N(N-1)}\log\frac{\left(Z_{N,\beta}[(N+p)\phi]\right)^{1/\beta}}{Z_{N,1}[(N+1)\psi_{0}]}+E_{\phi}(\mu_{\phi})\geq-N^{-1}\beta^{-1}D_{dV}(\mu_{\phi})
\]
 Hence, 
\[
\frac{1}{\beta N(N+p)}\log\frac{\left(Z_{N,\beta}[(N+p)\phi]\right)^{1/\beta}}{Z_{N,1}[(N+1)\psi_{0}]}+E_{\phi}(\mu_{\phi})\geq\frac{p+1}{N+p}E_{\phi}(\mu_{\phi})-N^{-1}\beta^{-1}D_{dV}(\mu_{\phi}))
\]
Note that $E_{\phi}(\mu_{\phi})\geq\inf_{S}(\phi-\psi_{0})$ and $p+1=2/\beta.$
Moreover, $D_{dV}(\mu_{\phi})\leq\log\sup_{S}\frac{dd^{c}\phi}{dV}.$
This proves proves the proposition (also using the expansion in Lemma
\ref{lem:def and exp of F}).
\end{proof}

\subsubsection{The case when $V=N\phi$}

We recall that the error term $\tilde{\epsilon}_{N}$ appearing in
Theorem \ref{thm:sharp non-Gauss ine for strong} is defined in terms
of $t\phi$ for $t=N/N+p.$ Since we are assuming that $\beta\leq1$
we have $t\leq1.$ 

The following lemma shows how to estimate the outer radius $R(t\phi)$
of the support $S_{t\phi}$ in terms of $1-t$ and parameters quantifying
the growth of $\phi:$ 
\begin{lem}
\label{lem:bound on outer radi}Suppose that $\phi$ satisfies $\phi\geq(1+\epsilon)\log^{+}|z|^{2}+C_{1}$
for $\epsilon>0$ and $\phi$ is locally bounded, where $\log^{+}t:=\log\max\{1,t\}.$
Then the outer radius $R$ of $D_{\phi}$ satisfies
\[
2\log R\leq\epsilon^{-1}(C_{2}-C_{1})\,\,\,C_{2}=\sup_{|z|=1}\phi
\]
More generally, if $t=1-\delta$ for $\delta>0$ then the outer radius
$R(\phi t)$ of $D_{t\phi}$ satisfies 
\[
2\log R(\phi t)\leq(1-\delta)\frac{C_{2}-C_{1}}{\epsilon-\delta(1+\epsilon)}\,\,\,C_{2}=\sup_{|z|=1}\phi
\]
\end{lem}

\begin{proof}
First observe that $P\phi\leq\log^{+}|z|^{2}+C_{2}.$ Indeed, let
$\psi$ be a candidate for the sup defining $P\phi.$ Then $\psi\leq\log^{+}|z|^{2}+C_{2}$
almost everywhere wrt $\Delta(\log^{+}|z|^{2}).$ Hence, by the domination
principle, $\psi\leq\log^{+}|z|^{2}+C_{2}$ on all of $\C,$ as desired.
This means that on $D_{\phi}$ we have 
\[
(1+\epsilon)\log^{+}|z|^{2}+C_{1}\leq\phi\leq P\phi\leq\log^{+}|z|^{2}+C_{2}
\]
 Hence, $\epsilon\log^{+}R^{2}\leq C_{2}-C_{1},$ which concludes
the proof of the first inequality in the lemma. The second one then
follows from the observation that $C_{i}(t)=tC_{i}$ and $\epsilon(t)=1-\delta(1+\epsilon).$

Next, we establish an explicit upper bound on the term $B_{\phi}$
appearing in Theorem \ref{thm:gaussian ineq for E for strong} (applied
to $t=N/(N+p)).$
\end{proof}
\begin{prop}
\label{prop:H one conv for H-S derivative}Suppose that $\phi$ has
strictly logarithmic growth and satisfies A0. Then
\[
J(v_{t})\leq(1-t)\left|\sup_{S_{t\phi}}\phi-\inf_{S_{t\phi}}\phi\right|\leq2(1-t)\left|\sup_{D_{R(\phi t)}}\phi-\inf_{D_{R(\phi t)}}\phi\right|\int_{D_{R(\phi t)}}dd^{c}\phi,
\]
 where $D_{R(t\phi)}$ is the disc centered at $0$ of radius $R(t\phi)$
appearing in Lemma \ref{lem:bound on outer radi}. 
\end{prop}

\begin{proof}
Set $\psi_{t}:=P(t\phi)$ for $t>0.$ Integration by parts gives 
\[
J(v_{t})=\int dd^{c}(\psi_{t}-\psi_{1})(\psi_{1}-\psi_{t})=\int(\mu_{t\phi}-\mu_{\phi})(\psi_{1}-\psi_{t}).
\]
 Next, we may, after perhaps replacing $\phi$ with $\phi-\inf\phi,$
assume that $\phi\geq0.$ Then $\psi_{t}$ is increasing in $t$ and,
in particular, $\psi_{1}-\psi_{t}\geq0.$ Moreover, rewriting $\phi=t\phi+(1-t)\phi$
and using the domination principle for the Laplacian gives
\[
P(\phi)\leq P(t\phi)+(1-t)\sup_{S_{t\phi}}\phi
\]
which concludes the proof.
\end{proof}
The following bound on the error term $a(\phi,u)$ appearing in Theorem
\ref{thm:sharp non-Gauss ine for strong} is far from optimal, but
will be adequate for the applications to moderate deviation principles
in \cite[Lemma 3.13]{berm15}: 
\begin{prop}
Assume that $u$ is bounded on $\C$ and $\beta\leq1.$ Then 
\[
-E_{0}(\mu_{\phi_{N}+u_{N}})+E_{0}(\mu_{\phi})\leq2\left\Vert u\right\Vert _{L^{\infty}}
\]
\end{prop}

\begin{proof}
First by the last claim in Lemma \ref{lem:estimates on energy for phi stongly ad}
we have $E_{0}(\mu_{\phi})\leq E_{0}(\mu_{\phi_{N}}).$ Since $-E$
is a concave functional it follows that 
\[
-E_{0}(\mu_{\phi_{N}+u_{N}})+E_{0}(\mu_{\phi})\leq\int P\phi_{N}(\mu_{\phi_{N}+u_{N}}-\mu_{\phi_{N}})=\int P\phi_{N}dd^{c}\left(P(\phi_{N}+u_{N})-P(\phi_{N})\right)
\]
Now, $\psi:=P\phi_{N}$ can be viewed as a lsc function on the Riemann
sphere $X$ which is equal to $+\infty$ at the point at the point
$x_{\infty}$ at infinity. Since $\psi-\log|z|^{2}$ is bounded close
to $x_{\infty}$ we get $dd^{c}\psi=1_{\C}\mu_{\phi_{N}}-\delta_{x_{\infty}}.$
Hence, integrating by parts, 
\[
\int P\phi_{N}dd^{c}\left(P(\phi_{N}+u_{N})-P(\phi_{N})\right)=\int\left(P(\phi_{N}+u_{N})-P(\phi_{N})\right)\mu_{\phi_{N}}-\left(P(\phi_{N}+u_{N})-P(\phi_{N})\right)(x_{\infty}).
\]
(this is a special case of f \cite[formula 1.2]{b-t}{]}. All that
remains is to verify that $\left|P(\phi_{N}+u_{N})-P(\phi_{N})\right|\leq\left\Vert u_{N}\right\Vert _{L^{\infty}}.$
But this follows readily from the fact that $P$ is increasing and
$P(\phi+C)=P(\phi)+C$ for any $C\in\R.$
\end{proof}
\begin{prop}
Assume that $\phi$ has strictly logarithmic growth, Then 

\[
\tilde{\epsilon}_{N}\leq\beta^{-1}N^{-1}C_{\phi}+\frac{1}{2\beta}N^{-1}\log N+O(N).
\]
where $C_{\phi}$ only depends on upper bounds on $\phi$ and $\Delta\phi$
and a lower bound on $\phi$ on the disc $D_{R(t\phi)}$ centered
at $0$ of radius $R(t\phi)$ as in Lemma \ref{lem:bound on outer radi}
with $t=N/N+p.$
\end{prop}

\begin{proof}
This follows from combining Proposition \ref{prop:bound on error for V adj}
with Lemma \ref{lem:bound on outer radi}.
\end{proof}

\section{\label{sec:Sharpness}Sharpness}

In this section we consider the $N-$particle Coulomb gas in the plane
with $V=N\phi$ for $\phi$ a function with strictly super logarithmic
growth satisfying the following regularity assumption: there exists
a neighborhood $U$ of $S_{\phi}$ such that

\begin{equation}
\phi\in C^{5}(U),\,\,\,S_{\phi}\,\text{is a domain with\,\ensuremath{C^{1}-\text{boundary,\,\,\,\ensuremath{\Delta\phi>0\text{ on\ensuremath{\,S_{\phi}}}}}}}\label{eq:regul as strong}
\end{equation}
Then it follows from \cite[Cor 1.5]{l-s} and \cite[Thm 1.1]{b-b-n-y2}
that the following expansion holds 

\begin{equation}
-\frac{1}{N^{2}\beta}\log Z_{N}[N\phi]=\mathcal{F}(\phi)-\frac{N}{2}\log N+N\left((\frac{1}{\beta}-\frac{1}{2})D_{d\lambda}(\mu_{\phi})-\xi_{\beta}\right)+o(N),\label{eq:asympt exp of serf etc}
\end{equation}
 where $\xi_{\beta}$ denotes a constant only depending on $\beta$
(by \cite[Thm 1.1]{b-b-n-y2} the expansion above holds if the boundary
of $S_{\phi}$ is only assumed to be a finite union of $C^{1}-$curves).

In particular, under the regularity assumption \ref{eq:regul as strong}
the following bound on the error $\tilde{\epsilon}_{N}(\phi)$ (formula
\ref{eq:def of error tilde in intro}) holds for $N$ large:
\[
\tilde{\epsilon}_{N}(\phi)\leq C_{\phi}/N
\]
 for a constant $C_{\phi}$ only depending on $\phi.$ Hence, Theorem
\ref{thm:sharp non-Gauss ine for strong} then implies that
\begin{equation}
\frac{1}{\beta N^{2}}\log\E(e^{-N^{2}\beta U})\leq-\mathcal{F}(\phi+u)+\mathcal{F}(\phi)+O(N^{-1})\label{eq:sharp non-gauss ineq in sect sharpness}
\end{equation}
 where the term $O(N^{-1})$ denotes a sequence of functionals $g_{N}$
of $u\in C_{c}(\C)$ (for $\phi$ and $\beta$ fixed) such that 
\[
\limsup_{N\rightarrow\infty}N^{-1}g_{N}(u)<\infty
\]
In fact, we have 
\[
\limsup_{N\rightarrow\infty}N^{-1}g_{N}(u)=-pE(\mu_{\phi+u})+C_{\phi},\,\,\,p:=\frac{2}{\beta}-1
\]

\subsection{Sharpness of the error term in the inequality in Theorem \ref{thm:sharp non-Gauss ine for strong}}

The following result says that the order $O(N^{-1})$ in the inequality
\ref{eq:sharp non-gauss ineq in sect sharpness} can not be improved,
in general:
\begin{prop}
\label{prop:sharp order of error}Let $\phi$ be a function satisfying
the regularity assumption \ref{eq:regul as strong} and assume $\beta\leq1.$
Let $g_{N}(u)$ be a functional defined for any $u\in C(\C)$ such
that $\phi+u$ has super logarithmic growth. If,

\[
\frac{1}{N^{2}}\log\E(e^{-N^{2}U})\leq-\mathcal{F}(\phi+u)+\mathcal{F}(\phi)+g_{N}(u),
\]
 then $\limsup_{N\rightarrow\infty}N^{-1}g_{N}(u)>0.$ 
\end{prop}

\begin{proof}
Assume to get a contradiction that 
\begin{equation}
\liminf_{N\rightarrow\infty}N^{-1}g_{N}(u)=0.\label{eq:contr as}
\end{equation}

\emph{Step 1: }There exists a positive constant $a$ such that 
\begin{equation}
\mu_{\phi}=a1_{S}d\lambda\label{eq:mu is contant}
\end{equation}
To prove this fix $v\in C_{c}^{\infty}(\text{int}S)$ and set $u=tv.$
Then $\phi_{t}:=\phi+u$ satisfies the regularity assumption \ref{eq:regul as strong}
for $t$ sufficiently small. Indeed, $P(\phi_{t})=P(\phi)$ for $t$
sufficiently small and hence $S_{\phi_{t}}=S_{\phi}.$ Accordingly,
the asymptotic expansion \ref{eq:asympt exp of serf etc} yields 
\begin{equation}
\frac{1}{N^{2}}\log\E(e^{-N^{2}U})+\mathcal{F}(\phi+u)+\mathcal{F}(\phi)=-N^{-1}\frac{p}{2}\left(\left(D_{d\lambda}(\mu_{\phi+u})-D_{d\lambda}(\mu_{\phi})\right)\right)+o(N^{-1})\label{eq:expansion of Laplac in terms of entropy in pf}
\end{equation}
But then the assumption \ref{eq:contr as} forces the entropy inequality
\begin{equation}
\left(D_{d\lambda}(\mu_{\phi+tv})-D_{d\lambda}(\mu_{\phi})\right)\geq0\label{eq:entropi ineq}
\end{equation}
for any $t$ sufficiently small. Differentiating wrt $t$ we conclude
that $\log(\mu_{\phi}/d\lambda)$ has to be constant, as desired.

\emph{Step 2: }The contradiction

Next, take $v$ to be a function in $C_{c}^{5}(\C)$ such that $v=\phi$
on a neighborhood of $S_{\phi}.$ Then it follows from the main result
in \cite{s-s} that $\phi+tv$ satisfies the regularity assumption
\ref{eq:regul as strong} for $t$ sufficiently small and that $S_{t}$
is contained in a fixed neighborhood of $S_{0}$ for $|t|$ sufficiently
small. Hence, as before we conclude that the entropy inequality \ref{eq:entropi ineq}
holds. But, by \ref{eq:mu is contant}, 
\[
D_{d\lambda}(\mu_{\phi+tv})=\log(a(1+t))=\log a+\log(1+t)
\]
 and hence the entropy inequality \ref{eq:entropi ineq} is violated
when then $t<0,$ which gives the desired contradiction. 
\end{proof}
In fact, we have the following more precise result showing that Theorem
\ref{thm:sharp non-Gauss ine for strong} is, in terms of energy terms,
sharp up to an error of the order $o(N^{-1}):$ 
\begin{prop}
Let $\phi$ be a function satisfying the regularity assumption \ref{eq:regul as strong}
and assume $\beta<2.$ If

\[
\frac{1}{N^{2}}\log\E(e^{-N^{2}U})\leq-\mathcal{F}(\phi+u)+\mathcal{F}(\phi)+N^{-1}\left(-A_{\phi}E_{0}(\mu_{\phi+u})+C_{\phi}\right)
\]
 then $A_{\phi}\leq p:=2/\beta-1.$ 
\end{prop}

\begin{proof}
Write $A_{\phi}=p(1+\delta)$ and assume to get a contradiction that
$\delta>0.$ Assume that $\Phi:=\phi+u$ satisfies the regularity
assumption \ref{eq:regul as strong}. Then the expansion \ref{eq:asympt exp of serf etc}
implies that
\[
-pE_{0}(\mu_{\Phi})+\frac{p}{2}D(\mu_{\Phi})\geq\delta pE_{0}(\mu_{\Phi})-C
\]
where $C'$ is a constant independent of $u.$ Since $\beta<2$ this
equivalently means that 
\begin{equation}
-E_{0}(\mu_{\Phi})+\frac{1}{2}D(\mu_{\Phi})\geq\delta E_{0}(\mu_{\Phi})-C\label{eq:energy entropy ineq in proof}
\end{equation}
It is well-known that such an inequality cannot hold for all probability
measure $\mu$ of finite energy, as can shown by letting $\mu$ converge
towards a Dirac mass, in an appropriate way. Here we will exhibit
a particular family $\mu_{\Phi_{t}}$ violating the inequality \ref{eq:energy entropy ineq in proof}
and such that $\Phi_{t}$ satisfies the regularity assumption \ref{eq:regul as strong}.
Fix $t\in\R$ and set
\[
F_{t}(z):=e^{t}z,\,\,\,\Phi_{t}:=(F_{-t})^{*}\phi
\]
 so that $\phi=\Phi_{0}.$ Since $F_{t}$ is a holomorphic automorphism
of $\C$ of degree one it follows, by symmetry, that
\[
\mu_{\Phi_{t}}=(F_{t})_{*}\mu_{\phi}
\]
Hence,
\[
E_{0}(\mu_{\Phi_{t}})=-\int\log|F_{t}(z)-F_{t}(w)|\mu_{\phi}(z)\mu_{\phi}(w)=-t+E_{0}(\mu_{\phi}),
\]
 
\[
D_{d\lambda}(\mu_{\Phi_{t}})=\int\log\frac{\mu_{\phi}}{(F_{t})^{*}d\lambda}\mu_{\phi}=-2t+D_{d\lambda}(\mu_{\Phi_{0}})
\]
But this means that the inequality \ref{eq:energy entropy ineq in proof}
is violated for $\mu_{\Phi_{t}}$ when $t\rightarrow-\infty$ and
that gives the desired contradiction
\end{proof}

\subsection{\label{subsec:The-inequality-in fails drast}The inequality in Theorem
\ref{thm:sharp non-Gauss ine for strong} fails drastically for $\beta>2$}
\begin{prop}
\label{prop:drastic failure}Let $\phi(z)=\lambda|z|^{2}$ for $\lambda>0.$
\ref{eq:regul as strong}. For $\beta>2$ there exists no strictly
positive continuous function $f(t).$ 
\[
\frac{1}{\beta N^{2}}\log\E(e^{-N^{2}\beta U})\leq-\mathcal{F}(\phi+u)+\mathcal{F}(\phi)+N^{-1}f(E(\mu_{\phi+u})),\,\,\,N>>1
\]
\end{prop}

\begin{proof}
Assume that $\Phi:=\phi+u$ satisfies the regularity assumption \ref{eq:regul as strong}.
Then the expansion \ref{eq:asympt exp of serf etc} implies that
\[
-(\frac{1}{\beta}-\frac{1}{2})(D_{d\lambda}(\mu_{\Phi})\leq C+f(E(\mu_{\Phi}))
\]
for a constant $C$ only depending on $\phi$ (but not on $u).$ When
$\beta>2$ this means that, after perhaps changing $f,$ 
\[
(D_{d\lambda}(\mu_{\Phi})\leq f(E(\mu_{\Phi}))
\]
 Note that if $\mu_{\Phi}$ is replaced by, for example, the uniform
measure $\mu_{S^{1}}$ on a circle then the inequality above fails
since the right hand side is finite while the left hand side is equal
to infinity. Thus, since $D$ is lower semi-continuous in the weak
topology it will be enough to show that there exists a sequence $u_{j}$
such that $\Phi_{j}:=\phi+u_{j}$ satisfies the regularity assumption
\ref{eq:regul as strong}, $\mu_{\Phi_{j}}\rightarrow\mu_{S^{1}}$
weakly and $E(\mu_{\Phi_{j}})\rightarrow E(\mu_{S^{1}}).$ To this
end we may by a simple scaling argument assume that $\lambda=1.$
We then set $u_{j}=-(1-1/j)\log|z|^{2}.$ Then it follows from Prop
\ref{prop:regularity} and the radial symmetry of $\Phi_{j}$ that
$\mu_{\Phi_{j}}$ is the uniform probability measure supported on
the annulus with outer and inner radius $1$ and $1-j^{-1},$ respectively.
Hence, $\mu_{\Phi_{j}}$ has the desired properties. 
\end{proof}
Most likely the previous argument applies to all $\phi$ satisfying
the regularity assumption \ref{eq:regul as strong}. But the question
what happens for $\beta\in]1,2],$ even for $\phi(z)=|z|^{2}$ appears
to be rather intriguing and is left for the future.

\section{\label{sec:Outlook-on-relations}Outlook on relations to Kähler geometry }

In this last section we briefly discuss some relations to Kähler geometry
that will be elaborated on in the companion paper \cite{berm17}.
Let $(X,L)$ be a polarized compact Riemann surface. Fix a metric
$\phi\in\mathcal{H}(L),$ i.e. a smooth metric on $L$ with strictly
positive curvature form $\omega_{\phi}$ and consider the corresponding
setting of adjoint determinantal point processes on $X.$ It follows
from well-known Bergman kernel expansions that the corresponding error
sequence $\epsilon_{N_{k},1}[\phi]$ satisfies $\epsilon_{N_{k},1}[\phi]=O(N^{-1}).$
Hence, by Theorem \ref{thm:upper bd on part funct riem surf adj},
for any $u\in H^{1}(X)$ 
\begin{equation}
\frac{1}{kN_{k}}\log\E(e^{-kN_{k}U_{N_{k}}})\leq-\mathcal{F}(\phi+u)+\mathcal{F}(\phi)+O(N^{-1}),\label{eq:ineq Laplac transf in outlook}
\end{equation}
 where the error term only depends on $\phi.$ In general, the order
of the error term can not be improved to $o(N^{-1}).$ In fact, as
will be detailed in \cite{berm17} we have the following
\begin{thm}
The order of the error in the inequality \ref{eq:ineq Laplac transf in outlook}
can be improved to $o(N^{-1})$ iff the Riemannian metric defined
by $\omega_{\phi}$ has constant scalar curvature. 
\end{thm}

The ``if direction'' is the content of \cite[Thm 1.1]{berm3} and
the ``only if direction'' is shown as follows. First, if $\phi+u\in\mathcal{H}(L),$
then 
\begin{equation}
\frac{1}{kN_{k}}\log\E(e^{-kN_{k}U_{N_{k}}})=-\mathcal{F}(\phi+u)+\mathcal{F}(\phi)-N_{k}^{-1}\left(\mathcal{M}(\phi+u)-\mathcal{M}(\phi)\right)+o(N_{k}^{-1}),\label{eq:expansion in terms of Mab in outlook}
\end{equation}
 where $\mathcal{M}$ is Mabuchi's K-energy functional on $\mathcal{H}(L),$
which has the property that $\phi$ is a critical points of $\mathcal{M}$
iff $\omega_{\phi}$ has constant scalar curvature. In particular,
if the inequality \ref{eq:ineq Laplac transf in outlook} holds with
a smaller error term of the order $o(N^{-1}),$ then the expansion
\ref{eq:expansion in terms of Mab in outlook} forces $\mathcal{M}(\phi+u)\geq\mathcal{M}(\phi)$
for any $u$ such that $\phi+u\in\mathcal{H}(L).$ Since $\mathcal{H}(L)$
is an open subset of the space of all smooth metrics on $L$ it follows
that $\phi$ is a critical point of $\mathcal{M}$ and hence $\omega_{\phi}$
defines a metric with constant scalar curvature. 

The previous argument is analogous to the proof of Step one in Prop
\ref{prop:sharp order of error}, where the role of the expansion
\ref{eq:expansion in terms of Mab in outlook} is played by the expansion
\ref{eq:expansion of Laplac in terms of entropy in pf}. Accordingly,
one is led to make the following 
\begin{conjecture}
Let $(X,L)$ be a polarized compact Riemann surface endowed. Given
a metric $\phi$ on $L$ and a function $u$ on $X$ the weak form
of the conjecture says that 
\[
\lim_{k\rightarrow\infty}\left(\frac{1}{k}\log\E(e^{-kN_{k}U_{N_{k}}})+N_{k}\left(\mathcal{F}(\phi+u)-\mathcal{F}(\phi)\right)\right)=-\mathcal{M}(P(\phi+u))+\mathcal{M}(P\phi)
\]
if $\phi$ and $u$ are smooth and $S_{\phi}$ and $S_{\phi+u}$ are
domains with smooth boundaries with $\mu_{\phi+u}$ and $\mu_{\phi}$
strictly positive on their supports. The stronger form says that the
convergence holds under the assumption that $\mu_{\phi+u}$ and $\mu_{\phi}$
have $L^{\infty}-$densities and the strongest form of the conjecture
says that the convergence holds when $\mu_{\phi+u}$ and $\mu_{\phi}$
have finite entropy. 
\end{conjecture}

The original definition of $\mathcal{M}(\psi)$ requires that $\psi$
be in $\mathcal{H}(L),$ but, as observed in \cite{berm2 komma 5}
$\mathcal{M}(\psi)$ is well-defined as soon as $dd^{c}\psi$ has
finite entropy. As will be explained in \cite{berm17} the previous
conjecture naturally extends to the higher dimensional case when $(X,L)$
is a polarized compact complex manifold, as well as to certain $\beta-$ensembles.
In the case when $X$ is the Riemann sphere $\P^{1}$ the general
conjecture for $\beta-$ensembles case is equivalent to the following
conjecture, in its strongest form:
\begin{conjecture}
\label{conj:sphere}Let $\phi$ be a function on $\C$ with logarithmic
growth such that $D_{\mu_{0}}(\mu_{\phi})<\infty,$ where $\mu_{0}:=e^{-2\psi_{0}}/\pi$
corresponds to the standard probability measure on the Riemann sphere
$\P^{1}.$ Then, setting $p:=2/\beta-1,$ 
\begin{equation}
\frac{1}{N(N+p)\beta}\log\int_{\C^{N}}\left(|D^{(N)}|^{2}e^{-(N+p)\phi}\right)^{\beta}d\lambda^{\otimes N}=-\mathcal{F}(\phi)+\frac{\log N}{2}N^{-1}-N^{-1}\left((\frac{1}{\beta}-\frac{1}{2})\mathcal{M}(P\phi)-\xi_{\beta}\right)+o(1)\label{eq:expansion in conj sphere}
\end{equation}
where $\mathcal{M}$ is Mabuchi's K-energy functional for $\mathcal{O}(1)\rightarrow\P^{1}.$ 
\end{conjecture}

In general, $\mathcal{M}(\psi)$ is only defined up to an additive
constant, but in the case when $(X,L)=(\P^{1},\mathcal{O}(1))$ the
constant can be fixed by defining 
\begin{equation}
\mathcal{M}(\psi):=F(dd^{c}\psi),\,\,\,\,F(\mu)=-2E_{0}(\mu)+D_{d\lambda}(\mu)\label{eq:Mab in terms of E and D}
\end{equation}
\begin{prop}
\label{prop:mab on sphere}Let $\phi$ be a function on $\C$ which
has strictly super logarithmic growth. Assuming that $S_{\phi}$ is
a domain with $C^{2,1}-$boundary, $\mu_{\phi}/d\lambda\in C^{2,\alpha}(S_{\phi})$
for some $\alpha\in]0,1]$ and $\mu_{\phi}/d\lambda>0$ on $S_{\phi}$,
the expansion in the previous conjecture does hold.
\end{prop}

\begin{proof}
By \cite[Cor 1.5]{l-s} the expansion \ref{eq:asympt exp of serf etc}
for $Z_{N,\beta}[N\phi]$ holds. Setting $u:=p\phi$ it will be enough
to show that 
\begin{equation}
\lim_{N\rightarrow\infty}N^{-1}\log\frac{Z_{N,\beta}[N\phi+u-\bar{u}]}{Z_{N,\beta}[N\phi]}=0,\,\,\,\bar{u}:=\int u\mu_{\phi}\label{eq:pf mab on sphere}
\end{equation}
 Indeed, accepting this for the moment we then get 
\begin{equation}
\log Z_{N,\beta}[N\phi+u]=\log Z_{N,\beta}[N\phi]-N\bar{u}+o(N)\label{eq:of mab on sphere two}
\end{equation}
 and hence the expansion \ref{eq:asympt exp of serf etc} for $Z_{N,\beta}[N\phi]$
gives 
\[
\beta^{-1}\log Z_{N,\beta}[N\phi+u]=-N^{2}\mathcal{F}(\phi)+\frac{N}{2}\log N+N\left(-\frac{p}{2}D_{d\lambda}(\mu_{\phi})-p\int\phi\mu_{\phi}+\xi_{\beta}\right)+o(N)
\]
Finally, the expansion \ref{eq:expansion in conj sphere} follows
by rewriting
\[
-N^{2}\mathcal{F}(\phi)=-N(N+p)\mathcal{F}(\phi)+Np\mathcal{F}(\phi)=-N(N+p)\mathcal{F}(\phi)+Np(E_{0}(\mu_{\phi})+\int\phi\mu_{\phi})
\]
All that remains is to prove formula \ref{eq:pf mab on sphere}. To
this end first observe that after replacing $u$ with $u-\bar{u}$
we may as well assume that $\bar{u}=0.$ Set
\[
f_{N}(t):=N^{-1}\log\frac{Z_{N,\beta}[N\phi+tu]}{Z_{N,\beta}[N\phi]},\,\,\,t\in[0,1]
\]
Then $f_{N}(0)=0$ and
\begin{equation}
\frac{df_{N}(t)}{dt}=\left\langle \E_{t}(\delta_{N}),u\right\rangle ,\label{eq:deriv of fN in pf}
\end{equation}
 where $\E_{t}$ denotes expectations of the $N-$particle Coulomb
gas with exterior potential $V_{t}:=N\phi+tu.$ It follows from essentially
well-known results that, for any fixed $t,$
\[
\lim_{N\rightarrow\infty}\left\langle \E_{t}(\delta_{N}),u\right\rangle =\left\langle \mu_{\phi},u\right\rangle =0
\]
 (see for example \cite[Cor 1.2]{berm1}; which applies even though
$u$ is not bounded using that $\phi+tu$ has strictly super logarithmic
growth). By the dominated convergence theorem all that remains is
to show that the right hand side in formula \ref{eq:deriv of fN in pf}
is uniformly bounded for $t\in[0,1].$ When $\beta=1$ this follows
form the Bergman kernel estimates in \cite{berm 1 komma 1} and the
proof in the case of a general $\beta$ is given in \cite{berm17}. 
\end{proof}

\section{\label{sec:Main-results-in}Main results in Part II and companion
papers}

Applications of the results in the present paper are deferred to the
sequel \cite{berm15} and various elaborations are given in the companion
papers \cite{berm15b,berm16,berm17}. Here we will only state the
main results in \cite{berm15} in the case of the Coulomb gas in the
plane (a discussion about relations to previous results is given in
\cite{berm15}).

We first introduce the following assumption which is stronger than
the assumption A1 in Section \ref{sec:Main-results-in}:
\begin{itemize}
\item \textbf{(A2)}\emph{ $A0$ holds and there exists a constant $C$ such
that $C^{-1}\leq\Delta\phi\leq C$ on $S$ and $\phi\in C^{2}(S)$ }
\end{itemize}
We will first show that the main inequalities can be extended to $\beta>1$
if additional error terms are included
\begin{thm}
Assume that $\phi$ and are Lipschitz continuous in a neighborhood
of $S$ and that $u$ is bounded in $\C.$ Then, for any $\beta>1$
the inequalities in Theorem \ref{thm:sharp non Gaussian bd intro}
and Theorem \ref{thm:sharp non-Gauss ine for strong} are still valid
if the following term is added to the right hand side: 
\[
N^{-1}\left\Vert du\right\Vert _{L^{\infty}(\C)}+N^{-1}C+2N^{-1}\log N,
\]
 where the constant $C$ only depends on $\phi.$ 
\end{thm}

\subsection{Large deviations for singular potentials and linear statitistics}

The following result extends the well-known LDPs for continuous exterior
potentials and linear statistics to a singular setting:
\begin{thm}
\label{thm:LDP intro}Consider the $N-$particle Coulomb gas in $\C$
at inverse temperature $\beta\leq1$ with exterior potential $V.$
Assume that $\phi$ satisfies the assumption A0 and that $u$ is a
function in $H^{1}(\C).$

\begin{itemize}
\item If $V=(N+p)(\phi+u)$ (or $V=N(\phi+u)$ if $\phi$ has strictly super
logarithmic growth) the laws of the empirical measures $\delta_{N}$
satisfy a LDP with speed $N^{2}\beta$ and rate functional $E_{\phi+u}(\mu)-\mathcal{F}(\phi+u)$
on $\mathcal{P}(\C)$
\item If $V=(N+p)\phi$ (or $V=N\phi$ if $\phi$ has strictly super logarithmic
growth) the laws of the linear statistics $U_{N}:=\left\langle u,\delta_{N}\right\rangle $
satisfy a LDP speed $N^{2}\beta$ and rate functional $I(s)-\inf_{\R}I,$
where 
\[
I(s)=\inf_{\mu\in\mathcal{P}(X)\cap H^{-1}(\C)}\left\{ E_{\phi}(\mu):\,\left\langle u,\mu\right\rangle =s\right\} ,
\]
If $u\in C_{b}(\C)$ the LDPs above hold for any $\beta>0.$ 
\end{itemize}
\end{thm}

In particular, by the first point in the previous theorem, if $\phi$
has strictly super logarithmic growth and $u\in H^{1}(\C),$ then
\begin{equation}
\lim_{N\rightarrow\infty}\frac{1}{N^{2}}\log Z_{N,\beta}[(N(\phi+u)]=-\mathcal{F}(\phi+u).\label{eq:as of Z after theorem LDP intro}
\end{equation}
for $\beta\leq1.$ 

\subsection{Moderate deviations}

The LDP in the previous section yields, in particular, a LDP for the
random variable $Y_{N}$ defined by the deviation $U_{N}-\bar{u}.$
However, the corresponding rate functional is, in general, not quadratic.
We next establish a Moderate Deviation Principle (MDP) for $Y_{N}$
with a quadratic rate functional. Recall that the general notion of
a Moderate Deviation Principle (MDP) interpolates between the notion
of a Large Deviation Principle and a Central Limit Theorem (CLT) by
introducing a sequence of positive numbers $s_{N}$ tending to zero,
that we will refer to as the \emph{deviation scale} (since it measures
the order of the deviations). Briefly, in our setting the MDP is said
to hold at the deviation scale $s_{N}$ if the random variable defined
by the \emph{scaled deviations }
\[
Y_{s_{N}}:=s_{N}^{-1}(U_{N}-\bar{u}):=s_{N}^{-1}\left\langle \delta_{N}-\mu_{\phi},u\right\rangle 
\]
 satisfies a LDP. In what follows we will use the standard notation
$a_{N}\gg b_{N}$ if $a_{N}$ and $b_{N}$ are sequences of positive
numbers such that $a_{N}/b_{N}\rightarrow\infty$ as $N\rightarrow\infty.$ 
\begin{thm}
\label{thm:MDP intro}Consider the Coulomb gas in $\C$ at inverse
temperature $\beta>0.$ Assume that A2 holds and that the support
$S$ is a topological domain (i.e. $S$ coincides with the closure
of its interior) and the complement $S^{c}$ of $S$ is regular for
the Dirichlet problem (i.e. the Dirichlet problem for the Laplacian
on $S^{c}$ preserves continuity). Given $u$ such that $\Delta u$
is bounded in a neighborhood of $S$ we have
\begin{equation}
\lim_{N\rightarrow\infty}\frac{1}{N^{2}s_{N}^{2}}\log\E(e^{-tN^{2}s_{N}^{2}Y_{s_{N}}})=\frac{\beta}{2}\sigma^{2}t^{2}/2\label{eq:as of E in thm mod intro}
\end{equation}
where $\sigma^{2}$ is the Dirichlet $H^{1}-$norm of the bounded
harmonic extension $u^{S}$ of $u$ from $S:$
\[
\sigma^{2}=\frac{1}{4\pi}\int_{X}|\nabla u^{S}|^{2},
\]
As a consequence, the random variable $Y_{S_{N}}$ satisfies a MDP
at speed $N^{2}s_{N}^{2}$ for a quadratic rate functional with variance
$\sigma^{2}.$ The same result also holds when the deviation $Y_{s_{N}}$
is replaced by the ``fluctuation''
\begin{equation}
\widetilde{Y}_{s_{N}}:=s_{N}^{-1}\left\langle \delta_{N}-\E(\delta_{N}),u\right\rangle \label{eq:def of scaled fluct intro}
\end{equation}
\end{thm}

CLTs have previously been established with the same variance $\sigma^{2}$
and are equivalent to the asymptotics \ref{eq:as of E in thm mod intro}
for $s_{N}=N,$ when $Y_{s_{N}}$ has been replaced by the fluctuation
$\widetilde{Y}_{s_{N}}$ \cite{ahm2,l-s-2,b-b-n-y2}. However, it
should be stressed that, in contrast to the MDP in the previous theorem,
the CLT does \emph{not} hold, in general, when $S$ has several components.
On the other hand the CLT is known to always hold when $u$ is supported
in the ``bulk'', i.e. in the interior of $S.$ We will prove the
MDP in the bulk, down to the CLT-scale, both globally and at mesoscopic
scales (further discussed in the next section): 
\begin{thm}
\label{thm:MDP bulk intro}Fix a deviation scale $s_{N}$ such that
$s_{N}\geq N^{-1}.$ Assume that $\beta=1,$ $\phi$ is $C^{2}-$smooth
on a neighborhood of $S$ and that $u$ is a $C^{2}-$smooth function
compactly supported in the interior of $S.$ Consider the corresponding
fluctuation $\widetilde{Y}_{s_{N}}.$ Then 
\begin{equation}
\lim_{N\rightarrow\infty}\frac{1}{N^{2}s_{N}^{2}}\log\E(e^{-tN^{2}s_{N}^{2}\widetilde{Y}_{s_{N}}})=\sigma^{2}t^{2}/2,\,\,\,\,\sigma^{2}:=\frac{1}{4\pi}\int_{S}|\nabla u|^{2}d\lambda\label{eq:asympt of log E in thm bulk MDP meso-1}
\end{equation}
As a consequence, when $s_{N}\gg N^{-1}$ the random variable $\widetilde{Y}_{s_{N}}$
satisfies a MDP at the scale $s_{N}$ and with rate functional $\sigma^{-2}t^{2}/2$
and $\widetilde{Y}_{s_{N}}$ satisfies when $s_{N}=N^{-1}$ a CLT
with variance $\sigma^{2}.$
\begin{itemize}
\item If instead $\phi$ is assumed to be $C^{4}-$smooth the same results
hold for $Y_{s_{N}}$ 
\item The same results hold in the mesoscopic setting where $u$ is replaced
by $u(z_{0}+l_{N}^{-1}(z-z_{0}))$ for any sequence $l_{N}\rightarrow0$
such that $l_{N}^{2}\gg s_{N}.$ In case when $s_{N}=N^{-1}$ the
regularity assumption on $u$ can be weakened to $u$ Lipschitz continuous.
\end{itemize}
\end{thm}

\subsection{Local laws at mesoscopic scales}

Given a positive sequence $l_{N}\rightarrow0$ of ``length scales''
and a fixed point $z_{0}$ in $\C,$ consider the corresponding ``blow-up
map'' $F_{N}$ on $\C$ defined by
\[
z\mapsto F_{N}(z)=z_{0}+l_{N}^{-1}(z-z_{0})
\]
The case when $l_{N}=N^{-b}$ for $b=1/2$ is usually called the \emph{microscopic
scale} (since it corresponds to the typical particle distances) while
the cases when $b\in]1/2,1[$ are called the \emph{mesoscopic scales.}
The induced map on measures 
\[
\mu\mapsto\mu_{z_{0}}:=l_{N}^{-2}(F_{N})_{*}\mu
\]
probes the measure $\mu$ closed to $z_{0}$ at the length scale $l_{N}.$ 
\begin{thm}
\label{thm:local law intro}Consider the Coulomb gas in $\C$ at inverse
temperature $\beta\leq1$ and the length scale $l_{N}=N^{-b}$ for
a fixed $b\in]0,1/4[.$ If A1 holds, then 

\begin{equation}
l_{N}^{-2}(F_{N})_{*}(\delta_{N}-\mu_{\phi})\rightarrow0\label{eq:conv in prob in local law intro}
\end{equation}
 exponentially in probability in the following sense: for any $u\in H^{1}(\C),$
not identically zero, there exists a positive constant $C$ such that
\begin{equation}
\P|\left(\left\langle l_{N}^{-2}(F_{N})_{*}(\delta_{N}-\mu_{\phi}),u\right\rangle |\geq\delta\right)\leq e^{-CN},\label{eq:exp conc in thm local intro}
\end{equation}
 where $C$ depends on $u,\phi$ and $\delta.$ As a consequence,
if $p_{0}$ is point in $S$ which is a Lebesgue point for $\mu_{\phi}$,
then 
\[
l_{N}^{-2}(F_{N})_{*}(\delta_{N})\rightarrow\rho(z_{0})d\lambda
\]
exponentially in probability. In particular, this is the case if $\phi$
is assumed to be $C^{2}-$smooth close to $p_{0}$ and $p_{0}$ is
in the interior of $S$ and then $\rho(z_{0})=\frac{1}{4\pi}\Delta\phi(z_{0}).$
If moreover $\Delta\phi>0$ close to $p_{0}$ then the following dichotomy
holds for points $p_{0}$ in the boundary of the complement $S^{c}:$
\begin{itemize}
\item If $\partial S^{c}$ is $C^{1}-$smooth close to $p_{0},$ then
\[
l_{N}^{-2}(F_{N})_{*}(\delta_{N})\rightarrow1_{T}\frac{1}{4\pi}\Delta\phi(z_{0})d\lambda(>0)
\]
exponentially in probability, where $T$ is a halfplane defined by
the tangent line of $\partial S$ passing through $p_{0}.$ 
\item If $\partial S^{c}$ is not $C^{1}-$smooth close to $p_{0},$ then
\[
l_{N}^{-2}(F_{N})_{*}(\delta_{N})\rightarrow0
\]
exponentially in probability. 
\end{itemize}
\end{thm}

It should be stressed that in the previous theorem no assumptions
are made on the structure of the set $S$ of as opposed to Theorem
\ref{thm:MDP intro}. The set $S$ may be extremely irregular even
if A1 holds, i.e. $\Delta\phi$ is strictly positive close to $S.$
The dichotomy above provides a probabilistic interpretation of Cafarelli's
dichotomy theorem for free boundaries: the singular points $p_{0}$
in the free boundary are precisely those points where no particles
are seen - with overwhelming probability - when zooming in at mesoscopic
scales.

In the ``bulk case'' we show that the local laws hold down to the
optimal length-scale: 
\begin{thm}
\label{thm:bulk local law intro}Assume that $\beta\leq1,$ that $\phi$
is $C^{2}-$smooth close to $p_{0}$ that $u$ is compactly supported
in the interior of $S$ and $\Delta u\in L^{\infty}.$ Suppose that
$l_{N}\gg N^{-1/2}.$ Then there exists a constant $C$ (only depending
on an upper bound on $\left\Vert \Delta u\right\Vert _{L^{\infty}})$
such that 
\[
\P\left(\left|\left\langle \left(\delta_{N}-\E(\delta_{N})\right),u_{z_{0}}\right\rangle \right|\geq\delta\right)\leq e^{-\beta(l_{N}^{2}N)^{2}\delta^{2}/C}
\]
If it is moreover assumed that $\phi$ is $C^{4}-$smooth close to
$p_{0},$ then $l_{N}^{-2}(F_{N})_{*}(\delta_{N}-\mu_{\phi})\rightarrow0$
exponentially in probability. As a consequence, 
\begin{equation}
l_{N}^{-2}(F_{N})_{*}\delta_{N}\rightarrow\frac{1}{4\pi}\Delta\phi(z_{0})d\lambda\label{eq:conv in thm bulk local law intro}
\end{equation}
exponentially in probability. 
\end{thm}

\subsection{Bergman kernels and orthogonal polynomials }

Denote by $\mathcal{H}_{N}(\C)$ the $N-$dimensional space of all
polynomials in $\C$ of degree at most $N-1.$ Given a function $\phi$
of strictly super logarithmic growth such that A0 holds we endow $\mathcal{H}_{N}(\C)$
with the Hilbert space structure defined by the weighted norm
\[
\left\Vert f\right\Vert _{N\phi}^{2}:=\int_{\C}|f|^{2}e^{-N\phi}d\lambda
\]
Denote by $\B_{N}$ the corresponding normalized Bergman measure:
\[
\B_{N}:=\frac{1}{N}\left(\sum_{i=1}^{N}|\Psi_{i}^{(N)}|^{2}\right)e^{-N\phi}d\lambda,
\]
 where $(\Psi_{i}^{(N)})_{i=1}^{N}$ is a fixed orthonormal bases
in the Hilbert space $\mathcal{H}_{N}(\C).$ In classical terms, the
density of $N\B_{N}$ is the restriction to the Christofel-Darboux
kernel for the corresponding space of weighted orthogonal polynomials.
The connection to Coulomb gases stems from the following well-known
formula:
\[
\B_{N}=\E(\delta_{N}),\,\,\,\beta=1\,\,\,\,V=N\phi
\]
As a consequence, $\B_{N}$ converges, as $N\rightarrow\infty,$ towards
the equilibrium measure $\mu_{\phi},$ in the weak topology of measures.
The following result provides a quantitative rate of convergence:
\begin{thm}
\emph{Assume that $\beta\leq1$ and that $\phi$ satisfies A1 (or
more generally, that $\mu_{\phi}$ has finite entropy). Then there
exists an explicit constant $C$ such that}
\begin{equation}
\left\Vert \E(\delta_{N})-\mu_{\phi}\right\Vert _{H^{-1}(\C)}\leq C\sqrt{\frac{\log N}{N}}\label{eq:bound on fluct in lemma-1-1-1}
\end{equation}
Moreover, if $\phi$ satisfies the regularity assumption \ref{eq:regul as strong},
then the factor $\log N$ above can be removed (if the constant $C$
is also changed).
\end{thm}

For $\phi$ with $H^{1}-$singularities we obtain the following qualitative
convergence result:
\begin{thm}
Assume that $\beta\leq1$ and that there exists a function $\phi_{0}$
satisfying A0 such that $\phi-\phi_{0}\in H^{1}(\C).$ Then 
\[
\lim_{N\rightarrow\infty}\left\Vert \E(\delta_{N})-\mu_{\phi}\right\Vert _{H^{-1}(\C)}=0.
\]
\end{thm}

Note that, in general, $\mu_{\phi}$ does not have better regularity
than $H^{-1}(\C)$ and hence the previous theorem appears to be rather
optimal.

\subsection{Concentration of measure, Monte-Carlo integration and the Gaussian
Free Field}

Let $X$ be the Riemann sphere, viewed as the one-point compactification
of $\C.$ We endow $X$ with the standard invariant metric $g$ and
denote by $\Delta_{g}$ the corresponding Laplace operator, with the
sign convention which makes it positive, viewed as a symmetric densely
defined operator on $L^{2}(X,dV_{g})$ and normalized so that on $\C\subset X$
\[
\Delta_{g}dV_{g}=-\frac{1}{4\pi}\Delta d\lambda,\,\,\,\Delta:=\partial_{x}^{2}+\partial_{y}^{2}
\]
 We will use the following notion for Sobolev spaces of fractional
order:
\begin{itemize}
\item $H^{s}(X)$ is the Sobolev space of all distributions $u$ on $X$
such that$\Delta^{s/2}u\in L^{2}(X,dV_{g}).$ The scalar product defined
by 
\[
\left\langle u,u\right\rangle _{s}:=\int_{X}\Delta^{s/2}u\Delta^{s/2}udV_{g}
\]
 induces a Hilbert space structure on $H^{s}(X)/\R.$ 
\item The dual of $H^{s}(X)/\R$ is denoted by $H_{0}^{-s}(X)$ and is endowed
with the dual Hilbert structure 
\[
\left\langle \nu,\nu\right\rangle _{-s}:=\int_{X}\Delta^{-s/2}\nu\Delta^{-s/2}\nu dV_{g}=\sup_{u\in C^{\infty}(X)}\frac{\left\langle \nu,u\right\rangle }{\left\langle u,u\right\rangle _{s}}
\]
\end{itemize}
By the Sobolev embedding theorem, for $s>1$ the space $H_{0}^{-s}(X)$
contains all Dirac masses. As a consequence, we have a map 
\[
\delta_{N}-\mu_{\phi}:\,\,\C^{N}\rightarrow H_{0}^{-s}(X).
\]
Accordingly, we can view $\delta_{N}-\mu_{\phi}$ as a $H_{0}^{-s}(X)-$valued
random variable on the $N-$particle Coulomb gas ensemble corresponding
to a given function $\phi$ and inverse temperature $\beta.$ 

Theorem \ref{thm:sharp Gauss bound on E intro} implies the following
concentration of measure inequality for the law of $\delta_{N}-\mu_{\phi}$
wrt the $H^{-s}-$norm, for $s>2$ \cite{berm15b}:
\begin{thm}
\label{thm:conc wrt Sob norm}Consider the $N-$particle Coulomb gas
in $\C$ at inverse temperature $\beta\leq1$ with exterior potential
$V:=(N+p)\phi$ for a given function $\phi$ on $\C$ of super logarithmic
growth. Given $s>2$ there exists an explicit positive constant $C$
(independent of $\phi)$ such that for any $\beta\in]0,1]$
\[
\P_{N,\beta}\left(\left\Vert \delta_{N}-\mu_{\phi}\right\Vert _{H^{-s}}^{2}>\delta\right)\leq e^{-\beta N(N+1)\left(\frac{1}{2C}\delta^{2}+\epsilon_{N,\beta}\right)+\frac{C}{(s-2)}+C},
\]
 The error sequence $\epsilon_{N,\beta}$ vanishes for the spherical
ensemble, i.e. when $\beta=1$ and $V=(N+1)\log(1+|z|^{2}).$ 
\end{thm}

As will be next explained the previous result, specialized to the
spherical ensemble, fits naturally into the setup of (Quasi-)Monte-Carlo
integration on the two-sphere $X$ endowed with its invariant measure
$dV_{g}.$ Following \cite{b-s-s-w}, given a configuration $\boldsymbol{x}_{N}\in X^{N}$
of $N$ points on $X$ the \emph{worst-case error} \emph{for the integration
rule on $X$ with node set $\boldsymbol{x}_{N}$ wrt the smoothness
parameter} $s\in]1,\infty[$ is defined by

\[
\text{wce }(\boldsymbol{x}_{N};s):=\sup_{u:\,\left\Vert u\right\Vert _{H^{s}(X)}\leq1}\left\langle u,\left(\delta_{N}(\boldsymbol{x}_{N})-dV_{g}\right)\right\rangle 
\]
(also called the \emph{generalized discrepancy }\cite{c-f}\emph{
}because of the similarity with the Koksma-Hlawka inequality for numerical
integration on Euclidean cubes). In other words, 
\[
\text{wce }(\boldsymbol{x}_{N};s)=\left\Vert \delta_{N}(\boldsymbol{x}_{N})-dV_{g}\right\Vert _{H^{-s}}
\]
We will say that a sequence $\boldsymbol{x}_{N}\in X^{N}$ is of\emph{
convergence order $\mathcal{O}(N^{-\kappa})$ wrt the smoothness parameter
$s$} if 
\[
\text{wce }(\boldsymbol{x}_{N};s)\leq\frac{C_{\kappa,s}}{N^{\kappa}}
\]
By the previous theorem a random sequence $\boldsymbol{x}_{N}$ in
the spherical ensemble is, with high probability, of convergence order
$\mathcal{O}(N^{-1})$ wrt any smoothness parameter $s>2.$

As shown in \cite{b-s-s-w} the optimal convergence order wrt a smoothness
parameter $s\in]1,\infty[$ is $s/2,$ i.e. there exists a constant
$A(s)$ such that for any sequence $\boldsymbol{x}_{N}\in X^{N}$
\begin{equation}
\text{wce }(\boldsymbol{x}_{N};s)\geq A(s)/N^{s/2},\text{\,}\label{eq:wce for speci}
\end{equation}
and the bound is saturated for a sequence of \emph{spherical designs}
\cite{b-s-s-w}. Thus, by the previous corollary, a random sequence
$\boldsymbol{x}_{N}$ in the spherical ensemble has - with high probability
- nearly optimal convergence order for smoothness parameters $s$
close to $2.$ Note that the lower bound \ref{eq:wce for speci} implies
that Theorem \ref{thm:conc wrt Sob norm} does not hold when $s\in]1,2[,$
in the case of the spherical ensemble. On the other hand, combining
Theorem \ref{thm:conc wrt Sob norm} with \cite[Lemma 26]{b-s-s-w},
which says that, if $s_{0}<s_{1}$ and $\text{wce }(\boldsymbol{x}_{N};s_{1})<1$
then
\[
\text{wce }(\boldsymbol{x}_{N};s_{0})\leq C_{s_{0},s_{1}}\text{(wce }(\boldsymbol{x}_{N};s_{1})^{s_{0}/s_{1}},
\]
reveals that the concentration inequality in Theorem \ref{thm:conc wrt Sob norm}
generalizes to $s\in]1,2[$ if $\delta^{2}$ is replaced by $\delta^{q(s)}$
for an exponent $q(s)$ which can be taken arbitrarily close to $s.$
More precisely, exploiting the explicit form of the constants in Theorem
\ref{thm:conc wrt Sob norm} gives the following
\begin{cor}
Consider the spherical ensemble with $N$ particles and fix $s\in]1,2].$
Then there exists a constant $C$ such that for any $R$ in the inverval
$C^{1/2}\leq R\leq N/(\log N)^{1/2}$ 
\[
\P\left(\text{wce }(\boldsymbol{x}_{N};H^{s}(X))\leq R^{s/2}\frac{(\log N)^{s/4}}{N^{s/2}}\right)\geq1-\frac{1}{N^{R^{2}/C}}
\]
\end{cor}

As a consequence, a random sequence $\boldsymbol{x}_{N}$ in the spherical
ensemble has, with very high probability, nearly optimal convergence
order for any smoothness parameter $s\in]1,2].$ 

\subsubsection{The Laplacian of the Gaussian free field}

The concentration inequality in Theorem \ref{thm:conc wrt Sob norm}
is deduced from Theorem \ref{thm:sharp Gauss bound on E intro}, using
some Gaussian measure theory. Briefly, the point is that for any $s>2$
there exists a unique Gaussian measure $\gamma$ on $H_{0}^{-s}(X)$
whose Cameron-Martin Hilbert space is given by the intrinsic Hilbert
space $H_{0}^{-1}(X),$ viewed as a dense subspace of $H_{0}^{-s}(X).$
In other words, $\gamma$ is the unique Gaussian measure on $H_{0}^{-s}(X)$
with the following property: if $u$ and $v$ are in $C^{\infty}(X),$
the covariance $C(u,v)$ of the corresponding linear functionals on
$H^{-s}(X)$ is given by
\[
C(u,v):=\int_{H^{-s}(X)}\left\langle u,\nu\right\rangle \left\langle v,\nu\right\rangle \gamma(\nu)=\left\langle u,v\right\rangle _{H^{1}}
\]
The existence of $\gamma$ is essentially well-known. Indeed, denoting
by $G$ the Gaussian random variable defined by a random element in
$(H_{0}^{-s}(X),\gamma)$ we have 
\[
G=\frac{1}{4\pi}\Delta u
\]
 where $u$ is the Gaussian Free Field on the Riemann sphere $X,$
viewed as a random element in $H^{-(s-2)}(X)/\R$ (see \cite{she}). 

Now, Theorem \ref{thm:sharp Gauss bound on E intro} implies the following
inequality for the moment generating function of the $H_{0}^{-s}(X)-$valued
random variable 
\[
Y_{N}:=N(\delta_{N}-dV_{g})
\]
 defined wrt the corresponding Coulomb gas ensemble with $N-$particles:

\begin{equation}
\E(e^{\left\langle Y_{N},\cdot\right\rangle })\leq e^{\beta N(N+1)\epsilon_{N,\beta}}\E(e^{\left\langle G,\cdot\right\rangle })\label{eq:sub-gauss wrt gff}
\end{equation}
viewed as functions on $H^{s}(X)/\R,$ identified with the topological
dual of $H_{0}^{-s}(X).$ Using some Gaussian measure theory this
implies Theorem \ref{thm:conc wrt Sob norm}.
\begin{rem}
For the spherical ensemble, which satisfies $\epsilon_{N,\beta}=0,$
the inequality \ref{eq:sub-gauss wrt gff} says that the random variable
$Y_{N}$ taking values in the space $(H_{0}^{-s}(X),\gamma)$ is \emph{sub-Gaussian. }
\end{rem}

\subsection{Sharp deviation inequality for the 2D Coulomb restricted to $\R$}

The 2D Coulomb gas restricted to $\R,$ subject to the exterior potential
$V,$ is defined by i.e. the ensemble $(\R^{N},d\P_{N,\beta})$ obtained
by replacing the Lebesgue measure $d\lambda$ on $\C$ in formula\ref{eq:Gibbs measure as det intro}
with the Lebesgue measure $dx$ on $\R.$ In the case when $V=Nx^{2}/2$
and $\beta=1$ this ensemble is known as the\emph{ Gaussian Unitary
Ensemble (GUE)} since it coincides with the eigenvalue law of a random
Hermitian rank $N$ matrix with independent Gaussian entries of variance
$1/N$ \cite{fo}.

Now fix a lsc function $\phi$ on $\R$ of super logarithmic growth
which is continuous on the complement of a closed polar set. The corresponding
free energy functional $\mathcal{F}_{\R}(\phi)$ is defined as in
formula \ref{eq:def of free energ intro} with $\mathcal{P}(\C)$
replaced by $\mathcal{P}(\R).$ Set

\[
p:=1/\beta-1
\]
(which differs from the previous notation used in the setting of $\C)$
and define the \emph{error sequence}
\[
\epsilon_{N,\beta}[\phi]:=-\frac{1}{\beta}\frac{1}{N(N+p)}\log Z_{N,\beta}[(N+p)\phi]_{\R}+\mathcal{F}_{\R}(\phi)+\log N!
\]
It follows from essentially well-known results that $\epsilon_{N,\beta}[\phi]=o(1)$
as $N\rightarrow\infty.$

In this one-dimensional setting the following analogs of Theorem \ref{thm:sharp non Gaussian bd intro}
and Theorem \ref{thm:sharp Gauss bound on E intro} are established
in \cite{berm16}:
\begin{thm}
\emph{Consider the $N-$particle 2D Coulomb gas restricted to $\R.$
For any $\beta\leq1$ the following inequalities hold: 
\[
\frac{1}{\beta}\frac{1}{N(N+p)}\log\E(e^{-\beta(N+p)NU_{N}})\leq-\mathcal{F}_{\R}(\phi+u)+\mathcal{F}_{\R}(\phi)+\epsilon_{N,\beta}[\phi]
\]
and 
\[
\frac{1}{\beta}\frac{1}{N(N+p)}\log\E(e^{-\beta(N+p)N(U_{N}-\bar{u})})\leq\frac{1}{8\pi}\int_{\C}|\nabla u^{\R}|^{2}d\lambda+\epsilon_{N,\beta}[\phi],
\]
 where $u^{\R}$ denotes the bounded harmonic extension of $u$ to
$\C.$ }
\end{thm}

Moreover, it is shown in \cite{berm16} that the corresponding inequalities
drastically fail when $\beta>1$ (in the sense of Prop \ref{prop:drastic failure}).
This stems from the fact that in the one-dimensional setting the constant
$(\frac{1}{\beta}-\frac{1}{2})$ in front of the entropy term in formula
\ref{eq:asympt exp of serf etc} is replaced by $(\frac{1}{\beta}-1)$
(see \cite[Cor 1.5]{l-s} ).

\end{document}